\documentclass[10pt,a4paper]{amsart}

\usepackage{amsthm, amsfonts,amssymb,euscript}
\usepackage{amsmath}
\usepackage{amssymb}
\usepackage{amsthm}
\usepackage{graphicx}
\usepackage{amsfonts}
\usepackage{hyperref}
\usepackage{float}

\numberwithin{equation}{section}

\usepackage{color}
\def\f12{\frac 1 2}

\def\a{\alpha}

\def\f12{\frac 1 2}

\newcommand{\nabb}{\mbox{$\nabla \mkern-13mu /$\,}}
\newcommand{\lapp}{\mbox{$\triangle \mkern-13mu /$\,}}

\newtheorem{remark}{Remark}[section]
\newtheorem{lemma}{Lemma}[subsection]
\newtheorem{theorem}{Theorem}[section]
\newtheorem{proposition}{Proposition}[subsection]

\newtheorem{corollary}{Corollary}[section]

\newtheorem{mytheo}{Theorem}

\begin{document}

\title[Linear Stability and Instability Of Extreme Reissner-Nordstr\"{o}m I]{Stability and Instability Of Extreme Reissner-Nordstr\"{o}m Black Hole Spacetimes for Linear Scalar Perturbations I}
\author[Stefanos Aretakis]{Stefanos Aretakis$^*$}\thanks{$^*$University of Cambridge,
Department of Pure Mathematics and Mathematical Statistics,
Wilberforce Road, Cambridge, CB3 0WB, United Kingdom}

\date{October 18, 2010}
\maketitle

\begin{abstract}
We study the problem of stability and instability of extreme Reissner-Nordstr\"{o}m spacetimes  for  linear scalar perturbations. Specifically, we consider solutions to the linear wave equation $\Box_{g}\psi=0$ on a suitable globally hyperbolic subset of such a  spacetime, arising from regular initial data prescribed on a Cauchy hypersurface $\Sigma_{0}$ crossing the future event horizon $\mathcal{H}^{+}$.  We obtain boundedness, decay and non-decay results. Our estimates hold up to and including the horizon $\mathcal{H}^{+}$. The fundamental new aspect of this problem is the degeneracy of the redshift on $\mathcal{H}^{+}$. Several new analytical features of degenerate horizons are also presented.
\end{abstract}

\section{Introduction}
\label{sec:Introduction}

\textit{Black holes} are one of the most celebrated predictions of General Relativity and one of the most intriguing objects of Mathematical Physics. A particularly interesting but  peculiar example of a  black hole spacetime is given by the so-called  \textit{extreme Reissner-Nordstr\"{o}m} metric, which in local coordinates $(t,r,\theta,\phi)$ takes the form 
\begin{equation}
g=-Ddt^{2}+\frac{1}{D}dr^{2}+r^{2}g_{\scriptstyle\mathbb{S}^{2}},
\label{1g1}
\end{equation}
where 
\begin{equation*}
D=D\left(r\right)=\left(1-\frac{M}{r}\right)^{2},
\end{equation*}
$g_{\scriptstyle\mathbb{S}^{2}}$ is the standard metric on $\mathbb{S}^{2}$ and $M>0$.   This spacetime has been the object of considerable study in the physics literature (see for instance the  recent \cite{marolf}).   In this series of papers, we investigate the linear stability and instability of extreme Reissner-Nordstr\"{o}m for scalar perturbations, that is to say we shall attempt a more or less complete treatment of  the wave equation
\begin{equation}
\Box_{g}\psi=0
\label{1eq}
\end{equation}
on  extreme Reissner-Nordstr\"{o}m exterior backgrounds.   This resolves Open Problem 4 (for  extreme Reissner-Nordstr\"{o}m) from Section 8 of \cite{md}. The fundamentally new aspect of this problem is the degeneracy of the redshift on the event horizon $\mathcal{H}^{+}$. Several new analytical features of degenerate event horizons are also presented.

 The main results (see Section \ref{sec:TheMainTheorems}) of the present paper include:
\begin{enumerate}
 \item Local integrated decay of energy, up to and including the event horizon $\mathcal{H}^{+}$ (Theorem \ref{th1}).
	\item Energy and pointwise uniform boundedness of solutions, up to and including  $\mathcal{H}^{+}$ (Theorems \ref{t2}, \ref{t5}).
	\item Sharp second order $L^{2}$ estimates, up to and including $\mathcal{H}^{+}$ (Theorem \ref{theorem3}).
	\item Non-Decay of higher order translation invariant quantities along $\mathcal{H}^{+}$ (Theorem \ref{nondecay}).
\end{enumerate}
Note that the last result is a statement of \textit{instability}. In the companion paper \cite{aretakis2}, where we shall provide the complete picture of the linear stability and instability of extreme Reissner-Nordstr\"{o}m spacetimes, we obtain energy and pointwise decay of solutions up to and including $\mathcal{H}^{+}$ and non-decay and blow-up results for higher order derivatives of solutions along $\mathcal{H}^{+}$. Note that the latter blow-up estimates are in sharp contrast with the non-extreme case, for which decay holds for all higher order derivatives of $\psi$ along $\mathcal{H}^{+}$. 

\subsection{Preliminaries}
\label{sec:Preliminaries}

Before we discuss in detail our results, let us present the distinguishing properties  of extreme Reissner-Nordstr\"{o}m and put this spacetime in the context of previous results.

\subsubsection{Extreme Black Holes}
\label{sec:ExtremeBlackHoles}

We  briefly describe here  the geometry of the horizon of extreme Reissner-Nordstr\"{o}m. (For a nice introduction to the relevant notions, we refer the reader to \cite{haw}). The event horizon $\mathcal{H}^{+}$ corresponds to $r=M$, where the $(t,r)$ coordinate system \eqref{1g1} breaks down. The coordinate vector field $\partial_{t}$, however, extends to a regular null Killing vector field $T$ on $\mathcal{H}^{+}$. The integral curves of $T$ on $\mathcal{H}^{+}$ are in fact affinely parametrised:
\begin{equation}
\nabla_{T}T=0.
\label{1t1}
\end{equation}
More generally, if an event horizon admits a Killing  tangent vector field $T$ for which \eqref{1t1} holds, then the horizon is called degenerate and the black hole \textit{extreme}. In other words, a black hole is called extreme if the surface gravity vanishes on $\mathcal{H}^{+}$(see Section \ref{sec:RedshiftEffectAndSurfaceGravityOfH}). Under suitable circumstances, the notion of extreme black holes can in fact be defined even  in case the spacetime does not admit a Killing vector field (see \cite{price}). There has also been rapid progress as regards the problem of uniqueness of extreme black hole spacetimes. We refer the reader to \cite{chrus} for the classification of static electro-vacuum solutions with degenerate components. 

The extreme  Reissner-Nordstr\"{o}m corresponds to the $M=e$ subfamily of the two parameter  Reissner-Nordstr\"{o}m family with parameters  mass $M>0$ and  charge $e>0$. It  sits between the non-extreme black hole case $e<M$ and the so-called naked singularity case $M<e$. Note that the physical relevance of the black hole notion rests in the expectation  that black holes are ``stable'' objects in the context of the dynamics of the Cauchy problem for the Einstein equations. On the other hand, the so-called weak cosmic censorship conjecture suggests that naked singularities are dynamically unstable (see the discussion in \cite{neo}). That is to say, \textit{extreme black holes are expected to have both stable and unstable properties}; this makes their analysis very interesting and challenging.

\subsubsection{Linear Scalar Perturbations}
\label{sec:LinearWaves}

 The first step in understanding the non-linear  stability (or instability) of a spacetime  is by considering the wave equation \eqref{1eq} (we refer the reader to \cite{christab} for the details of the proof of the stability of Minkowski space). This is precisely the motivation of the present paper. Indeed, to show stability one would have  to prove that solutions of the wave equation decay sufficiently fast. For potential future applications all methods should ideally be robust and the resulting estimates   quantitative. Robust means that the methods still apply when the background metric is replaced by a nearby metric, and quantitative means that any estimate for $\psi$ must be in terms of uniform constants and (weighted Sobolev) norms of the initial data.  Note also that it is essential to obtain non-degenerate estimates for $\psi$ on $\mathcal{H}^{+}$ and to consider initial data that  do not vanish on the horizon. As we shall see,  the issues at the horizon  turn out to be the most challenging part in understanding the evolution of waves on  extreme Reissner-Nordstr\"{o}m.

\subsubsection{Previous Results for Waves on Non-Extreme Black Holes}
\label{sec:PreviousResultsForNonExtremeBlackHoles}

The wave equation \eqref{1eq} on black hole spacetimes  has   been studied for a long time beginning with the pioneering work of Regge and Wheeler \cite{RW} for Schwarzschild. Subsequently, a series of heuristic and numerical arguments were put forth for obtaining decay results for $\psi$ (see \cite{other2, price72}). However, the first complete quantitative result (uniform boundedness) was obtained only in 1989 by Kay and Wald \cite{wa1}, extending the restricted result of  \cite{drimos}. Note that the proof of \cite{wa1} heavily depends on the exact symmetries of the Schwarzschild spacetime.

During the last decade, the wave equation on black hole spacetimes  has become a very active area in mathematical physics. As regards the Schwarzschild spacetime,   ``$X$ estimates" providing local integrated energy decay (see Section \ref{sec:MorawetzAndXEstimates} below) were derived   in \cite{blu0,blu3,dr3}. Note that \cite{dr3} introduced a vector field estimate which captures in a stable manner the so-called  \textit{redshift effect}, which allowed the authors to obtain quantitative pointwise estimates on the horizon $\mathcal{H}^{+}$.  Refinements for Schwarzschild were achieved in \cite{dr5} and \cite{tataru1}. Similar estimates to \cite{blu0} were derived in \cite{blu1} for the whole parameter range of Reissner-Nordstr\"{o}m including the extreme case. However, these estimates  degenerate on $\mathcal{H}^{+}$ and require  the initial data to be supported away from $\mathcal{H}^{+}$.  
 
The first boundedness result for solutions of the wave equation on slowly rotating Kerr ($\left|a\right|\ll M$) spacetimes was proved in \cite{dr7} and decay results  were derived in \cite{blukerr,md,tataru2}. Decay results for general subextreme Kerr  spacetimes ($\left|a\right|<M$) are proven in \cite{megalaa}. Two new methods were presented recently for obtaining sharp decay of energy flux and pointwise decay on black hole spacetimes; see \cite{new,tataru3}. For results on the coupled wave equation see \cite{price}. For other results see \cite{other1,finster1,kro}. For an exhaustive list of references, see \cite{md}.

Note that all previous arguments for obtaining boundedness and  decay results on non-extreme black hole  spacetimes near the horizon would break down in our case (see Sections \ref{sec:RedshiftEffectAndSurfaceGravityOfH}, \ref{sec:TheSpacetimeTermKN}). The reason for this is precisely the degeneracy of the redshift on $\mathcal{H}^{+}$.

\subsection{Overview of Results and Techniques}
\label{sec:OverviewOfMethodsAndResults}

We use the robust vector field method (see Section \ref{sec:TheVectorFieldMethod}). Our methods at various points rely on the spherical symmetry but not on other unstable properties of the spacetimes (such as the complete integrability of the geodesic flow, the separability of the wave equation, etc.)

\subsubsection{Zeroth Order Morawetz and $X$ Estimates}
\label{sec:MorawetzAndXEstimates}

\noindent Our analysis begins with local $L^{2}$ spacetime estimates. We refer to local spacetime estimates controlling the derivatives of $\psi$  as ``$X$ estimates" and $\psi$ itself  as ``zeroth order Morawetz estimate". Both these types of estimates have a long history (see \cite{md}) beginning with the seminal work of Morawetz \cite{mor2} for the wave equation on Minkowski spacetime. They arise from the spacetime term of energy currents $J^{X}_{\mu}$ associated to a vector field $X$ (see Section \ref{sec:TheVectorFieldMethod} for the definition of energy currents). For Schwarzschild,  such estimates appeared in \cite{blu0,blu3,blu2,dr3,dr5}  and for  Reissner-Nordstr\"{o}m in \cite{blu1}. The biggest difficulty in deriving an $X$ estimate for black hole  spacetimes has to do with the \textit{trapping effect}. Indeed, from a continuity argument one can infer the existence of null geodesics which neither cross $\mathcal{H}^{+}$ nor terminate at $\mathcal{I}^{+}$. In our case, a class of such geodesics lie  in a hypersurface of constant radius (see Section \ref{sec:PhotonSphereAndTrappingEffect}) known as the \textit{photon sphere}. From the analytical point of view, trapping affects the derivatives tangential to the photon sphere  and any non-degenerate spacetime estimate must lose (tangential) derivatives (i.e.~must require high regularity for $\psi$).
 
In this paper, we first (making minimal use of the spherical decomposition)  derive a zeroth order Morawetz estimate for $\psi$ which does not degenerate at the photon sphere.  For the case $l\geq 1$ (where $l$ is related to the eigenvalues of the spherical Laplacian, see Section \ref{sec:EllipticTheoryOnMathbbS2}) we use a suitable current which captures the trapping in the extreme case.  As regards the zeroth spherical harmonics ($l=0$), we present a method which is 
robust and uses only geometric properties of the domain of outer communications. Our argument (for $l=0$) applies for a wider class of black hole spacetimes and, in particular, it applies for Schwarzschild. Note that no unphysical conditions are imposed on the initial data  which, in particular, are not required to be compactly supported or supported away from $\mathcal{H}^{+}$. Once this Morawetz estimate is established, we then show how to derive a degenerate (at the photon sphere) X estimate which does not require higher regularity and  a non-degenerate X estimate (for which we need, however, to commute with the  Killing vector field $T$). These estimates, however, degenerate on $\mathcal{H}^{+}$; this degeneracy will be eliminated later (see Section \ref{sec:CommutingWithATransvrersalToMathcalHVectorField}).  See Theorem \ref{th1} of Section \ref{sec:TheMainTheorems}.

\subsubsection{Uniform Boundedness of Non-Degenerate Energy}
\label{sec:UniformBoundednessOfEnergy}

The vector field $T=\partial_{t}$ is causal and Killing and  the energy flux of the current $J_{\mu}^{T}$ is non-negative definite (and bounded) but degenerates on the horizon (see Section \ref{sec:TheVectorFieldTextbfM}). Moreover, in view of the lack of redshift along $\mathcal{H}^{+}$, the divergence of the energy current $J_{\mu}^{N}$ associated to the redshift vector field $N$, first introduced in \cite{dr3}, is not positive definite near $\mathcal{H}^{+}$ (see Section \ref{sec:TheVectorFieldN}). For this reason \textit{we appropriately modify $J_{\mu}^{N}$ so the new bulk term is non-negative definite near $\mathcal{H}^{+}$}. Note  that  the bulk term is not positive far away from $\mathcal{H}^{+}$ and so to control these terms we use the $X$ and zeroth order Morawetz estimates. The arising boundary terms can be bounded using  Hardy-like inequalities. It is important here to mention that \textit{a Hardy inequality (in the first form presented in Section \ref{sec:HardyInequalities}) allows us to bound the  local $L^{2}$ norm of $\psi$ on hypersurfaces crossing $\mathcal{H}^{+}$ using the (conserved) degenerate energy of $T$.}  See Theorem \ref{t2} of Section \ref{sec:TheMainTheorems}.

\subsubsection{Non-Decay along $\mathcal{H}^{+}$}
\label{sec:NonDecayAlongMathcalH}

We next show that the degeneracy of redshift  gives rise to a conservation law along $\mathcal{H}^{+}$ for the zeroth spherical harmonics. This result implies that \textit{for generic waves a higher order translation invariant geometric quantity does not decay along $\mathcal{H}^{+}$} (see Theorem \ref{nondecay} of Section \ref{sec:TheMainTheorems}). This low frequency obstruction will be crucial for obtaining the definitive statement of instability.

\subsubsection{Local Integrated Energy Decay}
\label{sec:CommutingWithATransvrersalToMathcalHVectorField}

We next return to the problem of retrieving the derivative transversal to $\mathcal{H}^{+}$  in the $X$ estimate in a neighbourhood of $\mathcal{H}^{+}$;  see Theorem \ref{theorem3} of Section \ref{sec:TheMainTheorems}. We first show that  on top of the above low frequency obstruction comes another new feature of degenerate event horizons. Indeed, \textit{to obtain non-degenerate spacetime estimates, one needs to require higher regularity for $\psi$ and commute  with the vector field transversal to $\mathcal{H}^{+}$. This shows that $\mathcal{H}^{+}$ exhibits phenomena characteristic of trapping} (see also the discussion in Section \ref{sec:TrappingEffectOnMathcalH}). Then by using  appropriate modifications of the redshift current and Hardy inequalities along $\mathcal{H}^{+}$ we obtain the sharpest possible result. See Section \ref{sec:CommutingWithAVectorFieldTransversalToMathcalH}. Note that although (an appropriate modification of) the redshift current can be used as a multiplier for all angular frequencies, the redshift vector field $N$ can only be used as a commutator\footnote{The redshift vector field was used as a commutator for the first time in \cite{dr7}.} for $\psi$ supported on the frequencies $l\geq 1$. These results will be further investigated in \cite{aretakis2}.

\subsubsection{Pointwise Uniform Boundedness}
\label{sec:PointwiseUniformBoundedness}

Using the above higher order energy estimates and appropriate Sobolev inequalities we finally obtain uniform pointwise boundedness of solutions up to and including $\mathcal{H}^{+}$ (see Theorem \ref{t5} of Section \ref{sec:TheMainTheorems}). We note that in \cite{aretakis2} we show that the argument of Kay and Wald \cite{wa1} cannot be applied in the extreme case for obtaining similar boundedness results, i.e. for generic $\psi$, there does not exist a Cauchy hypersurface $\Sigma$ crossing $\mathcal{H}^{+}$ and a solution $\tilde{\psi}$ such that \begin{equation*}
T\tilde{\psi}=\psi
\end{equation*}
in the causal future of $\Sigma$.

\subsection{Remarks on the Analysis of Extreme Black Holes}
\label{sec:EventHorizonVsPhotonSphere}

We conclude this introductory section by discussing several new features of degenerate event horizons. 

\subsubsection{Dispersion vs Redshift}
\label{sec:DispersionVsRedshift}

In \cite{md}, it was shown that for a wide variety of non-extreme black holes, the redshift on $\mathcal{H}^{+}$ suffices to yield uniform boundedness (up to and including the horizon) of  waves $\psi$ without any need of understanding the dispersion properties of $\psi$. However, in the extreme case, the degeneracy of the redshift makes the understanding of the dispersion of $\psi$ essential even for the problem of  boundedness. In particular, one has to derive spacetime integral estimates for $\psi$ and its derivatives. Moreover, we show that dispersion can be completely decoupled from the redshift effect.

\subsubsection{Trapping Effect on $\mathcal{H}^{+}$}
\label{sec:TrappingEffectOnMathcalH}

According to the results of Sections \ref{sec:TheCurrentJMuNDeltaFrac12AndEstimatesForItsDivergence} and \ref{sec:CommutingWithAVectorFieldTransversalToMathcalH}, in order to obtain $L^{2}$ estimates in neighbourhoods of $\mathcal{H}^{+}$ one must require higher regularity for $\psi$ and commute with the vector field transversal to $\mathcal{H}^{+}$ (which is not Killing). This loss of a derivative is characteristic  of trapping. Geometrically, this is related to the fact that the null generators of $\mathcal{H}^{+}$ viewed as integral curves of the Killing vector field $T$ are affinely parametrised. 

The trapping properties of the photon sphere have a different analytical flavour. Indeed, in order to obtain $L^{2}$ estimates in regions which include the photon sphere, one needs to commute with  either $T$ or the generators of the Lie algebra so(3) (note that all these vector fields are Killing). Only the high angular frequencies are trapped on the photon sphere (and for the low frequencies no commutation is required) while all the angular frequencies are trapped (in the above sense) on $\mathcal{H}^{+}$.

\section{Geometry of  Extreme Reissner-Nordstr\"{o}m  Spacetime}
\label{sec:GeometryOfExtremeReissnerNordstromSpacetime}

The unique family of spherically symmetric  asymptotically flat solutions of the coupled Einstein-Maxwell equations is the two parameter  \textit{Reissner-Nordstr\"{o}m} family of 4-dimensional Lorentzian manifolds $\left(\mathcal{N}_{M,e},g_{M,e}\right)$ where the parameters $M$ and $e$ are called mass and (electromagnetic) charge, respectively. 

The Reissner-Nordstr\"{o}m  metric  was first written in local coordinates $\left(t,r,\theta,\phi\right)$ in 1916 \cite{r} and 1918 \cite{n}. In these coordinates,  $g=g_{M,e}=-Ddt^{2}+\frac{1}{D}dr^{2}+r^{2}g_{\scriptstyle\mathbb{S}^{2}}$, where $D=D\left(r\right)=1-\frac{2M}{r}+\frac{e^{2}}{r^{2}}$
and $g_{\scriptstyle\mathbb{S}^{2}}$ is the standard metric on $\mathbb{S}^{2}$. The extreme case corresponds to $M=\left|e\right|$. From now on we restrict our attention to the extreme case.

Clearly,  SO(3) acts by isometry on these spacetimes. We will refer to the SO(3)-orbits as (symmetry) spheres. The coordinate $r$ is defined intrinsically such that the area of the spheres of symmetry is $4\pi r^{2}$ (and thus should be thought of as a purely geometric function of the spacetime).  

In view of the coordinate singularity  at $r=M$ (note that $M$ is the  double root of $D$ in the extreme case), we introduce  the so-called \textit{tortoise} coordinate $r^{*}$ given by $\frac{\partial r^{*}\left(r\right)}{\partial r}=\frac{1}{D}.$ Note that in the extreme case $r^{*}$ is inverse linear (instead of logarithmic in the non-extreme case). The metric with respect to the system $\left(t,r^{*}\right)$ then becomes $g=-Ddt^{2}+D\left(dr^{*}\right)^{2}+r^{2}g_{\scriptstyle\mathbb{S}^{2}}.$ 

The way to extend the metric beyond $r=M$ is by considering the \textit{ingoing Eddington-Finkelstein coordinates} $\left(v, r\right)$ where $v=t+r^{*}.$ In these coordinates the metric is given by
 \begin{equation}
g=-Ddv^{2}+2dvdr+r^{2}g_{\scriptstyle\mathbb{S}^{2}},
\label{RN1}
\end{equation}
where 
$D=D\left(r\right)=\left(1-\frac{M}{r}\right)^{2}$ and $g_{\scriptstyle\mathbb{S}^{2}}$ is the standard metric on $\mathbb{S}^{2}$. The radial curves  $v=c$, where c is a constant, are the ingoing radial null geodesics. This means that the null coordinate vector field $\partial_{r}$ differentiates with respect to $r$ on these null hypersurfaces. This geometric property of $\partial_{r}$ makes this vector field very useful for understanding the behaviour of solutions to the wave equation close to $\mathcal{H}^{+}$. 

The Penrose diagram  of the spacetime $\mathcal{N}$ covered by this coordinate system for $v\in\mathbb{R},r\in\mathbb{R}^{+}$ is
 \begin{figure}[H]
	\centering
		\includegraphics[scale=0.135]{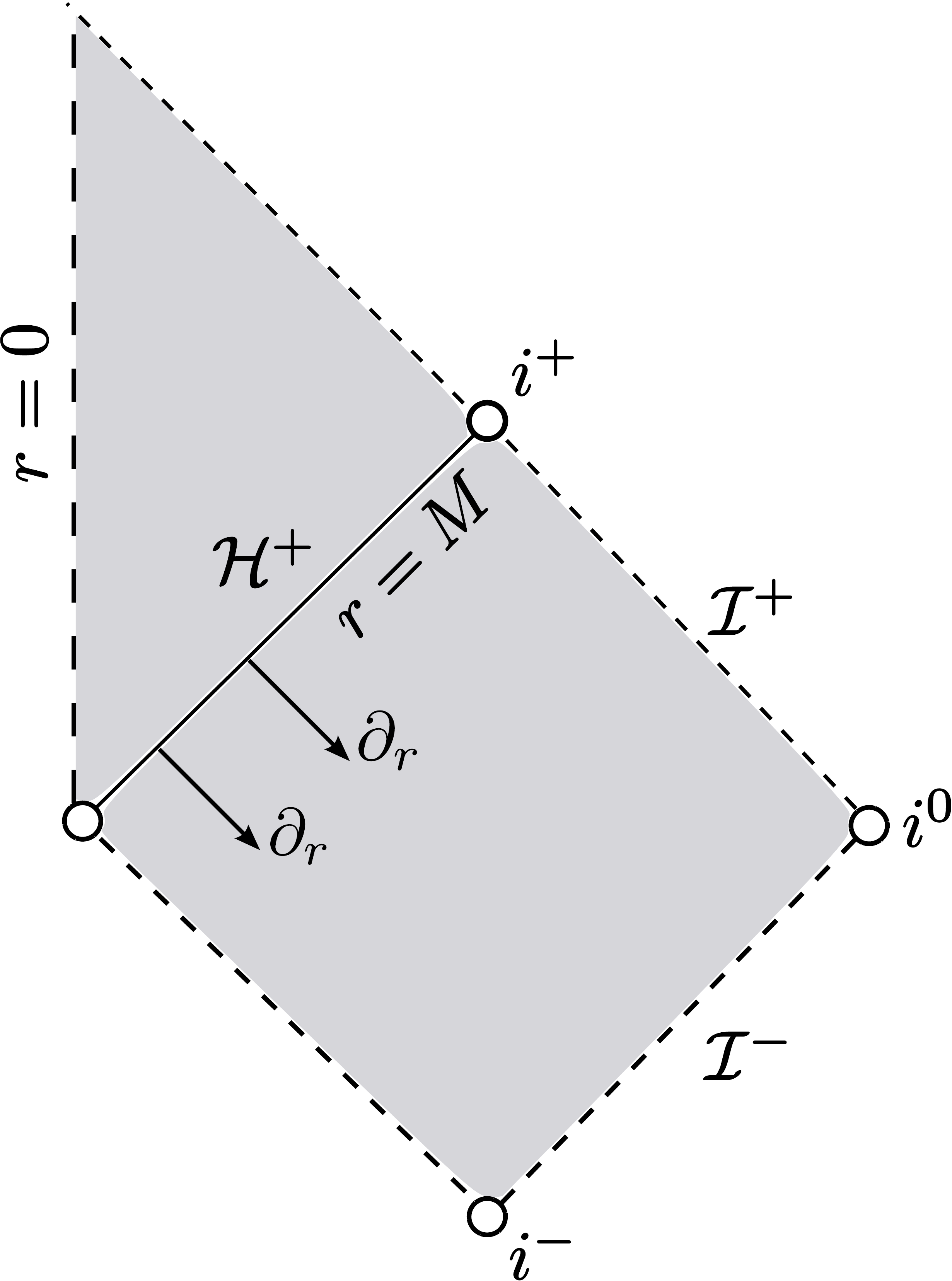}
	\label{fig:ern001}
\end{figure}
We will refer to the hypersurface  $r=M$ as the event horizon (and denote it by $\mathcal{H}^{+}$) and the region $r\leq M$ as the black hole region. The region where $M<r$ corresponds to  the domain of outer communications.

In view of the existence of the timelike curvature singularity $\left\{r=0\right\}$ `inside' the black hole (thought of here as a singular boundary of the black hole region) and its unstable behaviour,  we are only interested in studying the wave equation in the domain of outer communications \textbf{including the horizon} $\mathcal{H}^{+}$. Note that the study of the horizon is of fundamental importance since any attempt to prove the nonlinear stability or instability of the exterior of black holes must come to terms with the structure of the horizon.

We consider a connected asymptotically flat SO(3)-invariant spacelike hypersurface $\Sigma_{0}$ in $\mathcal{N}$ terminating at $i^{0}$ with boundary  such that $\partial\Sigma_{0}=\Sigma_{0}\cap\mathcal{H}^{+}$. We also assume that if $n$ is its future directed unit normal and $T=\partial_{v}$ then there exist positive constants $C_{1}<C_{2}$ such that 
\begin{equation*}
\begin{split}
&C_{1}<-g\left(n,n\right)<C_{2}, \ \ \  C_{1}<-g\left(n,T\right)<C_{2}.
\end{split}
\end{equation*}
Let $\mathcal{M}$ be the domain of dependence of $\Sigma_{0}$. Then, using the coordinate system $(v,r)$ we have
\begin{equation}
\mathcal{M}=\left(\left(-\infty,+\infty\right)\times\left[M\right.\!,+\infty\left.\right)\times\mathbb{S}^{2}\right)\cap J^{+}\left(\Sigma_{0}\right),
\label{Mern}
\end{equation}where  $J^{+}\left(\Sigma_{0}\right)$ is the causal future of $\Sigma_{0}$ (which by our convention includes $\Sigma_{0}$). Note  that $\mathcal{M}$ is a manifold with stratified (piecewise smooth) boundary  $\partial\mathcal{M}=(\mathcal{H}^{+}\cap\mathcal{M})\cup\Sigma_{0}$.
\begin{figure}[H]
	\centering
		\includegraphics[scale=0.11]{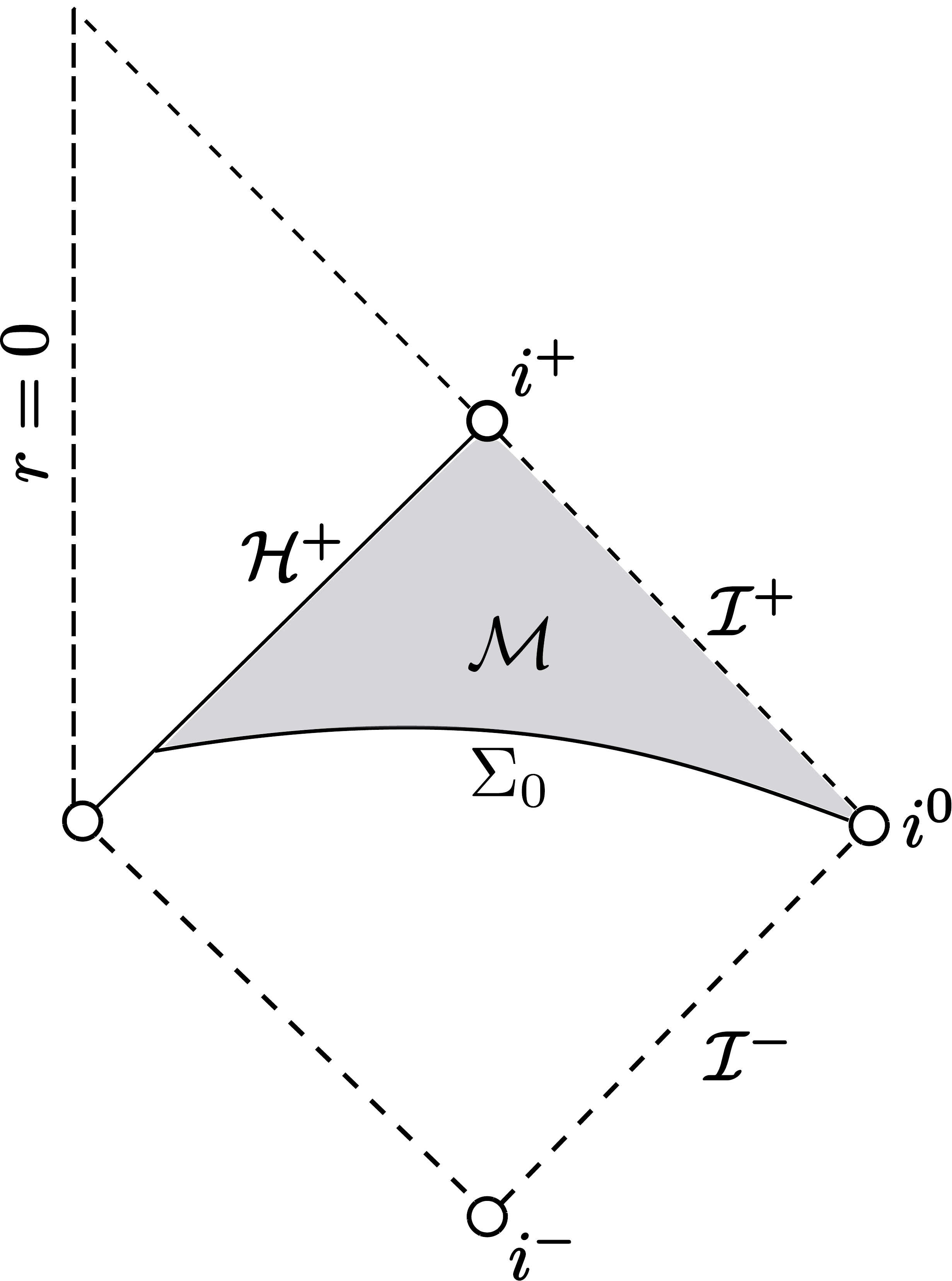}
	\label{fig:exrnw1}
\end{figure}

Another coordinate system that partially covers $\tilde{\mathcal{M}}$ is the null system $(u,v)$ where $u=t-r^{*},v=t+r^{*}$
and with respect to which the metric is $g=-Ddudv+r^{2}g_{\scriptstyle\mathbb{S}^{2}}.$
The hypersurfaces $v=c$ and $u=c$ are null and thus this system is useful for applying the method of characteristics.

\subsection{The Foliations $\Sigma_{\tau}$ and $\tilde{\Sigma}_{\tau}$}
\label{sec:TheFoliationsSigmaTauAndTildeSigmaTau}

We  consider the foliation $\Sigma_{\tau}=\varphi_{\tau}(\Sigma_{0})$, where $\varphi_{\tau}$ is the flow of $T=\partial_{v}$. Of course, since $T$ is Killing, the hypersurfaces $\Sigma_{\tau}$ are all isometric to $\Sigma_{0}$. 
 \begin{figure}[H]
	\centering
		\includegraphics[scale=0.11]{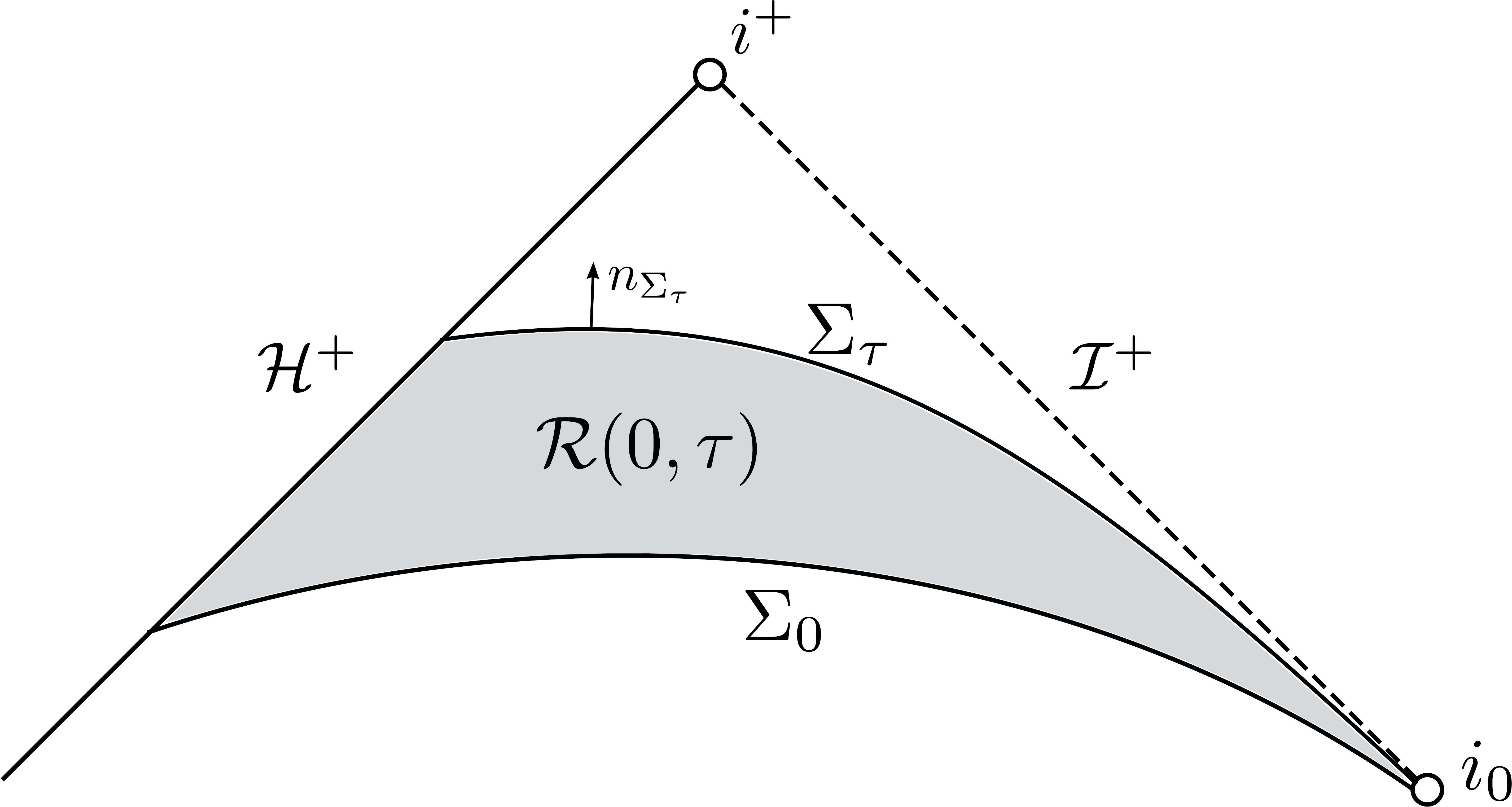}
	\label{fig:ernt1}
\end{figure}
We define the region
\begin{equation*} 
\mathcal{R}(0,\tau)=\cup_{0\leq \tilde{\tau}\leq \tau}\Sigma_{\tilde{\tau}}. 
\end{equation*}
On $\Sigma_{\tau}$ we have an induced Lie propagated coordinate system $(\rho,\omega)$ such that $\rho\in[\left.\right.\!\! M,+\infty)$ and $\omega\in\mathbb{S}^{2}$. These coordinates are defined such that if $Q\in\Sigma_{\tau}$ and $Q=(v_{Q},r_{Q},\omega_{Q})$  then $\rho=r_{Q}$ and $\omega=\omega_{Q}$. Our assumption for the normal $n_{\Sigma_{0}}$ (and thus for $n_{\Sigma_{\tau}})$ implies that there exists a bounded function $g_{1}$ such that 
\begin{equation}
\partial_{\rho}=g_{1}\partial_{v}+\partial_{r}.
\label{eq:rho}
\end{equation}
This defines a coordinate system since $[\partial_{\rho},\partial_{\theta}]=[\partial_{\rho},\partial_{\phi}]=0$. Moreover, the volume form of $\Sigma_{\tau}$ is
\begin{equation}
dg_{\scriptstyle\Sigma_{\tau}}=V\rho^{2}d\rho d\omega,
\label{volumeformfoliation}
\end{equation}
 where $V$ is a positive bounded function.
 \begin{figure}[H]
	\centering
		\includegraphics[scale=0.11]{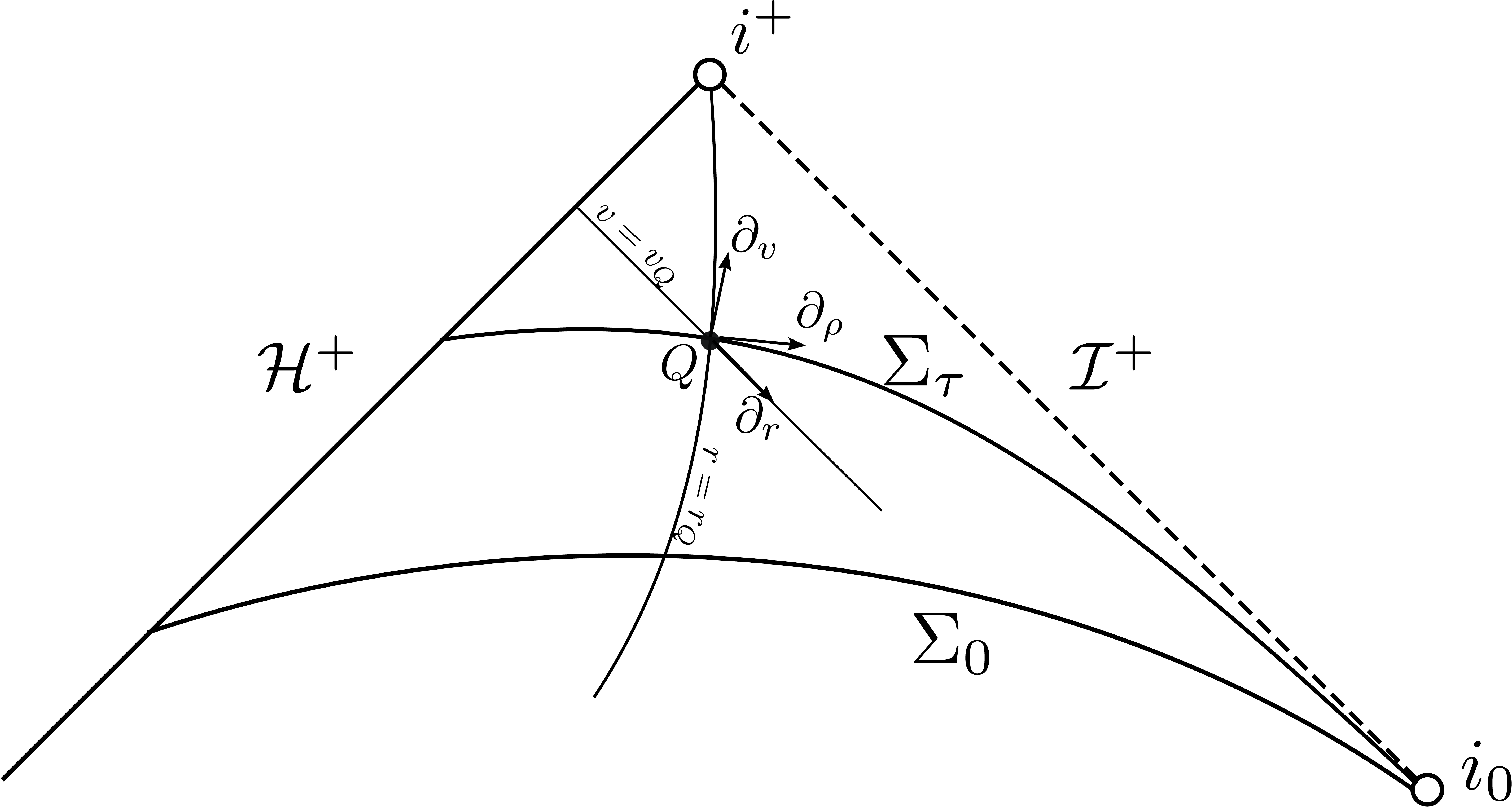}
	\label{fig:ernt2}
\end{figure}
In the companion paper \cite{aretakis2} we shall make use of another foliation denoted by $\tilde{\Sigma}_{\tau}$ whose leaves  terminate  at $\mathcal{I}^{+}$ and thus  ``follow" the waves to the future. 
 \begin{figure}[H]
	\centering
		\includegraphics[scale=0.1]{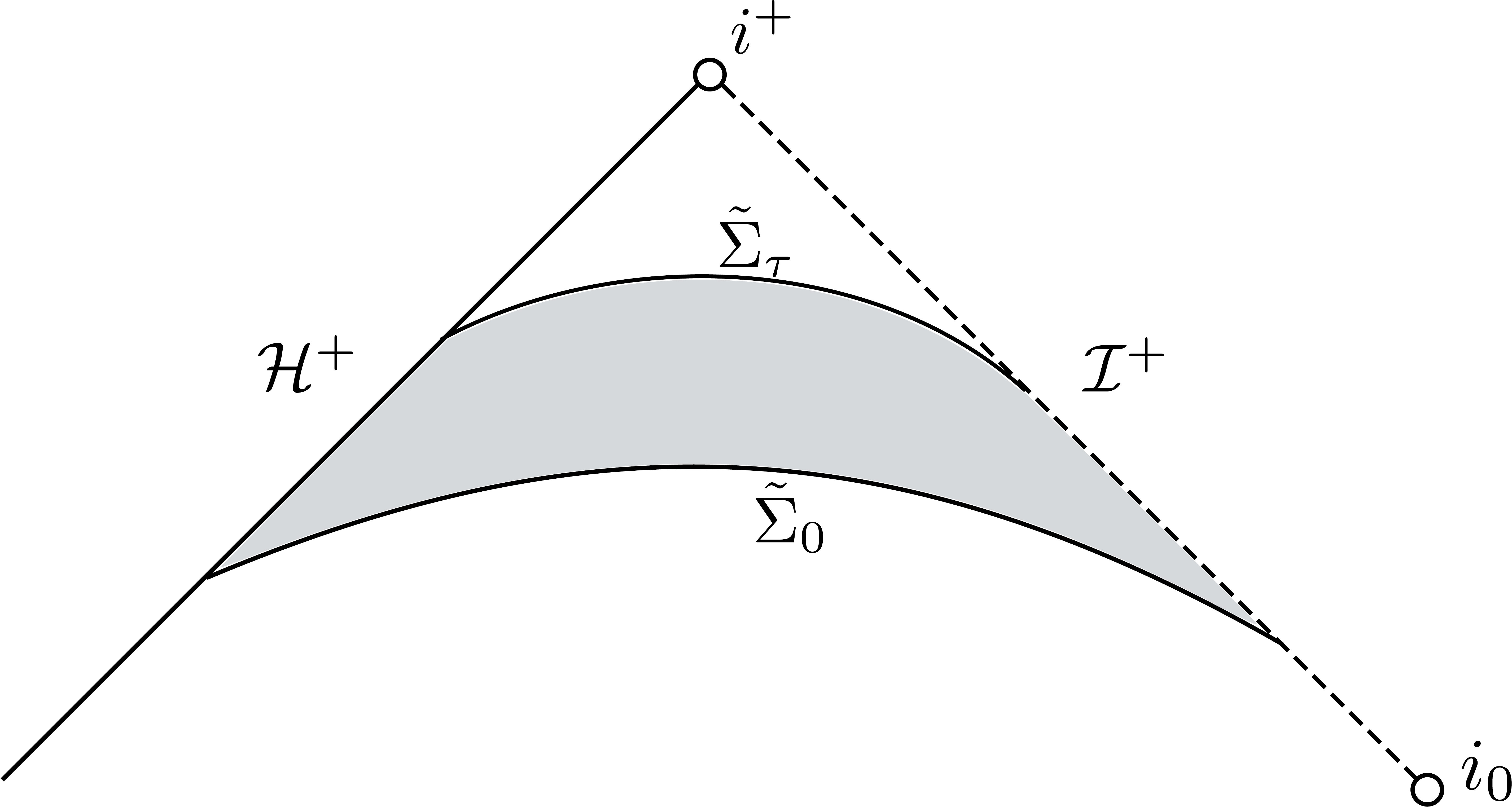}
	\label{fig:ernt3}
\end{figure}
One can similarly define an induced coordinate system on $\tilde{\Sigma}_{\tau}$. Note that only local elliptic estimates are to be applied on $\tilde{\Sigma}_{\tau}$.

\subsection{The Photon Sphere and Trapping Effect}
\label{sec:PhotonSphereAndTrappingEffect}

One can easily see that there exist orbiting future directed null geodesics, i.e.~null geodesics that neither cross the horizon $\mathcal{H}^{+}$ nor meet  null infinity $\mathcal{I}^{+}$. A class of such  geodesics $\gamma$ are of the form
\begin{equation*}
\begin{split}
\gamma:&\mathbb{R}\rightarrow\mathcal{M}\\
&\tau\mapsto \gamma\left(\tau\right)=\left(t\left(\tau\right),Q,\frac{\pi}{2},\phi\left(\tau\right)\right).
\end{split}
\end{equation*}
The conditions $\nabla_{\overset{.}{\gamma}}\overset{.}{\gamma}=0$ and $g\left(\overset{.}{\gamma},\overset{.}{\gamma}\right)=0$ imply that  $Q=2M,$ which is the radius of the so called \textit{photon sphere}. The $t,\phi$ depend linearly on $\tau$.
 \begin{figure}[H]
	\centering
		\includegraphics[scale=0.1]{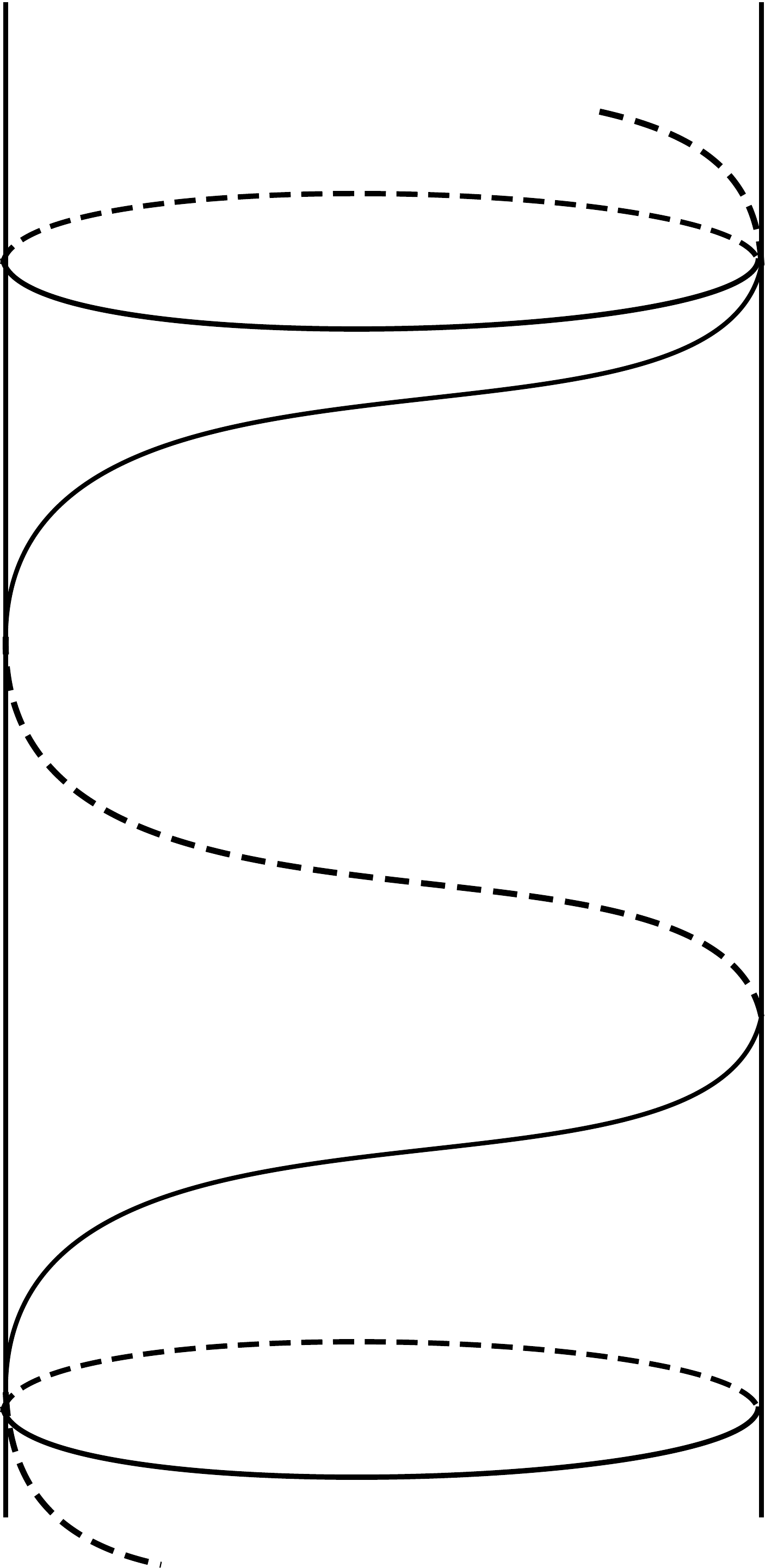}
	\label{fig:ps1}
\end{figure}
In fact, from any point  there is a codimension-one subset of future
directed null directions whose corresponding geodesics approach the photon sphere to the future.  The existence of this ``sphere'' (which is in fact a 3-dimensional timelike hypersurface) implies that the energy of some photons is not scattered to null  infinity or  the black hole region. This is the so called \textit{trapping effect}.  As we shall see, this effect forces us to require higher regularity for the waves in order to achieve decay results.

\subsection{The Redshift Effect and Surface Gravity of $\mathcal{H}^{+}$}
\label{sec:RedshiftEffectAndSurfaceGravityOfH}
The Killing vector field $\partial_{v}$ becomes null on the event horizon $\mathcal{H}^{+}$ and is also  tangent to it and thus  $\mathcal{H}^{+}$ is a Killing horizon. In general, if there exists a Killing vector field $V$ which is normal to a null hupersurface then
\begin{equation}
\nabla_{V}V=\kappa V
\label{sg}
\end{equation}
on the hypersurface. Since $V$ is Killing, the function $\kappa$ is constant along the integral curves of $V$. This can be seen by taking the pushforward of \eqref{sg} via the flow of $V$ and noting that since the flow of $V$ consists of isometries, the pushforward of the Levi-Civita connection is the same connection. The quantity $\kappa$ is called the \textit{surface gravity}\footnote{This plays a significant role in  black hole ``thermodynamics" (see also~\cite{wald} and~\cite{poisson}).} of the null hypersurface\footnote{Note that in Riemannian geometry, any Killing vector field that satisfies \eqref{sg} must have $\kappa=0$.}.  In the Reissner-Nordstr\"{o}m family, the surface gravities of the two horizons $\left\{r=r_{-}\right\}$ and $\left\{r=r_{+}\right\}$ are given by 
\begin{equation}
\kappa_{\pm}=\frac{r_{\pm}-r_{\mp}}{2r^{2}_{\pm}}=\frac{1}{2}\left.\frac{dD\left(r\right)}{dr}\right|_{r=r_{\pm}},
\label{sgrn}
\end{equation}
where $D$ is given in Section \ref{sec:GeometryOfExtremeReissnerNordstromSpacetime} (and $r_{+}, r_{-}$ are the roots of $D$). Note that in  extreme Reissner-Nordstr\"{o}m spacetime we have $D\left(r\right)=\left(1-\frac{M}{r}\right)^{2}$ and so $r_{+}=r_{-}=M$ which implies that the surface gravity vanishes. A horizon whose surface gravity vanishes is called \textit{degenerate}.

Physically, the surface gravity is related to the so-called \textit{redshift effect} that is observed along (and close to) $\mathcal{H}^{+}$. According to this effect, the wavelength of radiation close to $\mathcal{H}^{+}$ becomes longer as $v$ increases and thus the radiation gets less energetic.  This effect has a long history in the heuristic analysis of waves but only in the last decade has it  been used mathematically. For example,  Price's law (see \cite{price}) and the stability and instability of Cauchy horizons in an appropriate setting (see \cite{d1}) were proved using, in particular,  this effect. (Note that for the latter, one also needs to use the dual blueshift effect which is present in the interior of black holes.)

\section{The Cauchy Problem for the Wave Equation}
\label{sec:TheCauchyProblemForTheWaveEquation}

We consider solutions of the Cauchy problem of the wave equation \eqref{1eq} with initial data 
\begin{equation}
\left.\psi\right|_{\Sigma_{0}}=\psi_{0}\in H^{k}_{\operatorname{loc}}\left(\Sigma_{0}\right), \left.n_{\Sigma_{0}}\psi\right|_{\Sigma_{0}}=\psi_{1}\in H^{k-1}_{\operatorname{loc}}\left(\Sigma_{0}\right),
\label{cd}
\end{equation}
where the hypersurface $\Sigma_{0}$ is as defined in Section \ref{sec:GeometryOfExtremeReissnerNordstromSpacetime} and $n_{\Sigma_{0}}$ denotes the future unit normal of $\Sigma_{0}$.  In view of the global hyperbolicity of $\mathcal{M}$, there exists a unique solution to the above equation. Moreover, as long as $k\geq 1$, we have that for any spacelike hypersurface  $S$
\begin{equation*}
\left.\psi\right|_{S}\in H^{k}_{\operatorname{loc}}\left(S\right), \left.n_{S}\psi\right|_{S}\in H^{k-1}_{\operatorname{loc}}\left(S\right).
\end{equation*}
In this paper we will be interested in the case where $k\geq 2$.  Moreover, we assume that 
\begin{equation}
\lim_{x\rightarrow i^{0}}r\psi^{2}(x)=0.
\label{condition}
\end{equation}
 For simplicity, from now on, \textbf{when we say ``for all solutions $\psi$ of the wave equation" we will assume that $\psi$ satisfies the above conditions}. Note that for obtaining sharp decay results we will have to consider even higher regularity for $\psi$.

\section{The Main Theorems}
\label{sec:TheMainTheorems}

We consider the Cauchy problem for the wave equation (see Section \ref{sec:TheCauchyProblemForTheWaveEquation}) on the extreme Reissner-Nordstr\"{o}m spacetime. This spacetime is partially covered by the coordinate systems $(t,r)$, $(t,r^{*})$, $(v,r)$ and $(u,v)$  described in Section \ref{sec:GeometryOfExtremeReissnerNordstromSpacetime}. Recall that $M$ is a positive parameter and  $D=D(r)=\left(1-\frac{M}{r}\right)^{2}$. Recall also that the horizon $\mathcal{H}^{+}$ is located at $\left\{r=M\right\}$ and the photon sphere at $\left\{r=2M\right\}$.

We denote $T=\partial_{v}=\partial_{t}$, where $\partial_{v}$ corresponds to the system $(v,r)$ and $\partial_{t}$ corresponds to $(t,r)$.  From now on, \textbf{$\partial_{v},\partial_{r}$ are the coordinate vector fields corresponding to $(v,r)$}, unless otherwise stated. Note that $\partial_{r^{*}}=\partial_{v}=T$ on $\mathcal{H}^{+}$ and therefore it is not transversal to $\mathcal{H}^{+}$, whereas $\partial_{r}$ is transversal to $\mathcal{H}^{+}$.

The foliation $\Sigma_{\tau}$ is defined in Section \ref{sec:TheFoliationsSigmaTauAndTildeSigmaTau} and the current $J^{V}$ associated to the vector field $V$ is defined in Section \ref{sec:TheCurrentsJKAndMathcalE}. Note that every time we use such a current we refer to $V$ as a multiplier vectorfield. For reference, we mention that 
\begin{equation*}
J_{\mu}^{T}[\psi]n^{\mu}_{\Sigma}\sim \, (T\psi)^{2}+\left(1-\frac{M}{r}\right)^{2}(\partial_{r}\psi)^{2}+\left|\nabb\psi\right|^{2},
\end{equation*}
which degenerates on $\mathcal{H}^{+}$ whereas
\begin{equation*}
J_{\mu}^{n_{\Sigma}}[\psi]n^{\mu}_{\Sigma}\sim \, (T\psi)^{2}+(\partial_{r}\psi)^{2}+\left|\nabb\psi\right|^{2},
\end{equation*}
which does not degenerate on $\mathcal{H}^{+}$.
The Fourier decomposition of $\psi$ on $\mathbb{S}^{2}(r)$ is discussed in Section \ref{sec:EllipticTheoryOnMathbbS2}, where it is also defined what it means for a function to be supported on a given range of angular frequencies. Note that all the integrals are considered with respect to the volume form. The initial data are assumed to be as in Section \ref{sec:TheCauchyProblemForTheWaveEquation} and  sufficiently regular  such that the right hand side of the estimates below are all finite. Then we have the following

\begin{mytheo}(\textbf{Morawetz and $X$ Estimates})
Let $\delta >0$. There exists a  constant $C_{\delta}>0$ which depends  on $M$, $\delta$ and $\Sigma_{0}$ such that for all solutions $\psi$ of the wave equation the following estimates hold
\begin{enumerate}
	\item 
\textbf{Non-Degenerate Zeroth Order Morawetz Estimate:}
	 \begin{equation*}
	\begin{split}
	\displaystyle\int_{\mathcal{R}(0,\tau)}{\frac{1}{r^{3+\delta}}\psi^{2}}\leq C_{\delta}\displaystyle\int_{\Sigma_{0}}{J^{n_{\Sigma_{0}}}_{\mu}[\psi]n^{\mu}_{\Sigma_{0}}}.
		\end{split}
	\end{equation*}
	\item 
\textbf{$X$ Estimate with Degeneracy at $\mathcal{H}^{+}$ and Photon Sphere:}
	\begin{equation*}
	\begin{split}
\displaystyle\int_{\mathcal{R}(0,\tau)}\!\!\!{\left(\frac{(r-2M)^{2}\cdot \sqrt{D}}{r^{3+\delta}}\left((\partial_{v}\psi)^{2}+\left|\nabb\psi\right|^{2}+D^{2}(\partial_{r}\psi)^{2}\right)\right)+\chi_{\left[\frac{3M}{2},\frac{5M}{2}\right]}\cdot(\partial_{r^{*}}\psi)^{2}}\leq C_{\delta}\displaystyle\int_{\Sigma_{0}}{J_{\mu}^{T}[\psi]n^{\mu}_{\Sigma_{0}}},
	\end{split}
	\end{equation*}
	where $\chi_{\left[\frac{3M}{2},\frac{5M}{2}\right]}(r)$ is the indicator function of the interval $\left[\frac{3M}{2},\frac{5M}{2}\right]$.
		\item \textbf{$X$ Estimate with Degeneracy at $\mathcal{H}^{+}$:}
		\begin{equation*}
	\begin{split}
&\displaystyle\int_{\mathcal{R}(0,\tau)}\!\!\!{\left(\frac{\sqrt{D}}{r^{1+\delta}}\left(\partial_{v}\psi\right)^{2}+\frac{D^{2}\sqrt{D}}{r^{1+\delta}}\left(\partial_{r}\psi\right)^{2}+\frac{\sqrt{D}}{r}\left|\nabb\psi\right|^{2}\right)}
\leq C_{\delta}\displaystyle\int_{\Sigma_{0}}{\left(J_{\mu}^{T}[\psi]n^{\mu}_{\Sigma_{0}}+J_{\mu}^{T}[T\psi]n^{\mu}_{\Sigma_{0}}\right)}.
		\end{split}
	\end{equation*}
\end{enumerate}		
	\label{th1}
\end{mytheo}

For proving statement (1), we first derive  in Section \ref{sec:SpacetimeL2EstimateForPsi} an estimate which degenerates on $\mathcal{H}^{+}$. Note that the estimate of Section \ref{sec:SpacetimeL2EstimateForPsi} requires applying only the Killing vectorfield $T$ as a multiplier. Later in Section \ref{sec:UniformPointwiseBoundedness} we construct a timelike multiplier and we use an appropriate version of Hardy's inequality to eliminate this degeneracy (see Theorem \ref{theorem3}). 

Note that estimate (2)  degenerates on the photon sphere  with respect to all derivatives except  precisely $\partial_{r^{*}}$. Note moreover, that the same estimate degenerates on $\mathcal{H}^{+}$ with respect to all derivatives.   As we shall see, unlike the subextreme case,  the degeneracy on $\mathcal{H}^{+}$ with respect to the transversal derivative $\partial_{r}$ can not be eliminated even if we apply a timelike multiplier (see also Theorem \ref{t2}).

In Section \ref{sec:ANonDegenerateXEstimate} we commute with the vectorfield $T$ to  retrieve the tangential to the photon sphere derivatives thus obtaining  statement (3).

\begin{mytheo}(\textbf{Uniform Boundedness of Non-Degenerate Energy}) 
There exists $r_{0}$ such that $M<r_{0}<2M$ and a  constant $C>0$ which depends on $M$, and $\Sigma_{0}$ such that if $\mathcal{A}=\mathcal{R}(0,\tau)\cap\left\{M\leq r\leq r_{0}\right\}$ then for all solutions $\psi$ of the wave equation we have
\begin{equation*}
\int_{\mathcal{A}}{\left(\sqrt{D}(\partial_{r}\psi)^{2}+(\partial_{v}\psi)^{2}+\left|\nabb\psi\right|^{2}\right)}+
\displaystyle\int_{\Sigma_{\tau}}{J_{\mu}^{n_{\Sigma_{\tau}}}[\psi]n^{\mu}_{\Sigma_{\tau}}}\leq C\displaystyle\int_{\Sigma_{0}}{J_{\mu}^{n_{\Sigma_{0}}}[\psi]n^{\mu}_{\Sigma_{0}}}.
\end{equation*}
\label{t2}
\end{mytheo}
The above theorem is proved in Section \ref{sec:UniformBoundednessOfLocalObserverSEnergy} using  a novel redshift current constructed in Section \ref{sec:TheVectorFieldN} and a new version of Hardy's inequality (see Section \ref{sec:HardyInequalities}). Note that this theorem shows that we can not eliminate the degeneracy of $\partial_{r}$ in spacetime neighbourhoods of $\mathcal{H}^{+}$ even if we apply timelike multipliers. To eliminate this degeneracy we need to reveal a new feature of degenerate horizons captured in the next theorem.

\begin{mytheo}(\textbf{Trapping Effect on the Event Horizon $\mathcal{H}^{+}$}) Let $\mathcal{A}$ be the spacetime region defined in Theorem \ref{t2} and $\delta>0$.  Then  there exists a   constant $C>0$ (and $C_{\delta}>0$) which depends on $M$ and $\Sigma_{0}$ (and $\delta$) such that    for all solutions $\psi$ with vanishing spherical mean (i.e. for all $\psi$   supported on the angular frequencies $l\geq 1$), the following hold
\begin{enumerate}
	\item \textbf{Sharp Second Order $L^{2}$ Estimates:}
	\begin{equation*}
\begin{split}
&\int_{\Sigma_{\tau}\cap\mathcal{A}}{\left(\partial_{v}\partial_{r}\psi\right)^{2}+\left(\partial_{r}\partial_{r}\psi\right)^{2}+\left|\nabb\partial_{r}\psi\right|^{2}}+\int_{\mathcal{H}^{+}}{\left(\partial_{v}\partial_{r}\psi\right)^{2}+\chi_{1}\left|\nabb\partial_{r}\psi\right|^{2}}\\&+\int_{\mathcal{A}}{\left(\partial_{v}\partial_{r}\psi\right)^{2}+\sqrt{D}\left(\partial_{r}\partial_{r}\psi\right)^{2}+\left|\nabb\partial_{r}\psi\right|^{2}}\\\leq &\, C\int_{\Sigma_{0}}{J_{\mu}^{n_{\Sigma_{0}}}[\psi]n^{\mu}_{\Sigma_{0}}}+C\int_{\Sigma_{0}}{J_{\mu}^{n_{\Sigma_{0}}}[T\psi]n^{\mu}_{\Sigma_{0}}}+C\int_{\Sigma_{0}\cap\mathcal{A}}{J_{\mu}^{n_{\Sigma_{0}}}[\partial_{r}\psi]n^{\mu}_{\Sigma_{0}}},
\end{split}
\end{equation*}
where $\chi_{1}=0$ if $\psi$ is supported on $l=1$ and $\chi_{1}=1$ if $\psi$ is supported on $l\geq 2$. 
		\item \textbf{Local Integrated Energy Decay:}
\begin{equation*}
	\begin{split}
\displaystyle\int_{\mathcal{R}(0,\tau)}{\frac{1}{r^{1+\delta}}J_{\mu}^{n_{\Sigma_{\tau}}}[\psi]n^{\mu}_{\Sigma_{\tau}}}\leq C_{\delta}\displaystyle\int_{\Sigma_{0}}{J_{\mu}^{n_{\Sigma_{0}}}[\psi]n^{\mu}_{\Sigma_{0}}}+C_{\delta}\int_{\Sigma_{0}}{J_{\mu}^{n_{\Sigma_{0}}}[T\psi]n^{\mu}_{\Sigma_{0}}}+C_{\delta}\int_{\Sigma_{0}\cap\mathcal{A}}{J_{\mu}^{n_{\Sigma_{0}}}[\partial_{r}\psi]n^{\mu}_{\Sigma_{0}}}.
	\end{split}
	\end{equation*}	
\end{enumerate}
\label{theorem3}
\end{mytheo}

Statement (1) of Theorem \ref{theorem3} is proved in Section \ref{sec:CommutingWithAVectorFieldTransversalToMathcalH} where we construct a new current and derive appropriate Hardy inequalities along $\mathcal{H}^{+}$. Note that the restriction on frequencies $l\geq 1$ is required in view of a new low-frequency phenomenon of degenerate horizons (see also Theorem \ref{nondecay}). In particular, we note that in \cite{aretakis2} we show that there is no constant $C$ such that statement (1) holds for all solutions $\psi$ of the wave equation. Therefore, the assumption on the frequency range is sharp.

In statement (2) we have eliminated the degeneracy of $\partial_{r}$ in neighbourhoods of $\mathcal{H}^{+}$ at the expense, however, of commuting the wave equation with $\partial_{r}$ and thus requiring higher regularity for $\psi$. This allows us to conclude that the event horizon $\mathcal{H}^{+}$ exhibits trapping.

\begin{mytheo}(\textbf{Pointwise Boundedness})
There exists a constant $C$ which depends on $M$ and $\Sigma_{0}$ such that for all solutions $\psi$ of the wave equation we have
\begin{equation*}
\left|\psi\right|\leq C\sqrt{\tilde{E_{4}}},
\end{equation*}
everywhere in $\mathcal{R}$,
where
\begin{equation*}
\tilde{E_{4}}=\!\!\int_{\Sigma_{0}}{J_{\mu}^{n_{{\Sigma}_{0}}}[\psi]n^{\mu}_{\Sigma_{0}}}+C\!\!\displaystyle\int_{\Sigma_{0}}{J_{\mu}^{n_{{\Sigma}_{0}}}[n_{{\Sigma}_{0}}\psi]n^{\mu}_{\Sigma_{0}}}.
\end{equation*}
\label{t5}
\end{mytheo}
This theorem is proved in Section \ref{sec:UniformPointwiseBoundedness} using Theorem \ref{t2} and Sobolev inequalities.

\begin{mytheo}(\textbf{Non-Decay})
For generic initial data which give rise to solutions $\psi$ of the wave equation the quantity
\begin{equation*}
\psi^{2}+(\partial_{r}\psi)^{2}
\end{equation*}
does not decay along $\mathcal{H}^{+}$.
\label{nondecay}
\end{mytheo}
This theorem arises from a new low-frequency phenomenon of degenerate horizons. In particular,  in Section \ref{sec:TheSphericallySymmetricCase} we show that an appropriate expression which corresponds to the spherical mean of $\psi$ is conserved along $\mathcal{H}^{+}$ under the evolution. This low-frequency ``instability'' will be extensively investigated in the companion papers \cite{aretakis2, aretakis3}. 

\section{The Vector Field Method}
\label{sec:TheVectorFieldMethod}

For understanding the evolution of waves we will use the so-called vector field method. This is a geometric and robust method and involves mainly $L^{2}$ estimates. The main idea is to construct appropriate (0,1) currents and use Stokes' theorem (see Appendix \ref{sec:StokesTheoremOnLorentzianManifolds}) in appropriate regions. For a nice recent exposition see \cite{ali}. We briefly recall here the method to set notation.

Given a (0,1) current $P_{\mu}$ we have the continuity equation
\begin{equation}
\int_{\Sigma_{0}}{P_{\mu}n^{\mu}_{\Sigma_{0}}}=\int_{\Sigma_{\tau}}{P_{\mu}n^{\mu}_{\Sigma_{\tau}}}+\int_{\mathcal{H}^{+}\left(0,\tau\right)}{P_{\mu}n^{\mu}_{\mathcal{H}^{+}}}+\int_{\mathcal{R}\left(0,\tau\right)}{\nabla^{\mu}P_{\mu}}
\label{div}
\end{equation}
where all the integrals are with respect to the \textit{induced volume form} and the unit normals $n_{\scriptstyle\Sigma_{\tau}}$ are future directed. All the integrals  are to be considered  with respect to the induced volume form and thus we omit writing the measure. Our choice is $n_{\mathcal{H}^{+}}=\partial_{v}=T$ (and thus the volume element on $\mathcal{H}^{+}$ is chosen so \eqref{div} holds). 

\subsection{The Compatible Currents $J$, $K$ and the Current $\mathcal{E}$}
\label{sec:TheCurrentsJKAndMathcalE}

We usually consider currents $P_{\mu}$ that depend on the geometry of $\left(\mathcal{M},g\right)$ and are such that both $P_{\mu}$ and $\nabla^{\mu}P_{\mu}$ depend only on the 1-jet of $\psi$. This can be achieved by using the wave equation to make all second order derivatives disappear.  There is a general method for producing such currents using the energy momentum tensor \textbf{T}. Indeed, the Lagrangian structure of the wave equation gives us the following energy momentum tensor 
\begin{equation}
\textbf{T}_{\mu\nu}\left[\psi\right]=\partial_{\mu}\psi\partial_{\nu}\psi-\frac{1}{2}g_{\mu\nu}\partial^{a}\psi\partial_{a}\psi,
\label{tem}
\end{equation}
which is a symmetric divergence free (0,2) tensor. We will in fact consider this tensor for general functions $\psi:\mathcal{M}\rightarrow\mathbb{R}$ in which case we have the identity
\begin{equation}
\text{Div}\textbf{T}\left[\psi\right]=\left(\Box_{g}\psi\right)d\psi.
\label{divT}
\end{equation}
Since  \textbf{T} is a $\left(0,2\right)$ tensor we need to  contract it with vector fields of $\mathcal{M}$. It is here where the geometry of $\mathcal{M}$ makes its appearance. Given a vector field $V$  we define the $J^{V}$ current by
\begin{equation}
J^{V}_{\mu}[\psi]=\textbf{T}_{\mu\nu}[\psi]V^{\nu}.
\label{jcurrent}
\end{equation}
We say that we use the vector field $V$ as a \textit{multiplier}\footnote{the name comes from the fact that the tensor \textbf{T} is multiplied by \textit{V}.} if we apply \eqref{div} for the current $J^{V}_{\mu}$. The divergence of this current is $\operatorname{Div}(J)=\operatorname{Div}\left(\textbf{T}\right)V+\textbf{T}\left(\nabla V\right),$
where $\left(\nabla V\right)^{ij}=\left(g^{ki}\nabla _{k}V\right)^{j}=\left(\nabla ^{i}V\right)^{j}$. If $\psi$ is a solution of the wave equation then $\operatorname{Div}\left(\textbf{T}\right)=0$ and, therefore, $\nabla^{\mu}J_{\mu}^{V}$ is an expression of the 1-jet of $\psi$.
Given a vector field $V$  the scalar  current $K^{V}$ is defined by
\begin{equation}
K^{V}\left[\psi\right]=\textbf{T}\left[\psi\right]\left(\nabla V\right)=\textbf{T}_{ij}\left[\psi\right]\left(\nabla ^{i}V\right)^{j}.
\label{K}
\end{equation}
Note that from the symmetry of  $\textbf{T}$ we have
$K^{V}\left[\psi\right]=\textbf{T}_{\mu\nu}\left[\psi\right]\pi^{\mu\nu}_{V},$
where $\pi^{\mu\nu}_{V}=(\mathcal{L}_{V}g)^{\mu\nu}$ is the deformation tensor of $V$. Clearly if $\psi$ satisfies the wave equation then
$K^{V}[\psi]=\nabla^{\mu}J^{V}_{\mu}[\psi].$
Thus if we use Killing vector fields as multipliers then the divergence vanishes and so we obtain a conservation law. This is partly the content of a deep theorem of Noether\footnote{According to this theorem, any continuous family of isometries gives rise to a conservation law.}. In general we define the scalar  current $\mathcal{E}^{V}$ by
\begin{equation}
\mathcal{E}^{V}[\psi]=\text{Div}\left(\textbf{T}\right)V=\left(\Box_{g}\psi\right)d\psi\left(V\right)=\left(\Box_{g}\psi\right)\left(V\psi\right).
\label{E}
\end{equation}

\subsection{The Hyperbolicity of the Wave Equation}
\label{sec:TheHyperbolicityOfTheWaveEquation}
The hyperbolicity of the wave equation is captured by the following proposition
\begin{proposition}
Let $V_{1},V_{2}$ be two future directed timelike vectors. Then the quadratic expression $\textbf{T}\left(V_{1},V_{2}\right)$ is positive definite in $d\psi$. By continuity, if one of these vectors is null then $\textbf{T}\left(V_{1},V_{2}\right)$ is non-negative definite in $d\psi$.
\label{hyperb}
\end{proposition}
\begin{proof}
Consider a point $p\in\mathcal{M}$ and the normal coordinates around this point and that without loss of generality $V_{2}=(1,0,0,...,0)$. Then the proposition is an application of the Cauchy-Schwarz inequality.
\end{proof}
An important application of the vector field method and the hyperbolicity of the wave equation is the domain of dependence property. 

As we shall see, the exact dependence of $\textbf{T}$ on the derivatives of $\psi$ will be crucial later. For a general computation see Appendix \ref{sec:TheHyperbolicityOfTheWaveEquation1}.

\section{Hardy Inequalities}
\label{sec:HardyInequalities}

In this section we establish three Hardy inequalities which as we shall see will be very crucial for obtaining sharp estimates.  They do not require $\psi$ to satisfy the wave equation and so we only assume that $\psi$ satifies the regularity assumptions described in Section \ref{sec:TheCauchyProblemForTheWaveEquation} and \eqref{condition}. From now on,  $\Sigma_{\tau}$  is the foliation introduced  in Section \ref{sec:TheFoliationsSigmaTauAndTildeSigmaTau} (however, these inequalities hold also for the foliation $\tilde{\Sigma}_{\tau}$).  

\begin{proposition}(\textbf{First Hardy Inequality}) For all functions $\psi$ which satisfy the regularity assumptions of Section \ref{sec:TheCauchyProblemForTheWaveEquation}  we have
\begin{equation*}
\int_{\Sigma_{\tau}}{\frac{1}{r^{2}}\psi^{2}}\leq C\int_{\Sigma_{\tau}}{D[(\partial_{v}\psi)^{2}+(\partial_{r}\psi)^{2}]},
\end{equation*}
where the constant  $C$ depends only on $M$ and $\Sigma_{0}$.
\label{firsthardy}
\end{proposition}
\begin{proof}
Consider the induced coordinate system $(\rho,\omega)$ introduced in Section \ref{sec:TheFoliationsSigmaTauAndTildeSigmaTau}. We use the 1-dimensional identity
\begin{equation}
\int_{M}^{+\infty}(\partial_{\rho}h)\psi^{2}d\rho=\left[h\psi^{2}\right]_{r=M}^{+\infty}-2\int_{M}^{+\infty}h\psi(\partial_{\rho}\psi) d\rho
\label{111}
\end{equation}
with $h=r-M$. In view of the assumptions on $\psi$, the boundary terms vanish. Thus, Cauchy-Schwarz gives
\begin{equation*}
\begin{split}
\int_{\left\{r\geq M\right\}}{\psi^{2}d\rho}=&-2\int_{\left\{r\geq M\right\}}{\!\!(r-M)\psi(\partial_{\rho}\psi) d\rho}
\leq 2\left(\int_{\left\{r\geq M\right\}}{\!\!\psi^{2}d\rho}\right)^{\frac{1}{2}}\!\!\left(\int_{\left\{r\geq M\right\}}{\!(r-M)^{2}(\partial_{\rho}\psi)^{2}}d\rho\right)^{\frac{1}{2}}\!\!.
\end{split}
\end{equation*}
Therefore,
\begin{equation*}
\begin{split}
\int_{\left\{r\geq M\right\}}{\psi^{2}d\rho}\leq 4\int_{\left\{r\geq M\right\}}{(r-M)^{2}(\partial_{\rho}\psi)^{2}d\rho}.
\end{split}
\end{equation*}
Integrating this inequality over $\mathbb{S}^{2}$ gives us
\begin{equation*}
\begin{split}
\int_{\mathbb{S}^{2}}\int_{\left\{r\geq M\right\}}{\frac{1}{\rho^{2}}\psi^{2}}\rho^{2}d\rho d\omega\leq 4\int_{\mathbb{S}^{2}}\int_{\left\{r\geq M\right\}}{D(\partial_{\rho}\psi)^{2}\rho^{2}d\rho d\omega}.
\end{split}
\end{equation*}
The result follows from the fact that each $\Sigma_{\tau}$ is diffeomorphic to $\mathbb{S}^{2}\times [M,\left.\right.\!\! +\infty)$, from \eqref{eq:rho} and the boundedness of the factor $V$ in the volume form \eqref{volumeformfoliation} of $\Sigma_{\tau}$.
\end{proof}
The importance of the above inequality lies in the weights. The weight of $(\partial_{\rho}\psi)^{2}$ vanishes to second order on $\mathcal{H}^{+}$  but does not degenerate at infinity\footnote{Note that if we replace $r-M$ with another function $h$ then we will  not be able to make this weight degenerate fast enough without obtaining non-trinial boundary terms.} whereas the weight of $\psi$ degenerate at infinity but not at $r=M$. Similarly, one might derive estimates for the non-extreme case\footnote{These estimates turn out to be stronger since the weight of $\psi$ may diverge at $r=M$ in an integrable manner.}. We also mention that the right hand side is bounded by the (conserved) flux of $T$ through $\Sigma_{\tau}$ (see Section \ref{sec:TheVectorFieldTextbfM}).

\begin{proposition}(\textbf{Second Hardy Inequality})  Let  $r_{0}\in (M,2M)$. Then for all functions $\psi$ which satisfy the regularity assumptions of Section \ref{sec:TheCauchyProblemForTheWaveEquation}  and any positive number $\epsilon$ we have
\begin{equation*}
\int_{\mathcal{H}^{+}\cap\Sigma_{\tau}}{\psi^{2}}\leq \epsilon\int_{\Sigma_{\tau}\cap\left\{r\leq r_{0}\right\}}{\!\!(\partial_{v}\psi)^{2}+(\partial_{r}\psi)^{2}}\ \,+C_{\epsilon}\int_{\Sigma_{\tau}\cap\left\{r\leq r_{0}\right\}}{\psi^{2}},
\end{equation*}
where the constant $C_{\epsilon}$ depends on $M$, $\epsilon$, $r_{0}$ and $\Sigma_{0}$. 
\label{secondhardy}
\end{proposition}
\begin{proof}
We use as before the 1-dimensional identity \eqref{111} with $h=r-r_{0}$. Then 
\begin{equation*}
\begin{split}
(r_{0}-M)\psi^{2}(M)=\int_{M}^{r_{0}}{\psi^{2}+2g\psi(\partial_{\rho}\psi) d\rho}
\leq  \epsilon\int_{M}^{r_{0}}{(\partial_{\rho}\psi)^{2}}d\rho +\int_{M}^{r_{0}}{\left(1+\frac{g^{2}}{\epsilon}\right)\psi^{2}d\rho},
\end{split}
\end{equation*}
for any $\epsilon >0$. By integrating over $\mathbb{S}^{2}$, using \eqref{eq:rho} and noting that the $\rho^{2}$ factor that appears in the volume form of $\Sigma_{\tau}\cap\left\{r\leq r_{0}\right\}$ is bounded we obtain the required result.
\end{proof}
The previous two inequalities concern the hypersurfaces $\Sigma_{\tau}$ that cross $\mathcal{H}^{+}$. The next estimate concerns spacetime neighbourhoods of $\mathcal{H}^{+}$.

\begin{proposition} (\textbf{Third Hardy Inequality})  Let  $r_{0},r_{1}$ be such that  $M<r_{0}<r_{1}$. We define the regions
$\mathcal{A}=\mathcal{R}(0,\tau)\cap\left\{M\leq r\leq r_{0}\right\},\ \mathcal{B}=\mathcal{R}(0,\tau)\cap\left\{r_{0}\leq r\leq r_{1}\right\}.$ Then for all functions $\psi$ which satisfy the regularity assumptions of Section \ref{sec:TheCauchyProblemForTheWaveEquation}   we have
\begin{equation*}
\int_{\mathcal{A}}{\psi^{2}}\leq C\int_{\mathcal{B}}{\psi^{2}}\ \,+C\int_{\mathcal{A}\cup\mathcal{B}}{D[(\partial_{v}\psi)^{2}+(\partial_{r}\psi)^{2}]},
\end{equation*}
where the constant $C$ depends on $M$, $r_{0}$, $r_{1}$  and $\Sigma_{0}$.
\label{thirdhardy}
\end{proposition}
\begin{proof}
 \begin{figure}[H]
	\centering
		\includegraphics[scale=0.11]{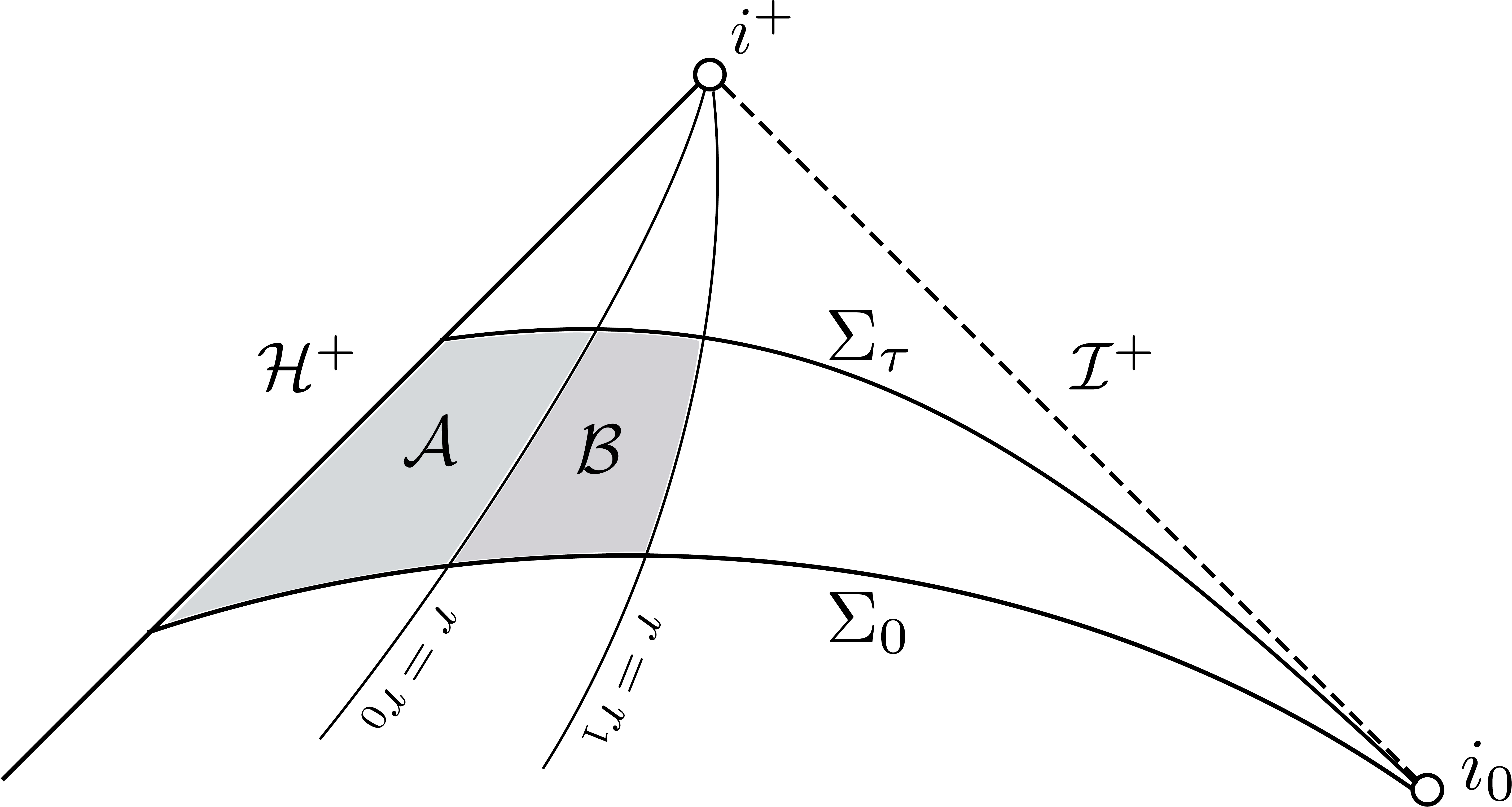}
	\label{fig:ernt4}
\end{figure}
We use again \eqref{111}
with $h$ such that $h=2(\rho-M)$ in  $[M,r_{0}]$ and $h(r_{1})=0$.   Then
\begin{equation*}
\begin{split}
2\int_{M}^{r_{0}}{\psi^{2}d\rho}\leq \int_{M}^{r_{0}}{\psi^{2}d\rho}+\int_{r_{0}}^{r_{1}}{\left(1-(\partial_{\rho}g)\right)\psi^{2} d\rho}+\int_{M}^{r_{1}}{g^{2}(\partial_{\rho}\psi)^{2} d\rho}.
\end{split}
\end{equation*}
Thus, integrating over $\mathbb{S}^{2}$ we obtain
\begin{equation*}
\begin{split}
\int_{\Sigma_{\tilde{\tau}}\cap\mathcal{A}}{\psi^{2}}\leq C \int_{\Sigma_{\tilde{\tau}}\cap\left(\mathcal{A}\cup\mathcal{B}\right)}{D[(\partial_{v}\psi)^{2}+(\partial_{r}\psi)^{2}]}+C\int_{\Sigma_{\tilde{\tau}}\cap\mathcal{B}}{\psi^{2}},
\end{split}
\end{equation*}
for all $\tilde{\tau}\geq 0$, where the constant $C$ depends on $M$, $r_{0}$, $r_{1}$ and $\Sigma_{0}$. Therefore, integrating over $\tilde{\tau}\in [0,\tau]$ and using the coarea formula
\begin{equation*}
\begin{split}
\int_{0}^{\tau}\left(\int_{\Sigma_{\tilde{\tau}}}{f}\right)d\tilde{\tau}\ \sim\,\int_{\bigcup_{0\leq\tilde{\tau}\leq\tau}\Sigma_{\tilde{\tau}}}{f}
\end{split}
\end{equation*}
completes the proof of the proposition.

\end{proof}

\section{Elliptic Theory on $\mathbb{S}^{2}(r), r>0$}
\label{sec:EllipticTheoryOnMathbbS2}
																									
In view of the symmetries of the spacetime, it is important to understand the behaviour of functions on the orbits of the action of SO(3), which are isometric to $\mathbb{S}^{2}(r)$ for  $r>0$. Recall that the eigenvalues of the spherical Laplacian $\lapp$ are equal to $\frac{-l(l+1)}{r^{2}}, l\in\mathbb{N}$. The dimension of the eigenspaces $E^{l}$ is equal to $2l+1$ 
and the corresponding eigenvectors are denoted by $Y^{m,l}, -l\leq m\leq l$ and called \textit{spherical harmonics}. We have $L^{2}(\mathbb{S}^{2}(r))=\displaystyle\oplus_{l\geq 0}E^{l}$ and, therefore, any function $\psi\in L^{2}(\mathbb{S}^{2}(r))$ can be written as
\begin{equation}
\psi=\sum_{l=0}^{\infty}\sum_{m=-l}^{l}\psi_{m,l}Y^{m,l}.
\label{sd}
\end{equation}
The right hand side converges to $\psi$ in  $L^{2}$ of the spheres and under stronger regularity assumptions the convergence is pointwise.  Let us denote by $\psi_{l}$ the projection of $\psi$ onto  $E^{l}$, i.e. 
\begin{equation*}
\psi_{l}=\sum_{m=-l}^{l}\psi_{m,l}Y^{m,l}.
\end{equation*}
Since  each eigenspace $E^{l}$ is finite dimensional (and so complete) and so closed we have $L^{2}(\mathbb{S}^{2}(r))=E^{l}\oplus \left(E^{l}\right)^{\bot}$. Therefore, $\psi$ can be written uniquely as
$\psi=\psi_{0}+\psi_{\geq 1},$
where $\psi_{\geq 1}\in\left(E^{1}\right)^{\bot}=\displaystyle\oplus_{l\geq 1}E^{l}$. We also have the following
\begin{proposition}(\textbf{Poincar\'e inequality}) If $\psi\in L^{2}\left(\mathbb{S}^{2}(r)\right)$ and $\psi_{l}=0$ for all $l\leq L-1$ for some finite natural number $L$ then we have
\begin{equation*}
\frac{L\left(L+1\right)}{r^{2}}\int_{\mathbb{S}^{2}(r)}{\psi^{2}}\leq \int_{\mathbb{S}^{2}(r)}\left|\nabb\psi\right|^{2}
\label{poincare}
\end{equation*}
and equality holds if and only if $\psi_{l}=0$ for all $l\neq L$.
\label{poin}
\end{proposition}
\begin{proof} 
We have
$\psi=\sum_{l\geq L}{\psi_{l}}.$
Therefore, 
\begin{equation*}
\begin{split}
\int_{\mathbb{S}^{2}(r)}{\left|\nabb\psi\right|^{2}}&=-\displaystyle\int_{\mathbb{S}^{2}(r)}{\psi\cdot\lapp\psi}=\int_{\mathbb{S}^{2}(r)}{\left(\sum_{l'\geq L}{\psi_{l'}}\right)\left(\sum_{l\geq L}{\frac{l\left(l+1\right)}{r^{2}}\psi_{l}}\right)}=\int_{\mathbb{S}^{2}(r)}{\sum_{l\geq L}{\frac{l\left(l+1\right)}{r^{2}}\psi_{l}^{2}}}\\
&\geq\frac{L\left(L+1\right)}{r^{2}}\int_{\mathbb{S}^{2}(r)}{\sum_{l\geq L}{\psi_{l}^{2}}}=\frac{L\left(L+1\right)}{r^{2}}\int_{\mathbb{S}^{2}(r)}{\psi^{2}},
\end{split}
\end{equation*}
where we have used the orthogonality of distinct eigenspaces.
\end{proof}

Returning to our 4-dimensional problem, if $\psi$ is  sufficiently regular in $\mathcal{R}$ then its restriction at each sphere can be written as in \eqref{sd}, where $\psi_{m,l}=\psi_{m,l}(v,r)$ and $Y^{m,l}=Y^{m,l}(\theta,\phi)$. We will not worry about the convergence of the series since we may assume that $\psi$ is sufficiently regular\footnote{Indeed, we may work only with functions which are smooth and such that  the non-zero terms in \eqref{sd} are finitely many. Then, since all of our results are quantitative and all the constants involved do not depend on $\psi$, by a density argument we may lower the regularity of $\psi$ requiring only certain norms depending on the initial data of $\psi$ to be finite.}. First observe that for each summand in \eqref{sd} we have
\begin{equation*}
\Box_{g}\psi_{m,l}Y^{m,l}=\left(S\psi_{m,l}-\frac{l(l+1)}{r^{2}}\psi_{m,l}\right)Y^{m,l},
\end{equation*}
where $S$ is an operator on the quotient $\mathcal{M}/\text{SO(3)}$. Therefore, if $\psi$ satisfies the wave equation then in view of the linear independence of $Y^{m,l}$'s the terms $\psi_{m,l}$ satisfy $S\psi_{m,l}=\frac{l(l+1)}{r^{2}}\psi_{m,l}$ and, therefore, each summand also satisfies the wave equation. This implies that if $\psi_{m,l}=0$ initially then $\psi_{m,l}=0$ everywhere. From now on, \textbf{we will say that the wave $\psi$ is supported on the angular frequencies  $l\geq L$ if $\psi_{i}=0,i=0,...,L-1$ initially (and thus everywhere). Similarly, we will also say that $\psi$ is supported on the angular frequency  $l=L$ if $\psi\in E^{L}$. }

Another important observation is that the operators $\partial_{v}$ and $\partial_{r}$ are endomorphisms of $E^{l}$ for all $l$. Indeed, $\partial_{v}$ commutes with $\lapp$ and if $\psi\in E^{l}$ then
\begin{equation*}
\begin{split}
\lapp\partial_{r}\psi &=\partial_{r}\lapp\psi+\frac{2}{r}\lapp\psi=-\frac{l(l+1)}{r^{2}}\partial_{r}\psi
\end{split}
\end{equation*}
and thus $\partial_{r}\psi\in E^{l}$. Therefore, if $\psi$ is supported on frequencies $l\geq L$ then the same holds for $\partial_{r}\psi$.

\section{The Vector Field $T$}
\label{sec:TheVectorFieldTextbfM}
One can easily see that $\partial_{v}=\partial_{t}=\partial_{t^{*}}$ in the intersection of the corresponding coordinate systems. Here $\partial_{t}$ corresponds to the coordinate basis vector field of either  $\left(t,r\right)$ or $\left(t,r^{*}\right)$. Recall that the region $\mathcal{M}$ where we want to understand the behaviour of waves is covered by the system $\left(v,r,\theta,\phi\right)$. Therefore we define $T=\partial_{v}$. It can be easily seen from \eqref{RN1} that $T$ is Killing vectorfield and timelike everywhere\footnote{Note that in the subextreme range $T$ becomes spacelike in the region bounded by the two horizons, which however coincide in the extreme case.} except on the horizon where it is null.
\subsection{Uniform Boundedness of Degenerate Energy}
\label{sec:TheCurrentKT}
Recall that 
$K^{T}=\textbf{T}_{\mu\nu}\pi_{T}^{\mu\nu},$
where $\pi_{T}^{\mu\nu}$ is the deformation tensor of $T$. Since $T$ is Killing vectorfield, its deformation tensor is zero and so $K^{T}=0$.
Therefore, the  divergence identity in the region $\mathcal{R}(0,\tau)$ gives us the following conservation law
\begin{equation}
\int_{\Sigma_{\tau}}{J^{T}_{\mu}[\psi]n^{\mu}_{\Sigma_{\tau}}}+\int_{\mathcal{H}^{+}}{J^{T}_{\mu}[\psi]n^{\mu}_{{\mathcal{H}^{+}}}}=\int_{\Sigma_{0}}{J^{T}_{\mu}[\psi]n^{\mu}_{\Sigma_{0}}}.
\label{t1}
\end{equation}
Since $T$ is null on the horizon $\mathcal{H}^{+}$ we have that $J_{\mu}^{T}n^{\mu}_{{\mathcal{H}^{+}}}\geq 0$. More presicely, since $n^{\mu}_{{\mathcal{H}^{+}}}=T$, from \eqref{GENERALT} (See Appendix \ref{sec:TheHyperbolicityOfTheWaveEquation1}) we have
$J_{\mu}^{T}n^{\mu}_{{\mathcal{H}^{+}}}=(\partial_{v}\psi)^{2},$
thus proving the following proposition:
\begin{proposition}
For all solutions $\psi$ of the wave equation we have
\begin{equation}
\int_{\Sigma_{\tau}}{J^{T}_{\mu}[\psi]n^{\mu}_{\Sigma_{\tau}}}\leq \int_{\Sigma_{0}}{J^{T}_{\mu}[\psi]n^{\mu}_{\Sigma_{0}}}.
\label{boundT}
\end{equation}
\label{ubdenergy}
\end{proposition}

We know by Proposition \ref{hyperb} that $J^{T}_{\mu}n^{\mu}_{\Sigma_{\tau}}$ is non-negative definite. However, we need to know the exact way $J^{T}_{\mu}n^{\mu}_{\Sigma_{\tau}}$ depends on $d\psi$.  Note  that   $\xi_{n}$ (see Appendix \ref{sec:TheHyperbolicityOfTheWaveEquation1} for the definition of $\xi$) is strictly positive and  uniformly bounded. However
$\xi_{T}=-\frac{1}{2}g(T,T)=\frac{1}{2}D$,
which clearly vanishes (to second order) at the horizon. Therefore, \eqref{GENERALT} implies 
\begin{equation*}
\begin{split}
J^{T}_{\mu}n^{\mu}\, \sim\ \left(\partial_{v}\psi\right)^{2}+D\left(\partial_{r}\psi\right)^{2}+\left|\nabb\psi\right|^{2},
\end{split}
\end{equation*}
where the constants in $\sim$ depend on the mass $M$ and $\Sigma_{0}$. Note that all these relations are invariant under the flow of $T$.

On $\mathcal{H}^{+}$,  the energy estimate \eqref{boundT} degenerates with respect to the transversal derivative $\partial_{r}\psi$. It is exactly this that does not allow us to use estimate $\eqref{boundT}$ to obtain the boundedness result for the waves in the whole region $\mathcal{R}$. However, if we restrict our attention to the region where $r\geq r_{0}>M$  then commuting the wave equation with $T$ and estimate $\eqref{boundT}$ in conjunction with elliptic and Sobolev estimates  give us the boundedness of $\psi$ in this region. This result is not satisfactory since it provides no information about the behaviour of waves on the horizon  and so it is not sufficient for non-linear stability problems.

\section{Morawetz and $X$ Estimates}
\label{sec:TheVectorFieldTextbfX}

In view of the absence of the redshift along $\mathcal{H}^{+}$, any proof of uniform boundedness of solutions to the wave equation essentially relies on the dispersion of waves. Therefore, we must derive estimates which capture these dispersive properties. The first result we prove in this direction is that there exists a constant $C$ which depends on $M$ and $\Sigma_{0}$ such that
\begin{equation*}
\int_{\mathcal{R}}\chi\cdot \left(J_{\mu}^{T}[\psi]n^{\mu}_{\Sigma}+\psi^{2}\right) \leq C\int_{\Sigma_{0}}J_{\mu}^{T}[\psi]n^{\mu}_{\Sigma_{0}},
\end{equation*}
where the weight $\chi=\chi (r)$ degenerates only at $\mathcal{H}^{+}$, at $r=2M$ (photon sphere) and at infinity. The degeneracy at the photon sphere is expected in view of  trapping. In particular, as we shall see, such an estimate degenerates only with respect to the derivatives tangential to the photon sphere. This degeneracy can be overcome at the expense of commuting with $T$. The degeneracy at $\mathcal{H}^{+}$ will be overcome in Section \ref{sec:CommutingWithAVectorFieldTransversalToMathcalH} by imposing necessary conditions on the spherical decomposition of $\psi$.

\subsection{The Vector Field $X$}
\label{sec:TheVectorFieldX}

We are looking for a vector field which gives rise to a current whose divergence is non-negative (upon integration on the spheres of symmetry).  We will be working with vector fields of the form $X=f\left(r^{*}\right)\partial_{r^{*}}$ (for the coordinate system $(t,r^{*})$ see Section \ref{sec:GeometryOfExtremeReissnerNordstromSpacetime}).

\subsubsection{The Spacetime Term $K^{X}$}
\label{sec:TheSpacetimeTermKX}
 We compute
\begin{equation*}
K^{X}=\left(\frac{f'}{2D}+\frac{f}{r}\right)\left(\partial_{t}\psi\right)^{2}+\left(\frac{f'}{2D}-\frac{f}{r}\right)\left(\partial_{r^{*}}\psi\right)^{2}+\left(-\frac{f'+2fH}{2}\right)\left|\nabb\psi\right|^{2},
\end{equation*}
where $f'=\frac{df\left(r^{*}\right)}{dr^{*}}$ and $H=\frac{1}{2}\frac{dD\left(r\right)}{dr}$. Note that all the derivatives in this section will be considered with respect to  $r^{*}$  unless otherwise stated.

Unfortunately, the trapping obstruction does not allow us to obtain a positive definite current $K^{X}$ so easily. Indeed, if we assume that all these coefficients are positive then we have 
\begin{equation*}
\begin{split}
f'>0,\, f<0,\, -\frac{fH}{D}>\frac{f'}{2D}>-\frac{f}{r}\Rightarrow \, \left(-\frac{H}{D}+\frac{1}{r}\right)f>0\Rightarrow\, \left(\frac{H}{D}-\frac{1}{r}\right)<0.
\end{split}
\end{equation*}
However, the quantity $\left(\frac{H}{D}-\frac{1}{r}\right)$ changes sign exactly at the radius $Q$ of the photon sphere. Therefore, there is no way to make all the above four coefficients positive. For simplicity, let us define
\begin{equation*}
\frac{P\left(r\right)}{r^{2}}=-H\cdot r+D=\frac{(r-M)(r-2M)}{r^{2}}.
\end{equation*}

\subsection{The Case $l\geq 1$}
\label{sec:TheCaseL0}
We first consider the case where $\psi$ is supported on the frequencies $l\geq 1$.

\subsubsection{The Currents $J^{X,1}_{\mu}$ and $K^{X,1}$}
\label{sec:TheCurrentJXGMuAndEstimatesForKXG}

To assist in overcoming the obstacle of the photon sphere we shall introduce zeroth order terms in order to modify the coef{}ficient of $\left(\partial_{t}\psi\right)^{2}$ and  create another which is more flexible. Let us consider the current
\begin{equation*}
J_{\mu}^{X,h_{1},h_{2},w}=J_{\mu}^{X}+h_{1}\left(r\right)\psi\nabla _{\mu}\psi+h_{2}\left(r\right)\psi^{2}\left(\nabla_{\mu}w\right),
\end{equation*}
where $h_{1},h_{2},w$ are functions on $\mathcal{M}$. Then we have
\begin{equation*}
\begin{split}
\tilde{K}^{X}&\!=\!\nabla^{\mu}J_{\mu}^{X,h_{1},h_{2},w}=K^{X}+\nabla^{\mu}\left(h_{1}\psi\nabla _{\mu}\psi\right)+\nabla^{\mu}\left(h_{2}\psi^{2}\left(\nabla_{\mu}w\right)\right)\\
&\!=\!K^{X}\!+h_{1}\left(\nabla^{a}\psi\nabla_{a}\psi\right)\!+\left(\nabla^{\mu}h_{1}+2h_{2}\nabla^{\mu}w\right)\psi\nabla_{\mu}\psi+\left(\nabla_{\mu}h_{2}\nabla^{\mu}w+h_{2}\left(\Box_{g} w\right)\right)\psi^{2}.\\
\end{split}
\end{equation*}
By taking $h_{2}=1,h_{1}=2G,w=-G$ we make the coefficient of $\psi\nabla_{\mu}\psi$ vanish. Therefore, let us define
\begin{equation}
J_{\mu}^{X,1}\overset{.}{=}J^{X}_{\mu}+2G\psi\left(\nabla_{\mu}\psi\right)-\left(\nabla_{\mu}G\right)\psi^{2}.
\label{modG}
\end{equation}
Then
\begin{equation*}
\begin{split}
K^{X,1}&\overset{.}{=}\nabla^{\mu}J_{\mu}^{X,1}=K^{X}+2G\left(\left(-\frac{1}{D}\right)\left(\partial_{t}\psi\right)^{2}+\left(\frac{1}{D}\right)\left(\partial_{r^{*}}\psi\right)^{2}+\left|\nabb\psi\right|^{2}\right)-\left(\Box_{g} G\right)\psi^{2}.\\
\end{split}
\end{equation*}
Therefore, if we take $G$ such that 
\begin{equation}
\begin{split}
\left(2G\right)\left(-\frac{1}{D}\right)=-\left(\frac{f'}{2D}+\frac{f}{r}\right)\Rightarrow G=\frac{f'}{4}+\frac{f\cdot D}{2r}
\label{G}
\end{split}
\end{equation}
then
\begin{equation*}
\begin{split}
K^{X,1}&=\frac{f'}{D}\left(\partial_{r^{*}}\psi\right)^{2}+\frac{f\cdot P}{r^{3}}\left|\nabb\psi\right|^{2}-\left(\Box_{g} G\right)\psi^{2}\\
&=\frac{f'}{D}\left(\partial_{r^{*}}\psi\right)^{2}+\frac{f\cdot P}{r^{3}}\left|\nabb\psi\right|^{2}-\left(\frac{1}{4D}f'''+\frac{1}{r}f''+\frac{D'}{D\cdot r}f'+\left(\frac{D''}{2D\cdot r}-\frac{D'}{2r^{2}}\right)f\right)\psi^{2}
\end{split}
\end{equation*}
Note that, since $f$ must be bounded and $f'$ positive, $-f'''$ must become negative and therefore has the wrong sign. Note also the factor $D$ at the denominator which degenerates at the horizon. The way that turns out to work is by borrowing from the coefficient of $\left(\partial_{r^{*}}\psi\right)^{2}$ which is accomplished by introducing a third current. Note that this current exploits an algebraic property of the trapping in the extreme case (see the computation \eqref{icomp} of the term $I$ below).

\subsubsection{The Current $J^{X,2}_{\mu}$ and Estimates for $K^{X,2}$}
\label{sec:TheCurrentJX2MuAndEstimatesForKX2}

We define
\begin{equation*}
\begin{split}
&J_{\mu}^{X,2}[\psi]=J_{\mu}^{X,1}[\psi]+\frac{f'}{D\cdot f}\beta \psi^{2} X_{\mu},
\end{split}
\end{equation*}
where $\beta$ will be a function of $r^{*}$ to be defined below.  Since
$\operatorname{Div}\left(\partial_{r^{*}}\right)=\frac{D'}{D}+\frac{2D}{r},
$ we have
\begin{equation*}
\begin{split}
K^{X,2}=&K^{X,1}+Div\left(\frac{f'}{D}\beta\psi^{2}\partial_{r^{*}}\right)\\
=&\frac{f'}{D}\left(\partial_{r^{*}}\psi+\beta\psi\right)^{2}+\frac{f\cdot P}{r^{3}}\left|\nabb{\psi}\right|^{2}+\\
&+\left[-\frac{1}{4}\frac{f'''}{D}+\left(\beta-\frac{D}{r}\right)\frac{f''}{D}+\left(\beta'-\beta^{2}+\frac{2D}{r}\beta-\frac{D'}{r}\right)\frac{f'}{D}+\left(\frac{D'}{2r^{2}}-\frac{D''}{2D\cdot r}\right)f\right]\psi^{2}.
\end{split}
\end{equation*}
Note that the coefficient of $f\psi^{2}$ is independent of the choice of the function $\beta$. Let us now take $\beta=\frac{D}{r}-\frac{x}{\a^{2}+x^{2}},$
where $x=r^{*}-\a-\sqrt{\a}\ $ and $\a>0$ a sufficiently large number to be chosen appropriately. As we shall see, the reason for introducing this shifted coordinate $x$ is that we want the origin $x=0$ to be far away from the photon sphere. 
We clearly need to choose a function $f$ that is stricly increasing and changes sign at the photon sphere. So, if we choose this function (which we call $f^{\a}$) such that $\left(f^{\a}\right)'=\frac{1}{\a^{2}+x^{2}}, \ f^{\a}\left(r^{*}=0\right)=0,$
then
\begin{equation*}
\begin{split}
F:=-\frac{1}{4}\frac{\left(f^{\a}\right)'''}{D}-\frac{x}{\a^{2}+x^{2}}\frac{\left(f^{\a}\right)''}{D}-\frac{\a^{2}}{\left(\a^{2}+x^{2}\right)^{2}}\frac{\left(f^{\a}\right)'}{D}=\frac{1}{2D}\frac{x^{2}-\a^{2}}{\left(x^{2}+\a^{2}\right)^{3}}.
\end{split}
\end{equation*}
If we define\footnote{Note that the role of the parameter $\a$ as an index in our notation in this section is to emphasise that the corresponding tensors (i.e. functions, vector fields, etc.) depend on the choice of $\a$. } $X^{\a}=f^{\a}\partial_{r^{*}}$ and $I=\left(\frac{D'}{2r^{2}}-\frac{D''}{2D\cdot r}\right)f^{\a}$ then
\begin{equation*}
\begin{split}
K^{X^{\a},2}[\psi]=\frac{\left(f^{\a}\right)'}{D}\left(\partial_{r^{*}}\psi+\beta\psi\right)^{2}+\frac{f^{\a}\cdot P}{r^{3}}\left|\nabb{\psi}\right|^{2}+\left(F+I\right)\psi^{2}.
\end{split}
\end{equation*}

\subsubsection{Nonnegativity of $K^{X^{\a},2}$}
\label{sec:NonnegativityOfK2}
 We have the following important proposition
\begin{proposition}
There exists a constant $C$ which depends only on $M$ such that for all solutions $\psi$ to the wave equation which are supported on the frequencies  $l\geq 1$  we have
\begin{equation}
\begin{split}
\int_{\mathbb{S}^{2}}{\left(\frac{(r-M)\left(r-2M\right)^{2}}{r^{4}}\left|\nabb\psi\right|^{2}+\frac{1}{D}\frac{1}{(\left(r^{*}\right)^{2}+1)^{2}}\psi^{2}\right)}\leq C\int_{\mathbb{S}^{2}}{K^{X^{\a},2}[\psi]}.
\label{Xa2psi}
\end{split}
\end{equation}
\label{xa2prop}
\end{proposition}
\begin{proof}
We have
\begin{equation*}
\begin{split}
\frac{f^{\a}\cdot P}{r^{3}}\left|\nabb{\psi}\right|^{2}+\left(F+I\right)\psi^{2}\leq K^{X^{\a},2}.
\end{split}
\end{equation*}
We compute the term $I$ 
\begin{equation}
\begin{split}
I=\left(\frac{D'}{2r^{2}}-\frac{D''}{2D\cdot r}\right)f^{\a}=\left(1-\frac{M}{r}\right)\frac{3M}{r^{6}}f^{\a}\cdot P.
\end{split}
\label{icomp}
\end{equation}
Therefore, $I\geq 0$ and vanishes (to second order) at the horizon and the photon sphere. Note now that the term $F$ is positive whenever  $x$ is not in the interval $\left[-\a, \a\right]$. On the other hand, the photon sphere is far away from the region where $x\in\left[-\a, \a\right]$, so $I$ is positive there. However, $I$ behaves like $\frac{1}{r^{4}}$ and so it is not sufficient to compensate the negativity of $F$. That is why we need to borrow from the coefficient of $\nabb\psi$ which behaves like $\frac{1}{r}$. Indeed, the Poincar\'e inequality gives us
$\int_{\mathbb{S}^{2}}{\frac{2}{r^{2}}\psi^{2}}\leq\int_{\mathbb{S}^{2}}{\left|\nabb\psi\right|^{2}}$ and, therefore, it suffices to prove that $
2\frac{f^{\a}\cdot P}{r^{5}}+F> 0$
for all $x\in\left[-\a,\a\right]$ or, equivalently, $r^{*}\in\left[\sqrt{\a}, 2\a+\sqrt{\a}\right]$. But
\begin{equation*}
\begin{split}
F=\frac{1}{D}\frac{x^{2}-\a^{2}}{2\left(x^{2}+\a^{2}\right)^{3}}=\Delta_{\a}\frac{x^{2}-\a^{2}}{2\left(x^{2}+\a^{2}\right)^{3}}
\end{split}
\end{equation*}
with $1\leq\Delta_{\a}\rightarrow 1$ as $\a\rightarrow +\infty$. Moreover, since $r^{*}\geq r$ for big $r$, by taking $\a$ sufficiently big we have
\begin{equation*}
\begin{split}
2\frac{f^{\a}\cdot P}{r^{5}}\geq 2\delta_{\a}\frac{f^{\a}}{\left(r^{*}\right)^{3}}
\end{split}
\end{equation*}
with $1\geq\delta_{\a}\rightarrow 1$ as $\a\rightarrow +\infty$. Therefore, we need to establish that 
\begin{equation*}
\begin{split}
\Pi :=\frac{\left(\a^{2}-x^{2}\right)\left(x+\a+\sqrt{\a}\right)^{3}}{4\left(x^{2}+\a^{2}\right)^{3}f^{\a}}<\frac{\delta_{\a}}{\Delta_{\a}}.
\end{split}
\end{equation*}
for all $x\in\left[-\a,\a\right]$. In view of the asymptotic behaviour of the constants $\Delta_{\a},\delta_{\a}$ it suffices to prove that the right hand side of the above inequality is strictly less than 1. 
\begin{lemma}
Given the function $f^{\a}\left(r^{*}\right)=f^{\a}\left(r^{*}\left(x\right)\right)$ defined above, we have the following: If $\ -\a\leq x\leq 0$ then
$f^{\a}>\frac{x+\a}{2\a^{2}}.$ Also, if $\ 0\leq x \leq \a$ then 
$f^{\a}>\frac{x+\a}{x^{2}+\a^{2}}-\frac{1}{2\a}.$
\label{lemmamefinetunning}
\end{lemma}
\begin{proof}
For $\ -\a\leq x\leq 0$ we have
\begin{equation*}
\begin{split}
f^{\a}\left(x\right)=\int_{-\a-\sqrt{\a}}^{x}{\frac{1}{\tilde{x}^{2}+\a^{2}}d\tilde{x}}>\int_{-\a}^{x}{\frac{1}{\tilde{x}^{2}+\a^{2}}d\tilde{x}}>\int_{-\a}^{x}{\frac{1}{2\a^{2}}d\tilde{x}}=\frac{x+\a}{2\a^{2}}.
\end{split}
\end{equation*}
Now for $\ 0\leq x \leq \a$ we have
\begin{equation*}
\begin{split}
f^{\a}\left(x\right)&>\int_{-\a}^{x}{\frac{1}{\tilde{x}^{2}+\a^{2}}d\tilde{x}}>\int_{-\a}^{-x}{\frac{1}{\tilde{x}^{2}+\a^{2}}d\tilde{x}}+\int_{-x}^{x}{\frac{1}{x^{2}+\a^{2}}d\tilde{x}}\\&=\int_{-\a}^{x}{\frac{1}{x^{2}+\a^{2}}d\tilde{x}}-\int_{-\a}^{-x}{\left(\frac{1}{x^{2}+\a^{2}}-\frac{1}{\tilde{x}^{2}+\a^{2}}\right)d\tilde{x}}>\frac{x+a}{x^{2}+\a^{2}}-\frac{1}{2a}.
\end{split}
\end{equation*}
\end{proof}
Consequently, if $\ -\a\leq x\leq 0$ and $x=-\lambda\a$ then $\ 0\leq\lambda\leq 1$ and
\begin{equation*}
\begin{split}
\Pi<&\frac{\left(\a-x\right)\left(x+\a+\sqrt{\a}\right)^{3}\a^{2}}{2\left(x^{2}+\a^{2}\right)^{3}}=\frac{\left(1-\lambda\right)\left(1+\lambda+\a^{-\frac{1}{2}}\right)^{3}}{2\left(\lambda^{2}+1\right)^{3}}=d_{\a}\frac{\left(1-\lambda\right)\left(1+\lambda\right)^{3}}{2\left(\lambda^{2}+1\right)^{3}}<d_{\a}\frac{2}{3}<\frac{9}{10},
\end{split}
\end{equation*}
since $\ d_{\a}\rightarrow 1$ as $\a\rightarrow +\infty$. Similarly, if $\ 0\leq x\leq \a$ and $x=\lambda\a$ then $\ 0\leq\lambda\leq 1$ and
\begin{equation*}
\begin{split}
\Pi&<\frac{\a\left(\a^{2}-x^{2}\right)\left(x+\a+\sqrt{\a}\right)^{3}}{2\left(x^{2}+\a^{2}\right)^{2}\left(\a^{2}+2\a x-x^{2}\right)}<\frac{\a\left(\a^{2}-x^{2}\right)\left(x+\a+\sqrt{\a}\right)^{3}}{2\left(x^{2}+\a^{2}\right)^{3}}=\tilde{d_{\a}}\frac{\left(1-\lambda^{2}\right)\left(1+\lambda\right)^{3}}{2\left(1+\lambda^{2}\right)^{3}}<\frac{9}{10},
\end{split}
\end{equation*}
since $\tilde{d_{\a}}\rightarrow 1$ as $\a\rightarrow +\infty$.
\end{proof}
Note, however, that although the coef{}ficient of $\psi$ does not degenerate on the photon sphere, the coef{}ficient of the angular derivatives vanishes at the photon sphere to second order. Having this estimate for $\psi$, we  obtain estimates for its derivatives (in Section \ref{sec:SpacetimeL2EstimateForPsi} we derive a similar estimate for $\psi$ for the case $l=0$).

\begin{proposition}
There exists a positive constant $C$ which depends only on $M$ such that for all solutions $\psi$ of the wave equation which are supported on the frequencies $l\geq 1$ we have
\begin{equation*}
\begin{split}
\int_{\mathbb{S}^{2}}{\!\left(\frac{(r-M)}{r^{4}}(\partial_{t}\psi)^{2}+\frac{(r-M)}{r^{4}}\psi^{2}\right)}\leq C\int_{\mathbb{S}^{2}}{\sum_{i=0}^{1}{K^{X^{\a},2}\left[T^{i}\psi\right]}}.
\end{split}
\end{equation*}
\label{mesaio}
\end{proposition}
\begin{proof}
From \eqref{Xa2psi} we have that the coef{}ficient of $\psi^{2}$ does not degenerate at the photon sphere. The weights at $\mathcal{H}^{+}$ and infinity are given by the Poincar\'e inequality. Commuting the wave equation with $T$ completes the proof of the proposition.
\end{proof}

\subsubsection{The Lagrangian Current $L_{\mu}^{f}$}
\label{sec:TheAuxiliaryCurrentsAnd}
In order to retrieve the remaining derivatives we consider the Lagrangian\footnote{The name Lagrangian comes from the fact that if $\psi$ satisfies the wave equation and $\mathcal{L}$ denotes the Lagrangian that corresponds to the wave equation, then $\mathcal{L}(\psi,d\psi,g^{-1})=g^{\mu\nu}\partial_{\mu}\psi\partial_{\nu}\psi=\operatorname{Div}(\psi\nabla_{\mu}\psi)$} current
\begin{equation*}
\begin{split}
L_{\mu}^{f}=f\psi\nabla_{\mu}\psi,
\end{split}
\end{equation*}
where $f=\frac{1}{r^{3}}D^{\frac{3}{2}}.$
\begin{proposition}
There exists a positive constant $C$ which depends only on $M$ such that for all solutions $\psi$ of the wave equation which are supported on the frequencies $l\geq 1$ we have
\begin{equation*}
\begin{split}
\int_{\mathbb{S}^{2}}{\left(\frac{\sqrt{D}}{r^{3}}\left(\partial_{t}\psi\right)^{2}+\frac{\sqrt{D}}{2r^{3}}\left(\partial_{r^{*}}\psi\right)^{2}+\frac{D^{3/2}}{r^{3}}\left|\nabb\psi\right|^{2}\right)}
\leq \int_{\mathbb{S}^{2}}{\left(\operatorname{Div}(L_{\mu}^{f})+C\sum_{i=0}^{1}{K^{X^{\a},2}\left[T^{i}\psi\right]}\right)}.
\label{rparagogos}
\end{split}
\end{equation*}
\label{rparagogosprop}
\end{proposition}
\begin{proof}
We have
\begin{equation*}
\begin{split}
\operatorname{Div}(L_{\mu}^{f})=&f\nabla^{\mu}\psi\nabla_{\mu}\psi +\nabla^{\mu}f\psi\nabla_{\mu}\psi \\
\geq &\frac{\sqrt{D}}{r^{3}}\left(\partial_{r^{*}}\psi\right)^{2}+\frac{D^{3/2}}{r^{3}}\left|\nabb\psi\right|^{2}-\frac{\sqrt{D}}{r^{3}}\left(\partial_{t}\psi\right)^{2}-\frac{1}{\epsilon}\frac{D}{r^{4}}\psi^{2}-\epsilon\frac{D}{r^{4}}(\partial_{r^{*}}\psi)^{2},
\end{split}
\end{equation*}
where $\epsilon >0$ is such that $\epsilon\sqrt{D}<\frac{r}{2}$. Therefore, in view of Proposition \ref{mesaio}, we have the required result.
\end{proof}

\begin{proposition}
There exists a positive constant $C$ which depends only on $M$ such that for all solutions $\psi$ of the wave equation which are supported on the frequencies $l\geq 1$ we have
\begin{equation*}
\begin{split}
&\int_{\mathbb{S}^{2}}{\left(\frac{\sqrt{D}}{r^{3}}\left(\partial_{t}\psi\right)^{2}+\frac{\sqrt{D}}{2r^{3}}\left(\partial_{r^{*}}\psi\right)^{2}+\frac{\sqrt{D}}{r}\left|\nabb\psi\right|^{2}+\frac{\sqrt{D}}{r^{3}}\psi^{2}\right)}
\leq \int_{\mathbb{S}^{2}}{\left(\operatorname{Div}(L_{\mu}^{f})+C\sum_{i=0}^{1}{K^{X^{\a},2}\left[T^{i}\psi\right]}\right)}.
\label{olaparagogos}
\end{split}
\end{equation*}
\label{finspacetimeprop}
\end{proposition}
\begin{proof}
Immediate from Propositions \ref{xa2prop} and \ref{rparagogosprop}.
\end{proof}

\subsubsection{The Current $J_{\mu}^{X^{d},1}$}
\label{sec:TheCurrentJMuXD1}

In case we allow some degeneracy at the photon sphere  we can  obtain similar estimates without commuting the wave equation with $T$. This will be very useful for estimating error terms in spacetime regions which do not contain the photon sphere. We define the function $f^{d}$ such that 
\begin{equation*}
\begin{split}
\left(f^{d}\right)'=\frac{1}{(r^{*})^{2}+1},\ \ \ \ f^{d}(r^{*}=0)=0.
\end{split}
\end{equation*}
\begin{proposition}
There exists a positive constant $C$ which depends only on $M$ such that for all solutions $\psi$ of the wave equation which are supported on the frequencies $l\geq 1$ we have
\begin{equation*}
\begin{split}
\frac{1}{C}\int_{\mathbb{S}^{2}}{\left(\frac{1}{r^{2}}(\partial_{r^{*}}\psi)^{2}+\frac{P\cdot (r-2M)}{r^{4}}\left|\nabb\psi\right|^{2}\right)}\leq \int_{\mathbb{S}^{2}}{\left(K^{X^{d},1}[\psi]+CK^{X^{\a},2}[\psi]\right)},
\end{split}
\end{equation*}
where $X^{d}=f^{d}\partial_{r^{*}}$ and the current $J_{\mu}^{X,1}$ is as defined in Section \ref{sec:TheCurrentJXGMuAndEstimatesForKXG}.
\label{nocomm1prop}
\end{proposition}
\begin{proof}
Note that the coef{}ficient of $\psi^{2}$ in $K^{X^{d},1}$ vanishes to first order on $\mathcal{H}^{+}$ (see also Lemma \ref{rnrstar}) and behaves like $\frac{1}{r^{4}}$ for large $r$. Note also that the coefficient of $(\partial_{r^{*}}\psi)^{2}$ converges to $M^{2}$ (see again Lemma \ref{rnrstar}). The result now follows from Proposition \ref{xa2prop}.
\end{proof}
In order to retrieve the $\partial_{t}$-derivative we introduce the current
$L_{\mu}^{h^{d}}=h^{d}\psi\nabla_{\mu}\psi,$ where $h^{d}$ is such that:
\begin{equation*}
\begin{split}
& h^{d}=-\frac{1}{(r^{*})^{2}+1} \text{ for } M\leq r\leq r_{0}<2M,\ h^{d}< 0 \text{ for } r_{0}< r <2M,\\
&h^{d}=0 \text{ for } r=2M \text{ to second order},\  h^{d}<0 \text{ for } 2M<r\leq r_{1},\ h^{d}=-\frac{1}{r^{2}} \text{ for } r_{1}\leq r.\\
\end{split}
\end{equation*}
\begin{proposition}
There exists a positive constant $C$ which depends only on $M$ such that for all solutions $\psi$ of the wave equation which are supported on the frequencies $l\geq 1$ we have
\begin{equation*}
\begin{split}
&\frac{1}{C}\int_{\mathbb{S}^{2}}{\left(\frac{P\cdot (r-2M)}{r^{5}}(\partial_{t}\psi)^{2}+\frac{1}{r^{2}}(\partial_{r^{*}}\psi)^{2}+\frac{P\cdot (r-2M)}{r^{4}}\left|\nabb\psi\right|^{2}\right)}\\&\ \ \ \ \ \ \ \ \ \ \ \ \ \ \leq \int_{\mathbb{S}^{2}}{\left(\operatorname{Div}(L_{\mu}^{h^{d}})+K^{X^{d},1}[\psi]+CK^{X^{\a},2}[\psi]\right)}.
\end{split}
\end{equation*}
\label{nocomm2prop}
\end{proposition}
\begin{proof}
We have as before
\begin{equation*}
\begin{split}
\operatorname{Div}(L_{\mu}^{h^{d}})=-\frac{h^{d}}{D}(\partial_{t}\psi)^{2}+\frac{h^{d}}{D}(\partial_{r^{*}}\psi)^{2}+h^{d}\left|\nabb\psi\right|^{2}+\frac{(h^{d})'}{D}\psi(\partial_{r^{*}}\psi).\\
\end{split}
\end{equation*}
Since the coef{}ficient of $\psi(\partial_{r^{*}}\psi)$ vanishes to first on the horizon, the Cauchy-Schwarz inequality and Proposition \ref{nocomm1prop} imply the result.
\end{proof}
Finally, we obtain:
\begin{proposition}
There exists a positive constant $C$ which depends only on $M$ such that for all solutions $\psi$ of the wave equation which are supported on the frequencies $l\geq 1$ we have
\begin{equation*}
\begin{split}
&\frac{1}{C}\int_{\mathbb{S}^{2}}{\left(\frac{P\cdot (r-2M)}{r^{5}}(\partial_{t}\psi)^{2}+\frac{1}{r^{2}}(\partial_{r^{*}}\psi)^{2}+\frac{P\cdot (r-2M)}{r^{4}}\left|\nabb\psi\right|^{2}+\frac{(r-M)}{r^{4}}\psi^{2}\right)}\\&\ \ \ \ \ \ \ \ \ \ \ \ \ \ \leq \int_{\mathbb{S}^{2}}{\left(\operatorname{Div}(L_{\mu}^{g^{d}})+K^{X^{d},1}[\psi]+CK^{X^{\a},2}[\psi]\right)}.
\end{split}
\end{equation*}
\label{nocomm3prop}
\end{proposition}
\begin{proof}
Immediate from Propositions \ref{xa2prop} and \ref{nocomm2prop}.
\end{proof}

\subsection{The Case $l=0$}
\label{sec:TheCaseLO}

In case $l=0$ the wave is not trapped. Indeed, if we consider the vector field $X^{0}=f^{0}\partial_{r^{*}}$ with $f^{0}=-\frac{1}{r^{3}}$  then we have: 
\begin{proposition}
For all spherically symmetric solutions $\psi$ of the wave equation  we have 
\begin{equation*}
\frac{1}{r^{4}}(\partial_{t}\psi)^{2}+\frac{5}{r^{4}}(\partial_{r^{*}}\psi)^{2}=K^{X^{0}}.
\end{equation*}
\label{1l0}
\end{proposition}
\begin{proof}
Immediate from the expression of $K^{X}$ and the above choice of $f=f^{0}$.
\end{proof}

\subsection{The Boundary Terms}
\label{sec:TheBoundaryTerms}

The positivity of the bulk terms shown so far will be useful only if we can control the arising boundary terms.
\subsubsection{Estimates for $J_{\mu}^{X}n^{\mu}_{S}$}
\label{sec:EstimatesForJMuXNMuSigma}

\begin{proposition}
Let $X=f\partial_{r^{*}}$ where $f=f(r^{*})$ is bounded and  $S$ be either a $SO(3)$ invariant spacelike (that may cross $\mathcal{H}^{+}$) or a $SO(3)$ invariant null hypersurface. Then there exists a uniform constant $C$ that depends  on $M$, $S$ and the function $f$ such that for all $\psi$ we have
\begin{equation*}
\begin{split}
\left|\int_{S}{J^{X^{i}}_{\mu}[\psi]n^{\mu}_{S}}\right|\leq C\int_{S}{J_{\mu}^{T}[\psi]n^{\mu}_{S}}.
\end{split}
\end{equation*}
\label{boun1}
\end{proposition}
\begin{proof}

We work using the coordinate system $\left(v,r,\theta,\phi\right)$. First note that 
$X=f\partial_{r^{*}}=f\partial_{v}+f\cdot D\partial_{r}.$
Now
\begin{equation*}
\begin{split} 
J^{X}_{\mu}n^{\mu}_{S}&=\textbf{T}_{\mu\nu}(X)^{\nu}n^{\mu}_{S}=\textbf{T}_{\mu v}fn^{\mu}_{S}+\textbf{T}_{\mu r}f D n^{\mu}_{S}\\
&=\left(fn^{v}\right)\left(\partial_{v}\psi\right)^{2}+D f n^{v}\left(\partial_{v}\psi\right)\left(\partial_{r}\psi\right)+D\left[\frac{1}{2}Dfn^{v}-\frac{1}{2}fn^{r}-\frac{1}{2}Dfn^{v}+fn^{r}\right]\left(\partial_{r}\psi\right)^{2}\\
&\ \ \ \ +\left[\frac{1}{2}Dfn^{v}-\frac{1}{2}Dfn^{v}-\frac{1}{2}fn^{r}\right]\left|\nabb\psi\right|^{2}.
\end{split}
\end{equation*}
The result now follows from the boundedness of $f$.
\end{proof}

\subsubsection{Estimates for $J_{\mu}^{X^{\a},i}\,n_{S}^{\mu},\, i=1,2$}
\label{sec:EstimatesForJMuXA2NSigmaMu}

\begin{proposition}
There exists a uniform constant $C$ that depends  on $M$ and $S$ such that
\begin{equation*}
\begin{split}
\left|\int_{S}{J^{X^{\a},i}_{\mu}[\psi]n^{\mu}_{S}}\right|\leq C\int_{S}{J_{\mu}^{T}[\psi]n^{\mu}_{S}},\, i=1,2,
\end{split}
\end{equation*}
where $S$ is as in Proposition \ref{boun1}.
\label{boun2}
\end{proposition}
\begin{proof}
It suffices to prove
\begin{equation*}
\begin{split}
\left|\int_{S}{\left(2G^{\a}\psi\left(\nabla_{\mu}\psi\right)-\left(\nabla_{\mu}G^{\a}\right)\psi^{2}+\left(\frac{(f^{\a})'}{D}\beta (\partial_{r^{*}})_{\mu}\right)\psi^{2}\right)n^{\mu}_{S}}\,\right|\leq B\int_{S}{J_{\mu}^{T}[\psi]n^{\mu}_{S}}.
\end{split}
\end{equation*}
For we first prove the following lemma that is true only in the case of  extreme Reissner-Nordstr\"{o}m (and not in the subextreme range).
\begin{lemma}
The function
\begin{equation*}
\begin{split}
F=\frac{1}{D}\frac{1}{((r^{*})^{2}+1)}
\end{split}
\end{equation*}
is bounded in $\mathcal{R}\cup\mathcal{H}^{+}$.
\label{rnrstar}
\end{lemma}
\begin{proof}
For the tortoise coordinate $r^{*}$ we have
$r^{*}(r)=r+2M\ln(r-M)-\frac{M^{2}}{r-M}+C.$
Clearly,  $F\rightarrow 0$ as $r\rightarrow +\infty$. Moreover, in a neighbourhood of $\mathcal{H}^{+}$
\begin{equation*}
\begin{split}
F\sim\frac{r^{2}}{(r-M)^{2}}\frac{1}{\left[\frac{M^{4}}{(r-M)^{2}}+r^{2}+4M^{2}(\ln(r-M))^{2}\right]}\rightarrow M^{2}<\infty.
\end{split}
\end{equation*}
This implies the required result.
\end{proof}
An immediate collorary of this lemma is that the functions
$G_{1}^{\a}=r\frac{G^{\a}}{D} \text{  and  } \frac{(f^{\a})'}{D}$
are bounded in $\mathcal{R}\cup\mathcal{H}^{+}$.
Furthermore, for $M\ll r$ we have $(f^{\a})'\sim\frac{1}{r^{2}}$, $D\sim1$ and $\partial_{r}D\sim\frac{1}{r^{2}}$.
Therefore,
$G_{2}^{\a}=r^{2}\nabla_{r}G^{\a}$
is also bounded. Finally, $\nabla_{v}G^{\a}=0$.
The above bounds and the first Hardy inequality complete the proof of the proposition. 
\end{proof}

\subsection{A Degenerate $X$ Estimate}
\label{sec:ADegenateXEstimate}

We first obtain an estimate which does not lose derivatives but degenerates at the photon sphere.

\begin{theorem}
There exists a constant $C$ which depends  on $M$ and $\Sigma_{0}$ such that for all solutions $\psi$ of the wave equation we have
\begin{equation}
\begin{split}
\int_{\mathcal{R}(0,\tau)}\!\!\!{\left(\frac{1}{r^{4}}(\partial_{r^{*}}\psi)^{2}+\frac{(r-M) (r-2M)^{2}}{r^{7}}\left((\partial_{t}\psi)^{2}+\left|\nabb\psi\right|^{2}\right)\right)} \leq C\int_{\Sigma_{0}}{J_{\mu}^{T}[\psi]n^{\mu}_{\Sigma_{0}}}.
\end{split}
\label{degX}
\end{equation}
\label{degXprop}
\end{theorem}
\begin{proof}
We first decompose $\psi$  as 
\begin{equation*}
\psi=\psi_{\geq 1}+\psi_{0}.
\end{equation*}
We apply Stokes' theorem for the current 
\begin{equation*}
J_{\mu}^{d}[\psi_{\geq 1}]=J_{\mu}^{X^{d},1}[\psi_{\geq 1}]+J_{\mu}^{X^{\a},2}[\psi_{\geq 1}]+L_{\mu}^{g^{d}}[\psi_{\geq 1}]
\end{equation*}
in the spacetime region $\mathcal{R}(0,\tau)$ and use Propositions \ref{nocomm3prop} and \ref{boun2}. We also apply Stokes' theorem in $\mathcal{R}(0,\tau)$ for the current $J_{\mu}^{X^{0}}[\psi_{0}]$ and use Proposition \ref{boun1} and by adding these two estimates we obtain the required result.

\end{proof}
\begin{remark} One can further improve the weights at infinity. Indeed, if we apply Stokes' theorem in the region $\left\{r\geq R\right\}$ for sufficiently large $R$ for the current
\begin{equation*}
J_{\mu}=J_{\mu}^{X^{f_{1}},1}+J_{\mu}^{X^{f_{{2}}}},
\end{equation*}
where $f_{1}=1-\frac{1}{r^{\delta}},\ f_{2}=\frac{\delta}{2+\delta}\frac{1}{r^{\delta}} \text{ and } \delta >0,$
then an easy calculation shows that there exists a constant $C_{\delta}>0$ such that for $r\geq R$
\begin{equation*}
\nabla^{\mu}J_{\mu}\geq C_{\delta}\left(r^{-1-\delta}(\partial_{t}\psi)^{2}+r^{-1-\delta}(\partial_{r^{*}}\psi)^{2}+r^{-1}\left|\nabb\psi\right|^{2}+r^{-3-\delta}\psi^{2}\right).
\end{equation*}
\label{infinity}
\end{remark}
The above remark and theorem and  the identity $\partial_{r^{*}}=\partial_{v}+D\partial_{r}$ imply statement (2) of Theorem \ref{th1} of Section \ref{sec:TheMainTheorems},

\subsection{A non-Degenerate $X$ Estimate}
\label{sec:ANonDegenerateXEstimate}

We  derive an $L^{2}$ estimate which does not degenerate at the photon sphere but requires higher regularity for $\psi$.
\begin{theorem}
There exists a constant $C$ which depends  on $M$ and $\Sigma_{0}$ such that for all solutions $\psi$ of the wave equation we have
\begin{equation}
\begin{split}
&\int_{\mathcal{R}(0,\tau)}\!\!\!{\left(\frac{\sqrt{D}}{r^{4}}\left(\partial_{t}\psi\right)^{2}+\frac{\sqrt{D}}{2r^{4}}\left(\partial_{r^{*}}\psi\right)^{2}+\frac{\sqrt{D}}{r}\left|\nabb\psi\right|^{2}\right)}
\leq C\int_{\Sigma_{0}}{\left(J_{\mu}^{T}[\psi]n^{\mu}_{\Sigma_{0}}+J_{\mu}^{T}[T\psi]n^{\mu}_{\Sigma_{0}}\right)}.
\label{x}
\end{split}
\end{equation}
\label{xtheo}
\end{theorem}
\begin{proof}
We again decompose  $\psi$ as 
$\psi=\psi_{\geq 1}+\psi_{0}$
and apply Stokes' theorem for the current 
\begin{equation*}
J_{\mu}[\psi_{\geq 1}]=L_{\mu}^{f}[\psi_{\geq 1}]+CJ_{\mu}^{X^{\a},2}[\psi_{\geq 1}]+CJ_{\mu}^{X^{\a},2}[T\psi_{\geq 1}]
\end{equation*}
in the spacetime region $\mathcal{R}(0,\tau)$, where $C$ is the constant of Proposition \ref{finspacetimeprop}, and use Propositions \ref{finspacetimeprop} and \ref{boun2}. Finally, we  apply Stokes' theorem in $\mathcal{R}(0,\tau)$ for the current $J_{\mu}^{X^{0}}[\psi_{0}]$ and use Proposition \ref{boun1}. Adding these two estimates completes the proof.
\end{proof}
The proof of statement (3) of Theorem \ref{th1} of Section \ref{sec:TheMainTheorems} is immediate from the above proposition, Theorem \ref{degXprop} and Remark \ref{infinity}.  
\begin{remark}
All the above estimates hold if $\Sigma_{\tau}$ is replaced by $\tilde{\Sigma}_{\tau}$. Indeed, the results of Section \ref{sec:TheBoundaryTerms}  allow us to bound the corresponding boundary terms in this case. Moreover, the additional boundary term on $\mathcal{I}^{+}$ can be bounded by the conserved $T$-flux.
\label{2ndkindx}
\end{remark}

\subsection{Zeroth Order Morawetz Estimate for $\psi$}
\label{sec:SpacetimeL2EstimateForPsi}

We now prove weighted $L^{2}$ estimates of the wave $\psi$ itself. In Section \ref{sec:NonnegativityOfK2}, we obtained such an estimate for $\psi_{\geq 1}$. Next we derive a similar  estimate for the zeroth spherical harmonic $\psi_{0}$.  We first prove the following lemma

\begin{lemma}
Fix $R>2M$. There exists a constant $C$ which depends  on $M$, $\Sigma_{0}$ and $R$ such that  for all spherically symmetric solutions $\psi$ of the wave equation 
\begin{equation*}
\int_{\left\{r=R\right\}\cap\mathcal{R}(0,\tau)}{(\partial_{t}\psi)^{2}+(\partial_{r^{*}}\psi)^{2}}\leq C_{R}\int_{\Sigma_{0}}{J^{T}_{\mu}[\psi]n^{\mu}_{\Sigma_{0}}}.
\end{equation*}
\label{lemmal0}
\end{lemma}
\begin{proof}
Consider the region $\mathcal{F}(0,\tau)=\left\{\mathcal{R}(0,\tau)\cap\left\{M\leq r\leq R\right\}\right\}.$
By applying the vector field $X=\partial_{r^{*}}$ as a multiplier in the region $\mathcal{F}(0,\tau)$ we obtain
\begin{equation*}
\begin{split}
\int_{\mathcal{F}\cap\mathcal{H}^{+}}{J_{\mu}^{X}[\psi]n^{\mu}_{\mathcal{H}^{+}}}+\int_{\mathcal{F}}{K^{X}[\psi]}+\int_{\Sigma_{\tau}\cap\mathcal{F}}{J_{\mu}^{X}[\psi]n^{\mu}_{\Sigma_{\tau}}}+\int_{\left\{r=R\right\}\cap\mathcal{R}(0,\tau)}{J_{\mu}^{X}[\psi]n^{\mu}_{\mathcal{F}}}=\int_{\Sigma_{0}\cap\mathcal{F}}{J_{\mu}^{X}[\psi]n^{\mu}_{\Sigma_{0}}},
\end{split}
\end{equation*}
where $n^{\mu}_{\mathcal{F}}$ denotes the unit normal vector to $\left\{r=R\right\}$ pointing in the interior of $\mathcal{F}$. 
 \begin{figure}[H]
	\centering
		\includegraphics[scale=0.10]{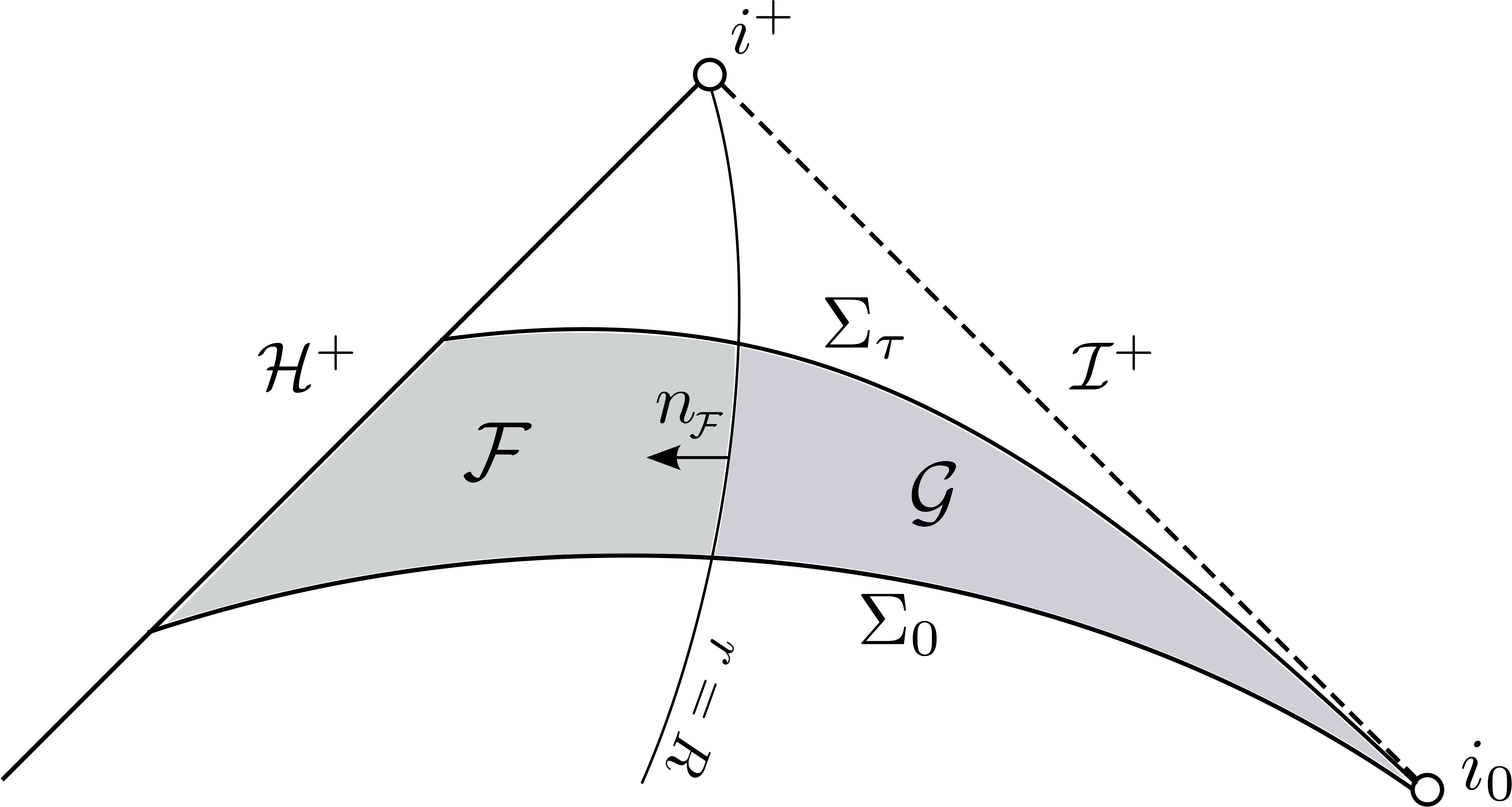}
	\label{fig:ernt5}
\end{figure}
In view of the spatial compactness of the region $\mathcal{F}$, the corresponding spacetime integral can be estimated using Propositions \ref{1l0} and \ref{boun1}. The boundary integrals over $\Sigma_{0}$ and $\Sigma_{\tau}$ can be estimated using Proposition \ref{boun1}. Moreover, for spherically symmetric waves $\psi$ we have $J_{\mu}^{X}[\psi]n^{\mu}_{\mathcal{F}}=-\frac{1}{2\sqrt{D}}\left((\partial_{t}\psi)^{2}+(\partial_{r^{*}}\psi)^{2}\right),$
which completes the proof.
\end{proof}
Consider now the region
$\mathcal{G}=\mathcal{R}\cap\left\{R\leq r\right\}.$
\begin{proposition}
Fix $R>2M$. There exists a constant $C$ which depends  on $M$, $\Sigma_{0}$ and $R$ such that  for all spherically symmetric solutions $\psi$ of the wave equation  
\begin{equation*}
\int_{\mathcal{G}}{\frac{1}{r^{4}}\psi^{2}} +\int_{\left\{r=R\right\}}{\psi^{2}}\leq C_{R}\int_{\Sigma_{0}}{J^{T}_{\mu}[\psi]n^{\mu}_{\Sigma_{0}}}.
\end{equation*}
\label{fl0}
\end{proposition}
\begin{proof}
By applying Stokes' theorem for the current $J_{\mu}^{X,1}[\psi]$ in the region $\mathcal{G}(0,\tau)$, where again $X=\partial_{r^{*}}$ (and, therefore, $f=1$), we obtain
\begin{equation*}
\int_{\mathcal{G}(0,\tau)}{K^{X,1}[\psi]}+\int_{\left\{r=R\right\}\cap\mathcal{G}(0,\tau)}{J_{\mu}^{X,1}[\psi]n^{\mu}_{\mathcal{G}}}\leq C\int_{\Sigma_{0}}{J^{T}_{\mu}[\psi]n^{\mu}_{\Sigma_{0}}}
\end{equation*}
where we have used Proposition \ref{boun2} to estimate the boundary integral over $\Sigma_{\tau}\cap\mathcal{G}$. Note that again $n^{\mu}_{\mathcal{G}}$ denotes the unit normal vector to $\left\{r=R\right\}$ pointing in the interior of $\mathcal{G}$. For $f=1$ we have
$K^{X,1}[\psi]=I\psi^{2},$
where $I>0$ and $I\sim \frac{1}{r^{4}}$ for $r\geq R>2M$. If  $G$ is the function defined in Section \ref{sec:TheCurrentJXGMuAndEstimatesForKXG} then $G=\frac{D}{2r}$ and therefore for sufficiently large $R$ we have $\partial_{r^{*}}G<0$. Then,
\begin{equation*}
J_{\mu}^{X,1}[\psi]n^{\mu}_{\mathcal{G}}=J_{\mu}^{X}n^{\mu}_{\mathcal{G}}+2G\psi(\nabla_{\mu}\psi)n^{\mu}_{\mathcal{G}}-(\nabla_{\mu}G)\psi^{2}n^{\mu}_{\mathcal{G}}.
\end{equation*}
Since $n_{\mathcal{G}}=\frac{1}{\sqrt{D}}\partial_{r^{*}}$, by applying Cauchy-Schwarz for the second term on the right hand side (and since the third term is positive) we obtain
\begin{equation*}
J_{\mu}^{X,1}[\psi]n^{\mu}_{\mathcal{F}}\sim\  \psi^{2}-\frac{1}{\epsilon}((\partial_{t}\psi)^{2}+(\partial_{r^{*}}\psi)^{2}),
\end{equation*}
for a sufficiently small $\epsilon$. Lemma \ref{lemmal0} completes the proof.
\end{proof}
It remains to obtain a (weighted) $L^{2}$ estimate for $\psi$ in the region $\mathcal{F}$.
\begin{proposition}Fix sufficiently large $R>2M$.  Then, there exists a constant $C$ which depends  on $M$, $\Sigma_{0}$ and $R$ such that  for all spherically symmetric solutions $\psi$ of the wave equation  
\begin{equation*}
\int_{\mathcal{F}}{D\psi^{2}}\leq C_{R}\int_{\Sigma_{0}}{J^{T}_{\mu}[\psi]n^{\mu}_{\Sigma_{0}}}.
\end{equation*}
\label{el0}
\end{proposition}
\begin{proof}
We apply Stokes' theorem for the current
\begin{equation*}
J_{\mu}^{H}[\psi]=(\nabla_{\mu}H)\psi^{2}-2H\psi\nabla_{\mu}\psi,
\end{equation*}
where $H=(r-M)^{2}$. All the boundary integrals can be estimated using Propositions \ref{boun2} and \ref{fl0}. Note also that 
\begin{equation*}
\nabla^{\mu}J_{\mu}^{H}=(\Box_{g}H)\psi^{2}-\frac{2H}{D}((\partial_{r^{*}}\psi)^{2}-(\partial_{t}\psi)^{2}).
\end{equation*}
Since $\Box_{g}H\sim D$ for $M\leq r \leq R$ and $\frac{H}{D}$ is bounded, the result follows in view of the spatial compactness of $\mathcal{F}$ and Proposition \ref{1l0}.
\end{proof}
We finally have the following zeroth order Morawetz  estimate:
\begin{theorem}
There exists a constant $C$ that depends  on $M$ and $\Sigma_{0}$ such that for all solutions $\psi$ of the wave equation we have
\begin{equation}
\int_{\mathcal{R}}{\frac{D}{r^{4}}\psi^{2}}\leq C\int_{\Sigma_{0}}{J^{T}_{\mu}[\psi]n^{\mu}_{\Sigma_{0}}}.
\label{moraw}
\end{equation}
\label{morawetz}
\end{theorem}
\begin{proof}
Write 
$\psi=\psi_{0}+\psi_{\geq 1}$
and use Propositions \ref{xa2prop}, \ref{el0} and \ref{fl0}.
\end{proof}
Clearly, the lemma used for Proposition \ref{fl0} holds stricly for the case $l=0$. In general, one could have argued by averaging  the $X$ estimate \eqref{degX} in $R$, which would imply that there exists a value $R_{0}$ of $r$ such that all the derivatives are controlled on the hypersurface of  constant radius $R_{0}$. This makes our argument work for all $\psi$ without recourse to the spherical decomposition.

\begin{remark}
The key property used in Proposition \ref{fl0} is that the d' Alembertian of $\frac{1}{r}$ is always negative, something  not true for larger powers of $\frac{1}{r}$. Note that this unstable behaviour of $\frac{1}{r}$ is expected since it is the static solution of the wave equation in Minkowski.
\end{remark}

\subsection{Discussion}
\label{sec:Discussion}

The current of Section  \ref{sec:TheCurrentJX2MuAndEstimatesForKX2} was first introduced in  \cite{dr3} and subsequently in \cite{dr5},  where a non-degenerate $X$ estimate is established for  Schwarzschild. In \cite{dr3}, the authors show that for each fixed spherical number $l$ there is a corresponding ``effective photon sphere'' centred at $r=-\gamma_{l}$. It is also shown that $\lim_{l\rightarrow+\infty}{-\gamma_{l}}\rightarrow 3M$. In our case, our computation \eqref{icomp} shows that all ``effective photon spheres'' coincide with the photon sphere, i.e. for all $l$ we have $-\gamma_{l}=2M$. Note also that in \cite{dr5}, one needs to commute with the generators of the Lie algebra so(3). Moreover, the boundary terms  could not be controlled by the flux of $T$  but one needed a  small portion of the redshift estimate. In our case the above geometric symmetry of the trapping and the asymptotic behavior of $r^{*}$ allowed us to completely decouple the dispersion from the redshift effect. For a nice exposition of previous work on $X$ estimates see \cite{md}.

\section{The Vector Field $N$}
\label{sec:TheVectorFieldN}

It is clear that in order to obtain an estimate for the non-degenerate energy of a local observer we need to use timelike multipliers at the horizon. Then uniform boundedness of  energy would  follow provided we can control the  spacetime terms which arise. For a suitable class of non-degenerate black hole spacetimes, not only have the bulk terms the right sign close to $\mathcal{H}^{+}$ but they in fact control the non-degenerate energy. Indeed, in  \cite{md} the following  is proved 
\begin{proposition}
Let $\mathcal{H}^{+}$ be a Killing horizon  with positive surface gravity and let $V$ be the Killing vector field  tangent to $\mathcal{H}^{+}$. Then there exist constants $b,B>0$ and a timelike vector field $N$ which is $\phi^{V}_{\tau}$-invariant (i.e. $\mathcal{L}_{V}N=0$)  such that for all functions $\psi$ we have
\begin{equation}
bJ^{N}_{\mu}[\psi]n^{\mu}_{\Sigma_{\tau}}\leq K^{N}[\psi]\leq BJ^{N}_{\mu}[\psi]n^{\mu}_{\Sigma_{\tau}} 
\label{knm}
\end{equation}
on $\mathcal{H}^{+}$.
\label{redshiftmd}
\end{proposition}
The construction of the above vector field does not require the global existence of a causal Killing vector field and  the positivity of the surface gravity suffices. Under suitable circumstances, one can prove that the $N$ flux is uniformly bounded  without  understanding the structure of the trapping (i.e.~no $X$ or Morawetz estimate is required). 

However, in our case, in view of the lack of redshift along $\mathcal{H}^{+}$ the situation is completely different. Indeed, we will see  that in  extreme Reissner-Nordstr\"{o}m   a vector field satisfying the properties of Proposition~\ref{redshiftmd} does not exist. Not only  will we show that there is no $\phi_{\tau}^{T}$-invariant vector field $N$ satisfying \eqref{knm} on $\mathcal{H}^{+}$ but we will in fact prove that there is no $\phi_{\tau}^{T}$-invariant timelike vector field $N$ such that 
\begin{equation*}
K^{N}[\psi]\geq 0
\end{equation*}
on  $\mathcal{H}^{+}$.  Our resolution to this problem uses an appropriate modification of $J_{\mu}^{N}$ and Hardy inequalities and thus is still robust. 

\subsection{The Effect of Vanishing Redshift on Linear Waves }
\label{sec:TheSpacetimeTermKN}
Let $N=N^{v}\left(r\right)\partial_{v}+N^{r}\left(r\right)\partial_{r}$ be a  future directed timelike $\phi_{\tau}^{T}$-invariant vector field. Then
\begin{equation*}
\begin{split}
K^{N}[\psi]=&F_{vv}\left(\partial_{v}\psi\right)^{2}+F_{rr}\left(\partial_{r}\psi\right)^{2}+F_{\scriptsize\nabb}\left|\nabb\psi\right|^{2}+F_{vr}\left(\partial_{v}\psi\right)\left(\partial_{r}\psi\right),
\end{split}
\end{equation*}
where the coef{}ficients are given by
\begin{equation}
\begin{split}
F_{vv}=\left(\partial_{r}N^{v}\right),  F_{rr}=D\left[\frac{\left(\partial_{r}N^{r}\right)}{2}-\frac{N^{r}}{r}\right]-\frac{N^{r}D'}{2}, F_{\scriptsize\nabb}=-\frac{1}{2}\left(\partial_{r}N^{r}\right), F_{vr}=D\left(\partial_{r}N^{v}\right)-\frac{2N^{r}}{r},
\label{list}
\end{split}
\end{equation}
where $D'=\frac{dD}{dr}$. Note that since $g\left(N,N\right)=-D\left(N^{v}\right)^{2}+2N^{v}N^{r},\  g\left(N,T\right)=-DN^{v}+N^{r},$
and so $N^{r}\left(r=M\right)$ can not be zero (otherwise the vector field $N$ would not be timelike). Therefore, looking back at the list \eqref{list} we see that the coefficient of $\left(\partial_{r}\psi\right)^{2}$ vanishes on the horizon $\mathcal{H}^{+}$ whereas the coefficient of $\partial_{v}\psi\partial_{r}\psi$ is equal to $-\frac{2N^{r}\left(M\right)}{M}$ which is not zero. Therefore, $K^{N}[\psi]$ is linear with respect to $\partial_{r}\psi$ on the horizon $\mathcal{H}^{+}$ and thus it necessarily fails to be non-negative definite. This linearity is a characteristic feature of the geometry of the event horizon of  extreme Reissner-Nordstr\"{o}m and degenerate black hole spacetimes more generally. This proves that  a vector field satisfying the properties  of Proposition \ref{redshiftmd} does not exist.

\subsection{A Locally Non-Negative Spacetime Current}
\label{sec:TheCurrentJMuNDeltaFrac12AndEstimatesForItsDivergence}

In view of the above discussion, we need to modify the bulk term by introducing new terms that will counteract the presence of $\partial_{v}\psi\partial_{r}\psi$.  We define 
\begin{equation}
J_{\mu}\equiv J_{\mu}^{N,h}=J_{\mu}^{N}+h\left(r\right)\psi\nabla _{\mu}\psi
\label{Jmodified}
\end{equation}
where $h$ is a function on $\mathcal{M}$. Then we have
\begin{equation*}
\begin{split}
K\equiv K^{N,h}=\nabla^{\mu}J_{\mu}=K^{N}+\left(\nabla^{\mu}h\right)\psi\nabla _{\mu}\psi+h\left(\nabla^{a}\psi\nabla_{a}\psi\right),
\end{split}
\end{equation*}
provided $\psi$ is a solution of the wave equation. Let us suppose that $h\left(r\right)=\frac{N^{r}\left(M\right)}{M}$. Then
\begin{equation}
\begin{split}
K^{N,h}[\psi]&=K^{N}[\psi]+h\left(\nabla^{a}\psi\nabla_{a}\psi\right)=K^{N}[\psi]+h\left(2\partial _{v}\psi\partial_{r}\psi+D\left(\partial _{r}\psi\right)^{2}+\left|\nabb\psi\right|^{2}\right)\\
&=F_{vv}\left(\partial_{v}\psi\right)^{2}+\left[F_{rr}+hD\right]\left(\partial_{r}\psi\right)^{2}+\left[-\frac{\partial_{r}N^{r}}{2}+h\right]\left|\nabb\psi\right|^{2}+\left[F_{vr}+2h\right]\left(\partial_{v}\psi\partial_{r}\psi\right).\\
\end{split}
\label{knn}
\end{equation}
Note that by taking $h$ to be constant  we managed to have no zeroth order terms in the  current $K$. Let us denote the above coef{}ficients of $\left(\partial_{a}\psi\partial_{b}\psi\right)$ by $G_{ab}$ where $a,b\in\left\{v,r,\nabb\right\}$ and define the vector field $N$ in the region $M\leq r\leq \frac{9M}{8}$ to be such that 
\begin{equation}
\begin{split}
N^{v}(r)=16r,\ N^{r}(r)=-\frac{3}{2}r+M
\end{split}
\label{n}
\end{equation}
and, therefore,
$h=-\frac{1}{2}.$
Clearly, $N$ is timelike future directed vector field. We have the following
\begin{proposition}
For all functions $\psi$,  the current $K^{N, -\frac{1}{2}}[\psi]$ defined by \eqref{knn} is non-negative definite in the region $\mathcal{A}_{N}=\left\{M\leq r\leq \frac{9M}{8}\right\}$ and, in particular, there is a positive constant $C$ that depends only on $M$ such that
\begin{equation*}
K^{N,-\frac{1}{2}}[\psi]\geq C\left(\left(\partial_{v}\psi\right)^{2}+\sqrt{D}\left(\partial_{r}\psi\right)^{2}+\left|\nabb\psi\right|^{2}\right).
\end{equation*}
\label{knprop}
\end{proposition}
\begin{proof}
We first observe that the coef{}ficient $G_{\scriptsize\nabb}$ of $\partial_{v}\psi\partial_{r}\psi$  is equal to
$G_{\scriptsize\nabb}=\frac{1}{4}$.
Clearly, $G_{vv}=16$ and $G_{rr}$ is non-negative since the factor of the dominant term $D'$ (which vanishes to first order on $\mathcal{H}^{+}$) is positive. As regards the coef{}ficient of the mixed term we have
$G_{vr}=16D+2\sqrt{D}=\epsilon_{1}\cdot\epsilon_{2},$
where
$\epsilon_{1}=\sqrt{D}, \epsilon_{2}=16\sqrt{D}+2.$ We will show that in region $\mathcal{A}_{N}$ we have
$\epsilon_{1}^{2}\leq G_{rr},\ \epsilon_{2}^{2}<G_{vv}.$ Indeed, if  we set $\lambda =\frac{M}{r}$ we have
\begin{equation*}
\begin{split}
\epsilon_{1}^{2}\leq G_{rr}\Leftrightarrow\left(1-\lambda\right)\leq\left(1-\lambda\right)\left(\frac{1}{4}-\lambda\right)-\lambda\left(-\frac{3}{2}+\lambda\right)\Leftrightarrow\lambda \leq \frac{3}{5},
\end{split}
\end{equation*}
which holds. Note also that $\epsilon_{1}^{2}$ vanishes to second order on $\mathcal{H}^{+}$ whereas $G_{rr}$ vanishes to first order. Therefore, in region $\mathcal{A}_{N}$ we have $G_{rr}-\epsilon_{1}^{2}\sim \sqrt{D}$. Similarly, 
\begin{equation*}
\begin{split}
&\epsilon_{2}^{2}< G_{vv}\Leftrightarrow 16\left[\left(1-\frac{M}{r}\right)+2\right]^{2}< 16\Leftrightarrow \lambda >\frac{7}{8},
\end{split}
\end{equation*}
which again holds. The proposition follows from the spatial compactness of $\mathcal{A}_{N}$ and  that $a^{2}+ab+b^{2}\geq 0$ for all $a,b\in\mathbb{R}$. 
\end{proof}

\subsection{The Cut-off $\delta$ and the Current $J_{\mu}^{N,\delta,-\frac{1}{2}}$}
\label{sec:TheCutOffDeltaAndTheCurrentJMuNDeltaFrac12}

Clearly, the current $K^{N,-\frac{1}{2}}$ will not be non-negative far away from $\mathcal{H}^{+}$ and thus is not useful there. For this reason we extend $N^{v},N^{r}$ such that 
\begin{equation*}
\begin{split}
&N^{v}\left(r\right)>0 \ \text{for all}\  r\geq M \ \text{and}\  N^{v}\left(r\right)=1\  \text{for all}\  r\geq \frac{8M}{7}, \\
&N^{r}\left(r\right)\leq 0\  \text{for all}\  r\geq M\  \text{and}\  N^{r}\left(r\right)=0\  \text{for all}\  r\geq \frac{8M}{7}.
\end{split}
\end{equation*}
$N$ as defined is a future directed timelike $\phi_{\tau}^{T}$-invariant vector field. Similarly, the modification term in the current $J^{N,-\frac{1}{2}}_{\mu}$ is not useful for consideration far away from  $\mathcal{H}^{+}$. We thus introduce a smooth cut-off function $\delta:[\left.M,+\infty\right.)\rightarrow\mathbb{R}$ such that $\delta\left(r\right)=1,r\in\left[M,\frac{9M}{8}\right]$ and $\delta\left(r\right)=0,r\in\left[\left.\frac{8M}{7},
+\infty\right.\right)$ and  consider the currents
\begin{equation}
\begin{split}
&J_{\mu}^{N,\delta,-\frac{1}{2}}\overset{.}{=}J_{\mu}^{N}-\frac{1}{2}\delta\psi\nabla_{\mu}\psi,\ \ K^{N,\delta,-\frac{1}{2}}\overset{.}{=}\nabla^{\mu}J_{\mu}^{N,\delta,-\frac{1}{2}}.\\
\end{split}
\end{equation}
We now consider the three regions $\mathcal{A}_{N},\mathcal{B}_{N},\mathcal{C}_{N}$
 \begin{figure}[H]
	\centering
		\includegraphics[scale=0.10]{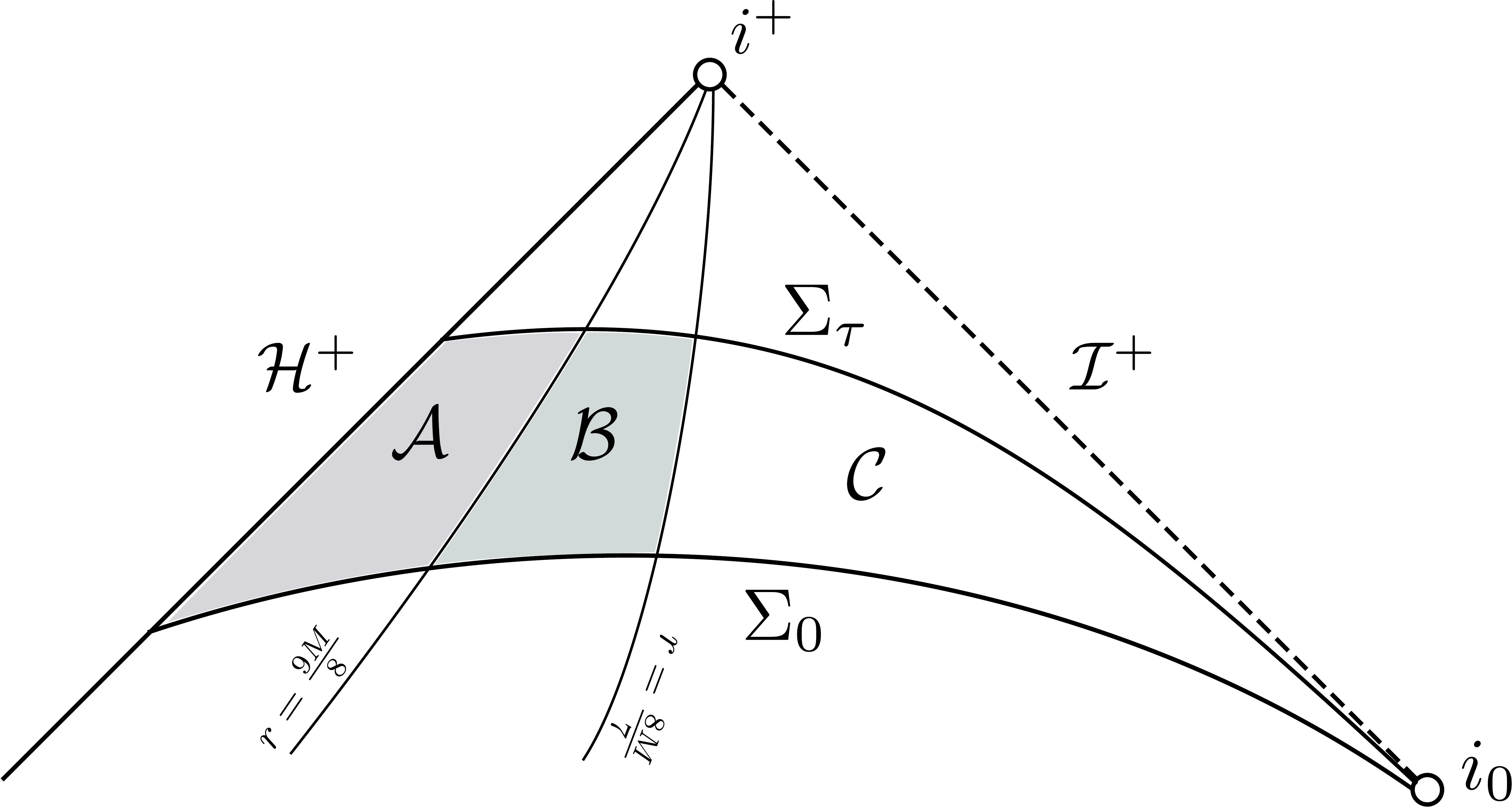}
	\label{fig:ernt6}
\end{figure}
In  region $\mathcal{C}_{N}=\left\{r\geq\frac{8M}{7}\right\}$ where $\delta=0$ and $ N=T$, we have
$K^{N,\delta,-\frac{1}{2}}= 0.$
However, this spacetime current, which depends on the 1-jet of $\psi$,  will generally be  negative in  region $\mathcal{B}_{N}$ and thus   will be controlled by the $X$ and Morawetz estimates (which, of course, are non-degenerate in $\mathcal{B}_{N}$).

We next control the Sobolev norm $\left\|\psi\right\|^{2}_{\dot{H}^{1}\left(\Sigma_{\tau}\right)}\lesssim\, \int_{\Sigma_{\tau}}{\left(\partial_{v}\psi\right)^{2}+\left(\partial_{r}\psi\right)^{2}+\left|\nabb\psi\right|^{2}}.$
Since $N$ and $n^{\mu}_{\Sigma_{\tau}}$ are timelike everywhere in $\mathcal{R}$ and since $\xi_{N},\xi_{n}$ are positive and uniformly bounded, \eqref{GENERALT} of Appendix \ref{sec:TheHyperbolicityOfTheWaveEquation1} implies $J_{\mu}^{N}[\psi]n^{\mu}_{\Sigma_{\tau}}\sim \, \left(\partial_{v}\psi\right)^{2}+\left(\partial_{r}\psi\right)^{2}+\left|\nabb\psi\right|^{2},$
where, in view of the $\phi^{T}_{\tau}$-invariance of $\Sigma_{\tau}$ and $N$, the constants in  $\sim$ depend only on $M$ and $\Sigma_{0}$. Therefore, it suffices to estimate the flux of $N$ through $\Sigma_{\tau}$. We have:
\begin{proposition}
There exists a  constant $C>0$ which depends  on $M$ and $\Sigma_{0}$ such that for all functions $\psi$ 
\begin{equation}
\int_{\Sigma_{\tau}}{J_{\mu}^{N}[\psi]n^{\mu}}\leq 2\int_{\Sigma_{\tau}}{J_{\mu}^{N,\delta,-\frac{1}{2}}[\psi]n^{\mu}}+C\int_{\Sigma_{\tau}}{J_{\mu}^{T}[\psi]n^{\mu}}.
\label{nboundary}
\end{equation}
\label{nb}
\end{proposition}
\begin{proof}
We have
\begin{equation*}
\begin{split}
J_{\mu}^{N,\delta,-\frac{1}{2}}n^{\mu}&=J_{\mu}^{N}n^{\mu}-\frac{1}{2}\delta\psi\partial_{\mu}\psi n^{\mu}=J_{\mu}^{N}n^{\mu}-\frac{1}{2}\delta\psi\partial_{v}\psi n^{v}-\frac{1}{2}\delta\psi\partial_{r}\psi n^{r}\\
&\geq J_{\mu}^{N}n^{\mu}-\epsilon\delta(\partial_{v}\psi)^{2}-\epsilon\delta(\partial_{r}\psi)^{2}-\frac{\delta}{\epsilon}\psi^{2}\geq  \frac{1}{2}J_{\mu}^{N}n^{\mu}-\frac{\delta}{\epsilon}\psi^{2}
\end{split}
\end{equation*}
for a sufficiently small $\epsilon$. The result follows from the first Hardy inequality.  
\end{proof}
\begin{corollary}
There exists a  constant $C>0$ which depends  on $M$ and $\Sigma_{0}$ such that for all functions $\psi$ 
\begin{equation}
\int_{\Sigma_{\tau}}{J_{\mu}^{N,\delta ,-\frac{1}{2}}[\psi]n^{\mu}}\leq C\int_{\Sigma_{\tau}}{J_{\mu}^{N}[\psi]n^{\mu}}.
\label{nboundary1}
\end{equation}
\label{nb1}
\end{corollary}
\begin{proof}
Note that
\begin{equation*}
\begin{split}
J_{\mu}^{N,\delta,-\frac{1}{2}}n^{\mu}&=J_{\mu}^{N}n^{\mu}-\frac{1}{2}\delta\psi\partial_{\mu}\psi n^{\mu}\leq J_{\mu}^{N}n^{\mu}+\delta(\partial_{v}\psi)^{2}+\delta(\partial_{r}\psi)^{2}+C\delta\psi^{2}\\
\end{split}
\end{equation*}
and use the first Hardy inequality.
\end{proof}

\subsection{Lower Estimate for an Integral over $\mathcal{H}^{+}$}
\label{sec:EstimatesForAnIntegralOverMathcalH}
As regards the integral over $\mathcal{H}^{+}$, we have
\begin{proposition}
For all functions $\psi$ and any $\epsilon >0$ we have 
\begin{equation*}
\int_{\mathcal{H}^{+}}{J_{\mu}^{N,\delta,-\frac{1}{2}}[\psi]n^{\mu}_{\mathcal{H}^{+}}}\geq \int_{\mathcal{H}^{+}}{J_{\mu}^{N}[\psi]n^{\mu}_{\mathcal{H}^{+}}}\ -C_{\epsilon}\int_{\Sigma_{\tau}}{J_{\mu}^{T}[\psi]n^{\mu}_{\Sigma_{\tau}}}\ -\epsilon\int_{\Sigma_{\tau}}{J_{\mu}^{N}[\psi]n^{\mu}_{\Sigma_{\tau}}},
\end{equation*}
\label{nhb}
where $C_{\epsilon}$ depends on $M$, $\Sigma_{0}$ and $\epsilon$.
\end{proposition}
\begin{proof}
Recall our convention  $n_{\scriptstyle\mathcal{H}^{+}}=T$. On $\mathcal{H}^{+}$ we have $\delta =1$. Therefore, $J_{\mu}^{N,\delta,-\frac{1}{2}}n^{\mu}_{\mathcal{H}^{+}}=J_{\mu}^{N}n^{\mu}_{\mathcal{H}^{+}}-\frac{1}{2}\psi\partial_{v}\psi$.
However,
\begin{equation*}
\begin{split}
\int_{\mathcal{H}^{+}}{-2\psi\partial_{v}\psi}&=\int_{\mathcal{H}^{+}}{-\partial_{v}\psi^{2}}=\int_{\mathcal{H}^{+}\cap\Sigma_{0}}{\psi^{2}}-\int_{\mathcal{H}^{+}\cap\Sigma_{\tau}}{\psi^{2}}.
\end{split}
\end{equation*}
From the first and second Hardy inequality  we have
\begin{equation*}
\int_{\mathcal{H}^{+}\cap\Sigma}{\psi^{2}}\leq C_{\epsilon}\int_{\Sigma}{J_{\mu}^{T}n^{\mu}_{\Sigma}}+\epsilon\int_{\Sigma}{(\partial_{v}\psi)^{2}+(\partial_{r}\psi)^{2}}\leq C_{\epsilon}\int_{\Sigma}{J_{\mu}^{T}n^{\mu}_{\Sigma}}+\epsilon\int_{\Sigma}{J_{\mu}^{N}n^{\mu}_{\Sigma}},
\end{equation*}
which completes the proof.
\end{proof}

\section{Uniform Boundedness of Local Observer's Energy}
\label{sec:UniformBoundednessOfLocalObserverSEnergy}
We have all tools in place in order to prove the following theorem
\begin{theorem}
There exists a  constant $C>0$ which depends  on $M$ and $\Sigma_{0}$ such that for all solutions $\psi$ of the wave equation 
\begin{equation}
\int_{\Sigma_{\tau}}{J_{\mu}^{N}[\psi]n^{\mu}_{\Sigma_{\tau}}}+\int_{\mathcal{H}^{+}}{J_{\mu}^{N}[\psi]n^{\mu}_{\mathcal{H}^{+}}}\leq C\int_{\Sigma_{0}}{J_{\mu}^{N}[\psi]n^{\mu}_{\Sigma_{0}}}.
\label{nener}
\end{equation}
\label{nenergy}
\end{theorem}
\begin{proof}
Stokes' theorem for the current $J_{\mu}^{N,\delta, -\frac{1}{2}}$ in region $\mathcal{R}(0,\tau)$ gives us
\begin{equation*}
\int_{\Sigma_{\tau}}{J_{\mu}^{N,\delta, -\frac{1}{2}}n^{\mu}_{\Sigma_{\tau}}}+\int_{\mathcal{R}}{K^{N,\delta, -\frac{1}{2}}}+\int_{\mathcal{H}^{+}}{J_{\mu}^{N,\delta, -\frac{1}{2}}n^{\mu}_{\mathcal{H}^{+}}}= \int_{\Sigma_{0}}{J_{\mu}^{N,\delta, -\frac{1}{2}}n^{\mu}_{\Sigma_{0}}}.
\end{equation*}
First observe that the right hand side is controlled by the right hand side of \eqref{nboundary1}. As regards the left hand side, the boundary integrals can be estimated using Propositions \ref{nb} and \ref{nhb}. The spacetime term is non-negative (and thus has the right sign) in region $\mathcal{A}_{N}$, vanishes in region $\mathcal{C}_{N}$ and can be estimated in the spatially compact region $\mathcal{B}_{N}$ (which does not contain the photon sphere) by the X estimate \eqref{degX} and Morawetz estimate \eqref{moraw}. The result follows from the boundedness of $T$-flux through $\Sigma_{\tau}$.
\end{proof}

\begin{corollary}
There exists a  constant $C>0$ which depends  on $M$ and $\Sigma_{0}$ such that for all solutions $\psi$ of the wave equation 
\begin{equation}
\int_{\mathcal{A}_{N}}{K^{N,-\frac{1}{2}}[\psi]}\leq C\int_{\Sigma_{0}}{J_{\mu}^{N}[\psi]n^{\mu}_{\Sigma_{0}}}.
\label{nk}\end{equation}
\label{nkcor}
\end{corollary}
Theorem \ref{t2} of Section \ref{sec:TheMainTheorems} is then implied by \eqref{nener}, Proposition \ref{knprop} and the above corollary. Note also that Corollary \ref{nkcor} and estimate \eqref{x} give us a spacetime integral where the only weight that locally degenerates (to first order) is that of the derivative tranversal to $\mathcal{H}^{+}$. Recall that in the  subextreme Reissner-Nordstr\"{o}m case there is no such degeneration. In the next section, this degeneracy is eliminated provided $\psi_{0}=0$. This condition is necessary as is shown in Section \ref{sec:TheSphericallySymmetricCase}.

One application of the above theorem is the following  Morawetz estimate which does not degenerate at $\mathcal{H}^{+}$.
\begin{proposition}
There exists a  constant $C>0$ which depends  on $M$ and $\Sigma_{0}$ such that for all solutions $\psi$ of the wave equation 
\begin{equation}
\int_{\mathcal{A}_{N}}{\psi^{2}}\leq C\int_{\Sigma_{0}}{J_{\mu}^{N}[\psi]n^{\mu}_{\Sigma_{0}}}.
\label{nodmoraw}
\end{equation}
\label{moranondeg}
\end{proposition}
\begin{proof}
The third Hardy inequality gives us
\begin{equation*}
\int_{\mathcal{A}_{N}}{\psi^{2}}\leq C\int_{\mathcal{B}_{N}}{\psi^{2}}+C\int_{\mathcal{A}_{N}\cup\mathcal{B}_{N}}{D\left((\partial_{v}\psi)^{2}+(\partial_{r}\psi)^{2},\right)}
\end{equation*}
where $C$ is a uniform positive constant that depends only on $M$. The integral over $\mathcal{B}_{N}$ of $\psi^{2}$ can be estimated using \eqref{moraw} and the last integral on the right hand side can be estimated using \eqref{degX} and Propositions  \ref{knprop} and Corollary \ref{nkcor}.
\end{proof}

The above proposition in conjunction with Remark \ref{2ndkindx} and Theorem \ref{morawetz} implies statement (1) of Theorem \ref{th1} of Section \ref{sec:TheMainTheorems}.

\begin{remark}
Note that all the above estimates hold if we replace the foliation $\Sigma_{\tau}$ with the foliation $\tilde{\Sigma}_{\tau}$ which terminates at $\mathcal{I}^{+}$ since the only difference is a boundary integral over $\mathcal{I}^{+}$ of the right sign that arises every time we apply Stokes' theorem. The remaining local estimates are exactly the same. 
\end{remark}

\section{Commuting with a Vector Field Transversal to $\mathcal{H}^{+}$}
\label{sec:CommutingWithAVectorFieldTransversalToMathcalH}

As we have noticed, the estimate \eqref{nk} degenerates with respect to the transversal derivative $\partial_{r}$ to $\mathcal{H}^{+}$. In order to remove this degeneracy we commute the wave equation $\Box_{g}\psi=0$ with  $\partial_{r}$ aiming at additionally controlling all the second derivatives of $\psi$ (on the spacelike hypersurfaces and the spacetime region up to and including the horizon $\mathcal{H}^{+}$). Such commutations first appeared in \cite{dr7}.

\subsection{The Spherically Symmetric Case}
\label{sec:TheSphericallySymmetricCase}
Let us first consider spherically symmetric waves. We have the following
\begin{proposition}
For all  spherically symmetric solutions $\psi$ to the wave equation the quantity 
\begin{equation}
\partial_{r}\psi +\frac{1}{M}\psi
\label{0l}
\end{equation}
is conserved along $\mathcal{H}^{+}$. Therefore, for generic initial data  this quantity does not decay.
\label{ndl0}
\end{proposition}
\begin{proof}
Since $\psi$ solves $\Box_{g}\psi =0$ and since $\lapp\psi =0$ we have
$\partial_{v}\partial_{r}\psi+\frac{1}{M}\partial_{v}\psi=0$
and, since $\partial_{v}$ is tangential to $\mathcal{H}^{+}$, this implies that $\partial_{r}\psi +\frac{1}{M}\psi$ remains constant along $\mathcal{H}^{+}$ and clearly for generic initial data it is not equal to zero.
\end{proof}
By projecting on the zeroth spherical harmonic we deduce that for a generic wave $\psi$ either $\psi^{2}$ or $(\partial_{r}\psi)^{2}$ does not decay along $\mathcal{H}^{+}$ completing the proof of Theorem \ref{nondecay} of Section \ref{sec:TheMainTheorems}. The nature of the above proposition (which turns out to be of fundamental importance for extreme black hole spacetimes) will be  investigated in great detail in the companion papers \cite{aretakis2,aretakis3}, where, in particular a similar non-decay result is shown to hold even for ``generic'' waves whose zeroth spherical harmonic vanishes.

Clearly, this proposition indicates that if one is to commute with $\partial_{r}$ then one must exclude the spherically symmetric waves and thus consider  only  angular frequencies $l\geq 1$. Indeed, we next obtain the sharpest possible result.

\subsection{Commutation with the Vector Field $\partial_{r}$}
\label{sec:CommutationWithTheVectorFieldPartialR}

We  compute the commutator $\left[\Box_{g},\partial_{r}\right]$. First note that 
$\left[\lapp,\partial_{r}\right]\psi=\frac{2}{r}\lapp\psi.$
Therefore, if 
$R=D'+\frac{D}{2r},  D'=\frac{d D}{dr}$
then
\begin{equation}
\begin{split}
\left[\Box_{g},\partial_{r}\right]\psi=-D'\partial_{r}\partial_{r}\psi+\frac{2}{r^{2}}\partial_{v}\psi-R'\partial_{r}\psi+\frac{2}{r}\lapp\psi,
\end{split}
\label{comm1}
\end{equation}
where $R'=\frac{d R}{dr}$.

\subsection{The Multiplier $L$ and the Energy Identity}
\label{sec:TheMultiplierLAndTheEnergyIdentity}

For any solution $\psi$ of the wave equation we have complete control of the second order derivatives of $\psi$ away from  $\mathcal{H}^{+}$ since
\begin{equation}
\begin{split}
\left\|\partial_{a}\partial_{b}\psi\right\|^{2}_{L^{2}\left(\Sigma_{\tau}\cap\left\{M<r_{0}\leq r\right\}\right)}\leq C\int_{\Sigma_{0}}{J^{T}_{\mu}[\psi]n^{\mu}_{\Sigma_{0}}}+C\int_{\Sigma_{0}}{J^{T}_{\mu}\left[T\psi\right]n^{\mu}_{\Sigma_{0}}},
\end{split}
\label{hypersurawayH}
\end{equation}
where $C$ depends  on $M$, $r_{0}$ and $\Sigma_{0}$ and $a,b\in\left\{v,r,\nabb \right\}$. Note that we have used elliptic estimates since $T$ is timelike away from $\mathcal{H}^{+}$ (the zeroth order terms can be estimated by the first Hardy inequality). Similarly, away from $\mathcal{H}^{+}$, we  have control of the bulk integrals of the second order derivatives since \eqref{moraw} and local elliptic estimates imply
\begin{equation}
\left\|\partial_{a}\partial_{b}\psi\right\|^{2}_{L^{2}\left(\mathcal{R}\left(0,\tau\right)\cap\left\{M<r_{0}\leq r\leq r_{1}<2M\right\}\right)}\leq C\int_{\Sigma_{0}}{J^{T}_{\mu}\left[\psi\right]n^{\mu}_{\Sigma_{0}}}+C\int_{\Sigma_{0}}{J^{T}_{\mu}\left[T\psi\right]n^{\mu}_{\Sigma_{0}}}.
\label{bulkawayH}
\end{equation}
The above estimates will be required for bounding several spacetime error terms away from $\mathcal{H}^{+}$. 

In order to estimate the second order derivatives of  $\psi$ in a neighbourhood of  $\mathcal{H}^{+}$ we will construct an appropriate future directed timelike $\phi_{\tau}^{T}$-invariant vector field 
\begin{equation*}
L=L^{v}\partial_{v}+L^{r}\partial_{r},
\end{equation*} 
which will be used as our multiplier. Since we are interested in the region $M\leq r\leq r_{0}<r_{1}$,   $L$ will be such that $L=\textbf{0}$ in $r\geq r_{1}$ and $L$ timelike in the region $M\leq r < r_{1}$. Note that $r_{0},r_{1}$ are constants to be determined later on. The regions $\mathcal{A},\mathcal{B}$ are depicted below
 \begin{figure}[H]
	\centering
		\includegraphics[scale=0.10]{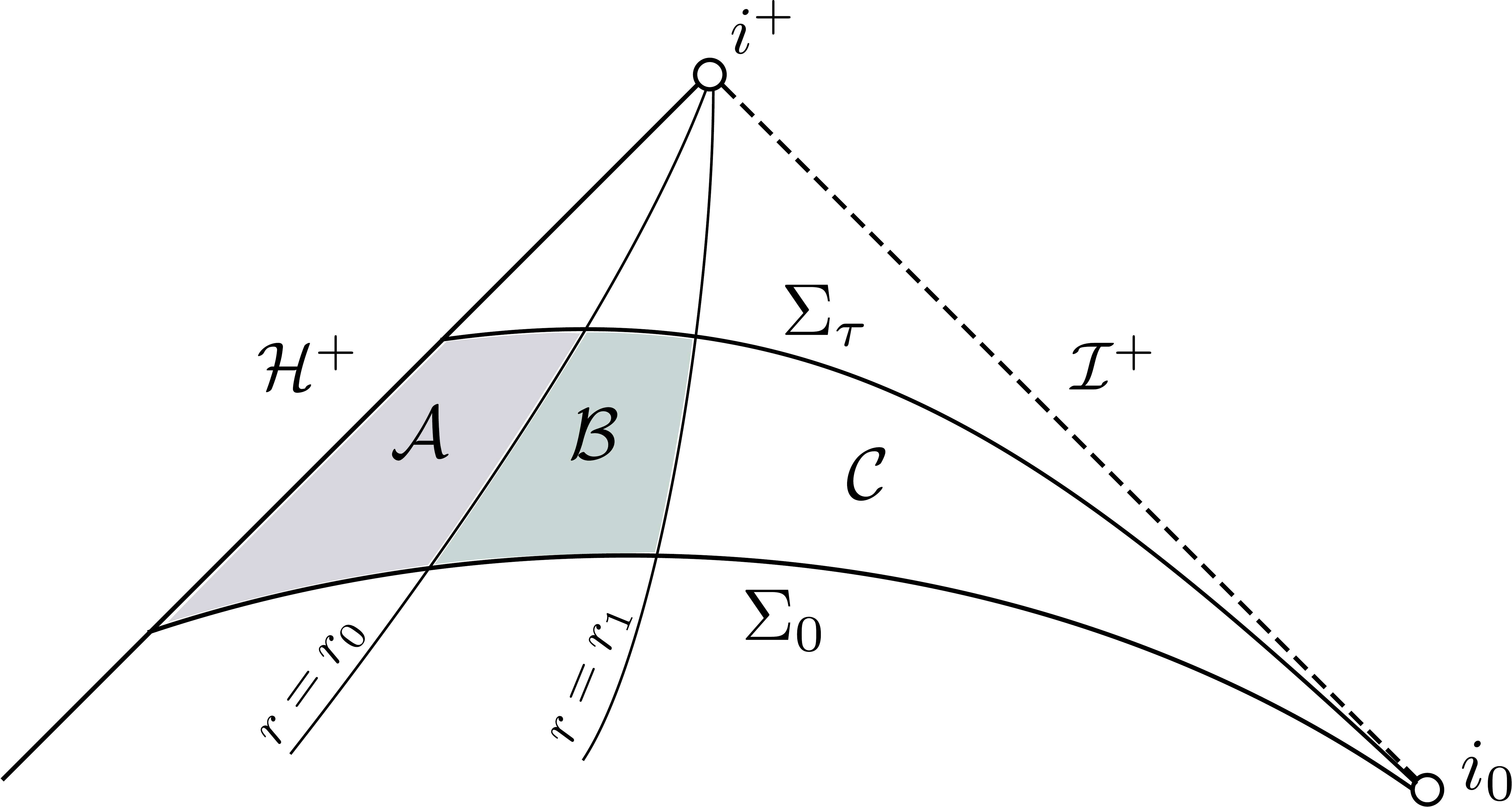}
	\label{fig:ernt7}
\end{figure}
For simplicity we will write $\mathcal{R}$ instead of $\mathcal{R}\left(0,\tau\right)$, $\mathcal{A}$ instead of $\mathcal{A}\left(0,\tau\right)$, etc. 
The ``energy'' identity for the current $J_{\mu}^{L}\left[\partial_{r}\psi\right]$ is 
\begin{equation}
\int_{\Sigma_{\tau}}{J_{\mu}^{L}\left[\partial_{r}\psi\right]n^{\mu}_{\Sigma_{\tau}}}+\int_{\mathcal{R}}{\nabla^{\mu}J_{\mu}^{L}\left[\partial_{r}\psi\right]}+\int_{\mathcal{H}^{+}}{J_{\mu}^{L}\left[\partial_{r}\psi\right]n^{\mu}_{\mathcal{H}^{+}}}=\int_{\Sigma_{0}}{J_{\mu}^{L}\left[\partial_{r}\psi\right]n^{\mu}_{\Sigma_{0}}}.
\label{eiL}
\end{equation}
The right hand side is controlled by the initial data and thus bounded. Also since $L$ is timelike in the compact region $\mathcal{A}$ we have from \eqref{GENERALT} of Appendix \ref{sec:TheHyperbolicityOfTheWaveEquation1}:
\begin{equation}
J_{\mu}^{L}\left[\partial_{r}\psi\right]n^{\mu}_{\Sigma_{\tau}}\sim \left(\partial_{v}\partial_{r}\psi\right)^{2}+\left(\partial_{r}\partial_{r}\psi\right)^{2}+\left|\nabb\partial_{r}\psi\right|^{2},
\label{Lhyper1}
\end{equation}
where the constants in $\sim$ depend on $M$, $\Sigma_{0}$ and $L$. Furthermore, on $\mathcal{H}^{+}$ we have 
\begin{equation}
J_{\mu}^{L}\left[\partial_{r}\psi\right]n^{\mu}_{\mathcal{H}^{+}}=L^{v}(M)(\partial_{v}\partial_{r}\psi)^{2}-\frac{L^{r}(M)}{2}\left|\nabb\partial_{r}\psi\right|^{2}. 
\label{horizonL1}
\end{equation}
Therefore, the term that remains to be understood is the bulk integral. Since $\partial_{r}\psi$ does not satisfy the wave equation, we have
\begin{equation*}
\begin{split}
\nabla^{\mu}J_{\mu}^{L}\left[\partial_{r}\psi\right]&=K^{L}\left[\partial_{r}\psi\right]+\mathcal{E}^{L}\left[\partial_{r}\psi\right]=K^{L}\left[\partial_{r}\psi\right]+\left(\Box_{g}\left(\partial_{r}\psi\right)\right)L\left(\partial_{r}\psi\right).
\end{split}
\end{equation*}
We know that
\begin{equation*}
\begin{split}
K^{L}\left[\partial_{r}\psi\right]=F_{vv}\left(\partial_{v}\partial_{r}\psi\right)^{2}+F_{rr}\left(\partial_{r}\partial_{r}\psi\right)^{2}+F_{\scriptsize{\nabb}}\left|\nabb\partial_{r}\psi\right|^{2}+F_{vr}\left(\partial_{v}\partial_{r}\psi\right)\left(\partial_{r}\partial_{r}\psi\right)
\end{split}
\end{equation*}
where the coefficients are as in \eqref{list}. In view of equation \eqref{comm1} we have
\begin{equation*}
\begin{split}
\mathcal{E}^{L}\left[\partial_{r}\psi\right]=&-D'
L^{r}\left(\partial_{r}\partial_{r}\psi\right)^{2}-D'L^{v}\left(\partial_{v}\partial_{r}\psi\right)\left(\partial_{r}\partial_{r}\psi\right)-R'L^{v}\left(\partial_{v}\partial_{r}\psi\right)\left(\partial_{r}\psi\right)\\
&+2\frac{L^{v}}{r^{2}}\left(\partial_{v}\partial_{r}\psi\right)\left(\partial_{v}\psi\right)+2\frac{L^{r}}{r^{2}}\left(\partial_{r}\partial_{r}\psi\right)\left(\partial_{v}\psi\right)-R'L^{r}\left(\partial_{r}\partial_{r}\psi\right)\left(\partial_{r}\psi\right)\\
&+2\frac{L^{v}}{r}(\partial_{v}\partial_{r}\psi)\lapp\psi+2\frac{L^{r}}{r}(\partial_{r}\partial_{r}\psi)\lapp\psi,
\end{split}
\end{equation*}
Therefore, we can write
\begin{equation*}
\begin{split}
\nabla^{\mu}J_{\mu}^{L}\left[\partial_{r}\psi\right]=&H_{1}\left(\partial_{v}\partial_{r}\psi\right)^{2}+H_{2}\left(\partial_{r}\partial_{r}\psi\right)^{2}+H_{3}\left|\nabb\partial_{r}\psi\right|^{2}+H_{4}\left(\partial_{v}\partial_{r}\psi\right)\left(\partial_{v}\psi\right)+H_{5}\left(\partial_{v}\partial_{r}\psi\right)\left(\partial_{r}\psi\right)\\&+H_{6}\left(\partial_{r}\partial_{r}\psi\right)\left(\partial_{v}\psi\right)+H_{7}(\partial_{v}\partial_{r}\psi)\lapp\psi+H_{8}(\partial_{r}\partial_{r}\psi)\lapp\psi+H_{9}\left(\partial_{v}\partial_{r}\psi\right)\left(\partial_{r}\partial_{r}\psi\right)\\&+H_{10}\left(\partial_{r}\partial_{r}\psi\right)\left(\partial_{r}\psi\right),
\end{split}
\end{equation*}
where the coefficients $H_{i},i=1,...,10$ are given by
\begin{equation}
\begin{split}
&H_{1}=\left(\partial_{r}L^{v}\right),\  H_{2}=D\left[\frac{\left(\partial_{r}L^{r}\right)}{2}-\frac{L^{r}}{r}\right]-\frac{3D'}{2}L^{r},\  H_{3}=-\frac{1}{2}\left(\partial_{r}L^{r}\right),\\
&H_{4}=+2\frac{L^{v}}{r^{2}}, \  H_{5}=-L^{v}R', \ H_{6}=+2\frac{L^{r}}{r^{2}},\  H_{7}=2\frac{L^{v}}{r},\  H_{8}=2\frac{L^{r}}{r},\\
&H_{9}=D\left(\partial_{r}L^{v}\right)-D'L^{v}-2\frac{L^{r}}{r} ,\  H_{10}=-L^{r}R'.
\end{split}
\label{listH}
\end{equation}
Since $L$ is a future directed timelike vector field we have $L^{v}\left(r\right)>0$ and $L^{r}\left(r\right)<0$ near $\mathcal{H}^{+}$. By taking $\partial_{r}L^{v}\left(M\right)$ sufficiently large we can make $H_{1}$ positive close to the horizon $\mathcal{H}^{+}$. Also since the term $D$ vanishes on the horizon to second order and the terms $R, D'$ to first order and since $L^{r}\left(M\right)<0$, the coefficient $H_{2}$ is positive close to $\mathcal{H}^{+}$ (and vanishes to first order on it). For the same reason we have $H_{9}D\leq\frac{H_{2}}{10}$  and $\left(H_{9}R\right)^{2}\leq\frac{H_{2}}{10}$ close to $\mathcal{H}^{+}$. Moreover, by taking  $-\partial_{r}L^{r}\left(M\right)$ sufficiently large we can also make the coefficient $H_{3}$ positive close to $\mathcal{H}^{+}$ such that $H_{9}<\frac{H_{3}}{10}$. Indeed, it suffices to consider $L^{r}$ such that $-\frac{L^{r}(M)}{M}<-\frac{\partial_{r}L^{r}(M)}{25}$ and then by continuity we have the previous inequality close to  $\mathcal{H}^{+}$. Therefore, we consider $M<r_{0}<2M$ such that in region $\mathcal{A}=\left\{M\leq r\leq r_{0}\right\}$ we have 
\begin{equation}
\begin{split}
& L^{v}> 1,\  \partial_{r}L^{v}> 1,\  -L^{r}>1,\  H_{1}>1,\  H_{2}\geq 0,\  H_{3}> 1,\\
& H_{8}<\frac{H_{3}}{10},\  H_{9}D\leq\frac{H_{2}}{10},\ \left(H_{9}R\right)^{2}\leq\frac{H_{2}}{10},\  H_{9}<\frac{H_{3}}{10}.
\end{split}
\label{listL}
\end{equation}
Clearly, $r_{0}$ depends only on $M$ and the precise choice for $L$ close to $\mathcal{H}^{+}$. In order to define $L$ globally, we just extend $L^{v}$ and $L^{r}$ such that 
\begin{equation*}
\begin{split}
&L^{v}>0 \text{ for all } r<r_{1}  \text{ and } L^{v}=0 \text{ for all } r\geq r_{1},\\
-&L^{r}>0 \text{ for all } r<r_{1} \text{ and } L^{r}=0 \text{ for all } r\geq r_{1},
\end{split}
\end{equation*}
for some $r_{1}$ such that $r_{0}<r_{1}<2M$. Again, $r_{1}$ depends only on $M$ (and the precise choice for $L$). Clearly, $L$ depends only on $M$ and thus all the functions that involve the components of $L$ depend only on $M$.

\subsection{Estimates of the Spacetime Integrals}
\label{sec:EstimatesOfTheBulkTerms}

It suffices to estimate the remaining 7 integrals with coef{}ficients $H_{i}$'s with $i=4,...,10$. Note that all these coef{}ficients do not vanish on the horizon. We will prove that each of these integrals can by estimated by the $N$ flux for $\psi$ and $T\psi$ and a small (epsilon) portion of the good terms in  $K^{L}[\partial_{r}\psi]$. First we prove the following propositions.

\begin{proposition}
For all solutions $\psi$ of the wave equation and  any positive number $\epsilon$ we have
\begin{equation*}
\begin{split}
\int_{\mathcal{A}}{(\partial_{r}\psi)^{2}}\leq \int_{\mathcal{A}}{\epsilon H_{1}(\partial_{v}\partial_{r}\psi)^{2}+\epsilon H_{2}(\partial_{r}\partial_{r}\psi)^{2}}+C_{\epsilon}\int_{\Sigma_{0}}{J_{\mu}^{N}[\psi]n^{\mu}_{\Sigma_{0}}}+C_{\epsilon}\int_{\Sigma_{0}}{J_{\mu}^{N}[T\psi]n^{\mu}_{\Sigma_{0}}},
\end{split}
\end{equation*}
where the constant $C_{\epsilon}$ depends  on $M$, $\Sigma_{0}$ and $\epsilon$.
\label{1comprop}
\end{proposition}
\begin{proof}
By applying the third Hardy inequality for the regions $\mathcal{A},\mathcal{B}$ we obtain
\begin{equation*}
\int_{\mathcal{A}}{(\partial_{r}\psi)^{2}}\leq C\int_{\mathcal{B}}{(\partial_{r}\psi)^{2}}+C\int_{\mathcal{A}\cup\mathcal{B}}{D\left[(\partial_{v}\partial_{r}\psi)^{2}+(\partial_{r}\partial_{r}\psi)^{2}\right]},
\end{equation*}
where the constant $C$ depends  on $M$ and $\Sigma_{0}$. Moreover, since $CD$ vanishes to second order at $\mathcal{H}^{+}$, there exists $r_{\epsilon}$ with  $M<r_{\epsilon}\leq r_{0}$ such that  in the region $\left\{M\leq r\leq r_{\epsilon}\right\}$ we have $CD<\epsilon H_{1}$ and $CD\leq \epsilon H_{2}$. Therefore,
\begin{equation*}
\begin{split}
\int_{\mathcal{A}}{\!(\partial_{r}\psi)^{2}}\leq &C\!\!\int_{\mathcal{B}}{\!(\partial_{r}\psi)^{2}}+\!\!\int_{\left\{M\leq r\leq r_{\epsilon}\right\}}{\!\!\!\!\!\epsilon H_{1}(\partial_{v}\partial_{r}\psi)^{2}+\epsilon H_{2}(\partial_{r}\partial_{r}\psi)^{2}}+\!\!\int_{\left\{r_{\epsilon}\leq r\leq r_{1}\right\}}{\!\!\!\!\!D\left[(\partial_{v}\partial_{r}\psi)^{2}+(\partial_{r}\partial_{r}\psi)^{2}\right]}.
\end{split}
\end{equation*}
The first integral on the right hand side is estimated using \eqref{degX} and the last integral using the local elliptic estimate \eqref{bulkawayH}  since  T is timelike in region $\left\{r_{\epsilon}\leq r \leq r_{1}\right\}$ since $M<r_{\epsilon}$. The result follows from the inclusion $\left\{M\leq r \leq r_{\epsilon}\right\}\subseteq\mathcal{A}$ and the non-negativity of $H_{i},i=1,2$ in $\mathcal{A}$.
\end{proof}
\begin{proposition}
For all solutions $\psi$ of the wave equation and any positive number $\epsilon$ we have 
\begin{equation*}
\left|\int_{\mathcal{H}^{+}}(\partial_{v}\psi)(\partial_{r}\psi)\right|\leq C_{\epsilon}\int_{\Sigma_{0}}{J^{N}_{\mu}[\psi]n^{\mu}_{\Sigma_{0}}}+\epsilon\int_{\Sigma_{0}\cup\Sigma_{\tau}}{J^{L}_{\mu}[\partial_{r}\psi]n^{\mu}_{\Sigma}}
\end{equation*}
where the positive constant $C_{\epsilon}$ depends  on $M$, $\Sigma_{0}$ and $\epsilon$.
\label{2comprop}
\end{proposition}
\begin{proof}
Integrating by parts gives us
\begin{equation*}
\begin{split}
\int_{\mathcal{H}^{+}}{(\partial_{v}\psi)(\partial_{r}\psi)}=-\int_{\mathcal{H}^{+}}{\psi(\partial_{v}\partial_{r}\psi)}+\int_{\mathcal{H}^{+}\cap\Sigma_{\tau}}{\psi(\partial_{r}\psi)}-\int_{\mathcal{H}^{+}\cap\Sigma_{0}}{\psi(\partial_{r}\psi)}.
\end{split}
\end{equation*}
Since $\psi$ solves the wave equation, it satisfies
$-\partial_{v}\partial_{r}\psi=\frac{1}{M}(\partial_{v}\psi)
+\frac{1}{2}\lapp\psi$
on $\mathcal{H}^{+}$. Therefore,
\begin{equation*}
\begin{split}
-\int_{\mathcal{H}^{+}}{\psi (\partial_{v}\partial_{r}\psi)}=&\frac{1}{M}\int_{\mathcal{H}^{+}}{\psi(\partial_{v}\psi)}+\frac{1}{2}\int_{\mathcal{H}^{+}}{\psi(\lapp\psi)}\\
=&\frac{1}{2M}\int_{\mathcal{H}^{+}\cap\Sigma_{\tau}}{\psi^{2}}-\frac{1}{2M}\int_{\mathcal{H}^{+}\cap\Sigma_{0}}{\psi^{2}}-\frac{1}{2}\int_{\mathcal{H}^{+}}{\left|\nabb\psi\right|^{2}}.
\end{split}
\end{equation*}
All the integrals on the right hand side can be estimated by $\int_{\Sigma_{0}}{J_{\mu}^{N}[\psi]n^{\mu}_{\Sigma_{0}}}$ using the second Hardy inequality  and \eqref{nenergy}. Furthermore, $\int_{\mathcal{H}^{+}\cap\Sigma}{\psi(\partial_{r}\psi)}\leq  \int_{\mathcal{H}^{+}\cap\Sigma}{\psi^{2}}+\int_{\mathcal{H}^{+}\cap\Sigma}{(\partial_{r}\psi)^{2}}.$
From the first and second Hardy inequality we have 
\begin{equation*}
\begin{split}
\int_{\mathcal{H}^{+}\cap\Sigma}{\psi^{2}}\leq C\int_{\Sigma_{0}}{J_{\mu}^{N}[\psi]n^{\mu}_{\Sigma_{0}}},
\end{split}
\end{equation*}
where $C$ depends  on $M$ and $\Sigma_{0}$. In addition, from the second Hardy inequality we have that for any positive number $\epsilon$ there exists a constant $C_{\epsilon}$ which depends on $M$, $\Sigma_{0}$ and $\epsilon$ such that 
\begin{equation*}
\begin{split}
\int_{\mathcal{H}^{+}\cap\Sigma}{(\partial_{r}\psi)^{2}}\leq C_{\epsilon}\int_{\Sigma}{J_{\mu}^{N}[\psi]n^{\mu}_{\Sigma}}+\epsilon\int_{\Sigma}{J_{\mu}^{L}[\partial_{r}\psi]n^{\mu}_{\Sigma}}.
\end{split}
\end{equation*}
\end{proof}

\begin{proposition}
For all solutions $\psi$ of the wave equation and any positive number $\epsilon$ we have 
\begin{equation*}
\left|\int_{\mathcal{H}^{+}}{(\lapp\psi)(\partial_{r}\psi)}\right|\leq  C_{\epsilon}\int_{\Sigma_{0}}{J^{N}_{\mu}[\psi]n^{\mu}_{\Sigma_{0}}}+\epsilon\int_{\Sigma_{0}\cup\Sigma_{\tau}}{J^{L}_{\mu}[\partial_{r}\psi]n^{\mu}_{\Sigma}},
\end{equation*}
where the positive constant $C_{\epsilon}$ depends  on $M$, $\Sigma_{0}$ and $\epsilon$.
\label{3comprop}
\end{proposition}
\begin{proof}
In view of the wave equation on $\mathcal{H}^{+}$ we have
\begin{equation*}
\int_{\mathcal{H}^{+}}{(\lapp\psi)(\partial_{r}\psi)}=-\frac{2}{M}\int_{\mathcal{H}^{+}}{(\partial_{v}\psi)(\partial_{r}\psi)}-2\int_{\mathcal{H}^{+}}{(\partial_{v}\partial_{r}\psi)(\partial_{r}\psi)}.
\end{equation*}
The first integral on the right hand side can be estimated using Proposition \ref{2comprop}. For the second integral we have
\begin{equation*}
\begin{split}
2\!\!\int_{\mathcal{H}^{+}}{(\partial_{v}\partial_{r}\psi)(\partial_{r}\psi)}=&\int_{\mathcal{H}^{+}}{\partial_{v}\left((\partial_{r}\psi)^{2}\right)}=\int_{\mathcal{H}^{+}\cap\Sigma_{\tau}}{(\partial_{r}\psi)^{2}}-\int_{\mathcal{H}^{+}\cap\Sigma_{0}}{(\partial_{r}\psi)^{2}}\\
\leq & C_{\epsilon}\!\int_{\Sigma_{0}}{J_{\mu}^{N}[\psi]n^{\mu}_{\Sigma_{0}}}+\epsilon\! \int_{\Sigma_{0}}{J_{\mu}^{L}[\partial_{r}\psi]n^{\mu}_{\Sigma_{0}}}+\epsilon\!\int_{\Sigma_{\tau}}{J^{L}_{\mu}[\partial_{r}\psi]n^{\mu}_{\Sigma_{\tau}}},
\end{split}
\end{equation*}
where, as above, $\epsilon$ is any positive number and $C_{\epsilon}$ depends on $M$, $\Sigma_{0}$ and $\epsilon$.

\end{proof}

\begin{center}
\large{\textbf{Estimate for} $\displaystyle\int_{\mathcal{R}}{H_{4}\left(\partial_{v}\partial_{r}\psi\right)\left(\partial_{v}\psi\right)}$}\\
\end{center}
For any $\epsilon >0$ we have
\begin{equation*}
\begin{split}
\left|\int_{\mathcal{A}}{H_{4}\left(\partial_{v}\partial_{r}\psi\right)\left(\partial_{v}\psi\right)}\right|&\leq \int_{\mathcal{A}}{\epsilon\left(\partial_{v}\partial_{r}\psi\right)^{2}+\int_{\Sigma_{0}}{\epsilon^{-1}H_{4}^{2}(\partial_{v}\psi)^{2}}}\leq \epsilon\int_{\mathcal{A}}{\left(\partial_{v}\partial_{r}\psi\right)^{2}+C_{\epsilon}\int_{\Sigma_{0}}{J^{N}_{\mu}[\psi]n^{\mu}_{\Sigma_{0}}}},
\end{split}
\end{equation*}
where the constant $C_{\epsilon}$ depends only on $M$,$\Sigma_{0}$ and  $\epsilon$.

\begin{center}
\large{\textbf{Estimate for} $\displaystyle\int_{\mathcal{R}}{H_{5}\left(\partial_{v}\partial_{r}\psi\right)\left(\partial_{r}\psi\right)}$}\\
\end{center}
As above, for any $\epsilon >0$
\begin{equation*}
\begin{split}
\left|\int_{\mathcal{A}}{H_{5}\left(\partial_{v}\partial_{r}\psi\right)\left(\partial_{r}\psi\right)}\right|&\leq\int_{\mathcal{A}}{\epsilon\left(\partial_{v}\partial_{r}\psi\right)^{2}}+\int_{\mathcal{A}}{\epsilon^{-1}H_{5}^{2} (\partial_{r}\psi)^{2}}\leq \int_{\mathcal{A}}{\epsilon\left(\partial_{v}\partial_{r}\psi\right)^{2}}+m\int_{\mathcal{A}}{ (\partial_{r}\psi)^{2}},
\end{split}
\end{equation*}
where $m=\operatorname{max}_{\mathcal{A}}\epsilon^{-1}H_{5}^{2}$. Then  from Proposition \ref{1comprop} (where $\epsilon$ is replaced with $\frac{\epsilon}{m}$) we obtain
\begin{equation}
\begin{split}
\left|\int_{\mathcal{A}}{H_{5}\left(\partial_{v}\partial_{r}\psi\right)\left(\partial_{r}\psi\right)}\right|\leq &\int_{\mathcal{A}}{\epsilon\left(\partial_{v}\partial_{r}\psi\right)^{2}}+\int_{\mathcal{A}}{\epsilon H_{1}(\partial_{v}\partial_{r}\psi)^{2}+\epsilon H_{2}(\partial_{r}\partial_{r}\psi)^{2}}\\ &+C_{\epsilon}\int_{\Sigma_{0}}{J_{\mu}^{N}[\psi]n^{\mu}_{\Sigma_{0}}}+C_{\epsilon}\int_{\Sigma_{0}}{J_{\mu}^{N}[T\psi]n^{\mu}_{\Sigma_{0}}},
\end{split}
\label{H6}
\end{equation}
where $C_{\epsilon}$ depends on $M$, $\Sigma_{0}$ and  $\epsilon$.

\begin{center}
\large{\textbf{Estimate for} $\displaystyle\int_{\mathcal{R}}{H_{6}\left(\partial_{r}\partial_{r}\psi\right)\left(\partial_{v}\psi\right)}$}\\
\end{center}
Since
$\operatorname{Div} \partial_{r}=\frac{2}{r},$
 Stokes' theorem yields
\begin{equation*}
\begin{split}
&\int_{\mathcal{R}}{H_{6}\left(\partial_{v}\psi\right)\left(\partial_{r}\partial_{r}\psi\right)}+\int_{\mathcal{R}}{\partial_{r}\left(H_{6}\partial_{v}\psi\right)\left(\partial_{r}\psi\right)}+\int_{\mathcal{R}}{H_{6}\left(\partial_{v}\psi\right)\left(\partial_{r}\psi\right)\frac{2}{r}}\\
=&\int_{\Sigma_{0}}{H_{6}\left(\partial_{v}\psi\right)\left(\partial_{r}\psi\right)\partial_{r}\cdot n_{\Sigma_{0}}}-\int_{\Sigma_{\tau}}{H_{6}\left(\partial_{v}\psi\right)\left(\partial_{r}\psi\right)\partial_{r}\cdot n_{\Sigma_{\tau}}}-\int_{\mathcal{H}^{+}}{H_{6}\left(\partial_{v}\psi\right)\left(\partial_{r}\psi\right)\partial_{r}\cdot n_{\mathcal{H}^{+}}}.
\end{split}
\end{equation*}
The boundary term over $\Sigma_{\tau}$ can be estimated by the $N$-flux whereas the boundary integral over $\mathcal{H}^{+}$ can be estimated using Proposition \ref{2comprop}. If $Q=\partial_{r}H_{6}+\frac{2}{r}H_{6}$, then it remains to estimate
\begin{equation*}
\begin{split}
\int_{\mathcal{R}}{Q\left(\partial_{v}\psi\right)\left(\partial_{r}\psi\right)}+\int_{\mathcal{R}}{H_{6}\left(\partial_{v}\partial_{r}\psi\right)\left(\partial_{r}\psi\right)}.
\end{split}
\end{equation*}
The first integral can be estimated using Cauchy-Schwarz,  \eqref{nk} and Proposition \ref{1comprop}. The second integral is estimated by \eqref{H6} (where $H_{5}$ is replaced with $H_{6}$).

\begin{center}
\large{\textbf{Estimate for} $\displaystyle\int_{\mathcal{R}}{H_{7}\left(\partial_{v}\partial_{r}\psi\right)\lapp\psi}$}\\
\end{center}

Stokes' theorem in region $\mathcal{R}$ gives us
\begin{equation*}
\begin{split}
&\int_{\mathcal{R}}{H_{7}(\partial_{v}\partial_{r}\psi)\lapp\psi}+\int_{\mathcal{R}}{H_{7}(\partial_{r}\psi)(\lapp\partial_{v}\psi)}=\!\!\int_{\Sigma_{0}}{H_{7}\left(\partial_{r}\psi\right)(\lapp\psi)\partial_{v}\cdot n_{\Sigma_{0}}}-\int_{\Sigma_{\tau}}{H_{7}\left(\partial_{r}\psi\right)\left(\lapp\psi\right)\partial_{v}\cdot n_{\Sigma_{\tau}}}\!.
\end{split}
\end{equation*}
For the boundary integrals we have the estimate
\begin{equation*}
\begin{split}
\int_{\mathcal{A}\cap\Sigma}{H_{7}(\partial_{r}\psi)\lapp\psi}=&-\int_{\mathcal{A}\cap\Sigma}{H_{7}\nabb\partial_{r}\psi\cdot\nabb\psi}\leq\epsilon\int_{\mathcal{A}\cap\Sigma}{J_{\mu}^{L}[\partial_{r}\psi]n^{\mu}_{\Sigma}}+C_{\epsilon} \int_{\Sigma}{J^{N}_{\mu}[\psi]n^{\mu}_{\Sigma}},
\end{split}
\end{equation*}
where $C_{\epsilon}$ depends on $M$, $\Sigma_{0}$ and $\epsilon$. As regards the second bulk integral, after applying Stokes' theorem on $\mathbb{S}^{2}(r)$ we obtain
\begin{equation*}
\begin{split}
\int_{\mathcal{A}}{H_{7}\nabb\partial_{r}\psi\cdot\nabb\partial_{v}\psi}\leq \epsilon\int_{\mathcal{A}}{\left|\nabb\partial_{r}\psi\right|^{2}}+C_{\epsilon}\int_{\mathcal{A}}\left|\nabb\partial_{v}\psi\right|^{2},
\end{split}
\end{equation*}
where $C_{\epsilon}$ depends on $M$, $\Sigma_{0}$ and $\epsilon$. Note that the second integral on the right hand side can be bounded by $\int_{\Sigma_{0}}{J_{\mu}^{N}[T\psi]n^{\mu}_{\Sigma_{0}}}$. Indeed, we commute with $T$ and use Proposition \ref{nkcor}. (Another way without having to commute with $T$ is by solving with respect to $\lapp\psi$ in the wave equation. As we shall see, this will be crucial in obtaining higher order estimates without losing derivatives.)

\begin{center}
\large{\textbf{Estimate for} $\displaystyle\int_{\mathcal{R}}{H_{8}\left(\partial_{r}\partial_{r}\psi\right)\lapp\psi}$}\\
\end{center}

We have
\begin{equation*}
\begin{split}
&\int_{\mathcal{R}}{H_{8}\left(\lapp\psi\right)\left(\partial_{r}\partial_{r}\psi\right)}+\int_{\mathcal{R}}{\partial_{r}\left(H_{8}\lapp\psi\right)\left(\partial_{r}\psi\right)}+\int_{\mathcal{R}}{H_{8}\left(\lapp\psi\right)\left(\partial_{r}\psi\right)\frac{2}{r}}\\
=&\int_{\Sigma_{0}}{H_{8}\left(\lapp\psi\right)\left(\partial_{r}\psi\right)\partial_{r}\cdot n_{\Sigma_{0}}}-\int_{\Sigma_{\tau}}{H_{8}\left(\lapp\psi\right)\left(\partial_{r}\psi\right)\partial_{r}\cdot n_{\Sigma_{\tau}}}-\int_{\mathcal{H}^{+}}{H_{8}\left(\lapp\psi\right)\left(\partial_{r}\psi\right)\partial_{r}\cdot n_{\mathcal{H}^{+}}}.
\end{split}
\end{equation*}
The integral over $\mathcal{H}^{+}$ is estimated in Proposition \ref{3comprop}. Furthermore, the Cauchy-Schwarz inequality implies
\begin{equation*}
\begin{split}
\int_{\Sigma_{\tau}\cap\mathcal{A}}{H_{8}\left(\lapp\psi\right)\left(\partial_{r}\psi\right)\partial_{r}\cdot n_{\Sigma_{\tau}}}=&-\int_{\Sigma_{\tau}\cap\mathcal{A}}{H_{8}\left(\nabb\psi\cdot\nabb\partial_{r}\psi\right)\partial_{r}\cdot n_{\Sigma_{\tau}}}\\
\leq & C_{\epsilon}\int_{\Sigma_{0}}{J_{\mu}^{N}[\psi]n^{\mu}_{\Sigma_{0}}}+\epsilon\int_{\Sigma_{0}}{J_{\mu}^{L}[\partial_{r}\psi]n^{\mu}_{\Sigma_{0}}},
\end{split}
\end{equation*}
where $\epsilon$ is any positive number and $C_{\epsilon}$ depends only on $M$, $\Sigma_{0}$ and $\epsilon$. For the remaining two spacetime integrals we have
\begin{equation*}
\begin{split}
&\int_{\mathcal{A}}{\partial_{r}\left(H_{8}\lapp\psi\right)\left(\partial_{r}\psi\right)}+\int_{\mathcal{A}}{\frac{2}{r}H_{8}\left(\lapp\psi\right)\left(\partial_{r}\psi\right)}=\int_{\mathcal{A}}{-H_{8}\left|\nabb\partial_{r}\psi\right|^{2}}+\int_{\mathcal{A}}{-\partial_{r}H_{8}\left(\nabb\psi\cdot\nabb\partial_{r}\psi\right)}.
\end{split}
\end{equation*}
The conditions \eqref{listL} assure us that 
$\left|H_{8}\right|<\frac{H_{3}}{10}$
and thus the first integral above can be estimated. For the second integral we apply the Cauchy-Schwarz inequality.

\begin{center}
\large{\textbf{Estimate for} $\displaystyle\int_{\mathcal{R}}{H_{9}\left(\partial_{v}\partial_{r}\psi\right)\left(\partial_{r}\partial_{r}\psi\right)}$}\\
\end{center}
In view of the wave equation we have
$\partial_{v}\partial_{r}\psi=-\frac{D}{2}\partial_{r}\partial_{r}\psi-\frac{1}{r}\partial_{v}\psi-\frac{R}{2}\partial_{r}\psi-\frac{1}{2}\lapp\psi.$
Therefore,
\begin{equation*}
\begin{split}
H_{9}\left(\partial_{v}\partial_{r}\psi\right)\left(\partial_{r}\partial_{r}\psi\right)=&-H_{9}\left[\frac{D}{2}\left(\partial_{r}\partial_{r}\psi\right)^{2}+\frac{R}{2}\left(\partial_{r}\psi\right)\left(\partial_{r}\partial_{r}\psi\right)+\frac{\left(\partial_{v}\psi\right)}{r}\left(\partial_{r}\partial_{r}\psi\right)+\frac{\left(\lapp\psi\right)}{2}\left(\partial_{r}\partial_{r}\psi\right)\right].
\end{split}
\end{equation*}
Note that in $\mathcal{A}$ we have
$H_{9}D\leq\frac{H_{2}}{10}$
and thus the first term on the right hand side poses no problem. Similarly, in $\mathcal{A}$ we have
\begin{equation*}
-H_{9}\frac{R}{2}\left(\partial_{r}\psi\right)\left(\partial_{r}\partial_{r}\psi\right)\leq \left(\partial_{r}\psi\right)^{2}+\left(H_{9}R\right)^{2}\left(\partial_{r}\partial_{r}\psi\right)^{2},\ (H_{9}R)^{2}\leq \frac{H_{2}}{10}.
\end{equation*}
According to what we have proved above, the integral 
$\int_{\mathcal{A}}{-\frac{H_{9}}{r}\left(\partial_{v}\psi\right)\left(\partial_{r}\partial_{r}\psi\right)}$
can be estimated (provided we replace $H_{6}$ with $-\frac{H_{9}}{r}$). Similarly, we have seen that  the integral
$\int_{\mathcal{A}}{-\frac{H_{9}}{2}(\partial_{r}\partial_{r}\psi)\lapp\psi}$
can be estimated provided we have $H_{9}\leq \frac{H_{3}}{10}$, which holds by the definition of $L$.

\begin{center}
\large{\textbf{Estimate for} $\displaystyle\int_{\mathcal{R}}{H_{10}\left(\partial_{r}\partial_{r}\psi\right)\left(\partial_{r}\psi\right)}$}\\
\end{center}
We have
\begin{equation*}
\begin{split}
& 2\int_{\mathcal{R}}{H_{10}\left(\partial_{r}\partial_{r}\psi\right)\left(\partial_{r}\psi\right)}+\int_{\mathcal{R}}{\left[\partial_{r}H_{10}+\frac{2}{r}H_{10}\right]\left(\partial_{r}\psi\right)^{2}}\\
=&\int_{\Sigma_{0}}{H_{10}\left(\partial_{r}\psi\right)^{2}\partial_{r}\cdot n_{\Sigma_{0}}}-\int_{\Sigma_{\tau}}{H_{10}\left(\partial_{r}\psi\right)^{2}\partial_{r}\cdot n_{\Sigma_{\tau}}}-\int_{\mathcal{H}^{+}}{H_{10}\left(\partial_{r}\psi\right)^{2}}.
\end{split}
\end{equation*}
The second spacetime integral can be estimated using Proposition \ref{1comprop} whereas the boundary integral over $\Sigma_{\tau}$ can be estimated using the $N$-flux. Finally we need to estimate the integral over $\mathcal{H}^{+}$. It turns out that this integral is the most problematic. Since
$R'=D''+\frac{2D'}{r}-\frac{2D}{r^{2}}$ we have $R'\left(M\right)=\frac{2}{M^{2}}$ and thus
\begin{equation}
H_{10}\left(M\right)=-L^{r}\left(M\right)R'\left(M\right)=-L^{r}\left(M\right)\frac{2}{M^{2}}>0.
\label{H.10}
\end{equation}
In order to estimate the integral along $\mathcal{H}^{+}$ we apply the Poincar\'{e} inequality
\begin{equation*}
\left|-\int_{\mathcal{H}^{+}}{\frac{H_{10}(M)}{2}\left(\partial_{r}\psi\right)^{2}}\right|\leq \int_{\mathcal{H}^{+}}{\frac{H_{10}(M)}{2}\frac{M^{2}}{l(l+1)}\left|\nabb\partial_{r}\psi\right|^{2}}.
\end{equation*}
Therefore, in view of  \eqref{horizonL1} it suffices to have
\begin{equation*}
\begin{split}
\frac{M^{2}}{2l(l+1)}H_{10}\left(M\right)&\leq -\frac{L^{r}\left(M\right)}{2}\overset{\eqref{H.10}}{\Leftrightarrow}
l\geq 1.
\end{split}
\end{equation*}
Note that for $l=1$ we need to use up all of the  good term over $\mathcal{H}^{+}$ which appears in identity \eqref{eiL}. This is not by accident. Indeed, in \cite{aretakis2} we will show that even for $l=1$ we have another conservation law on $\mathcal{H}^{+}$. This law is not implied only by the degeneracy of the horizon (as is the case for the conservation law for $l=0$ of Section \ref{sec:TheSphericallySymmetricCase}) but also uses an additional property of the metric tensor on $\mathcal{H}^{+}$. This law then implies that one needs to use up all of precisely this good term over $\mathcal{H}^{+}$. This is something we can do, since we have not used this term in order to estimate other integrals. That was possible via successive use of Hardy inequalities.

\subsection{$L^{2}$ Estimates for the Second Order Derivatives}
\label{sec:UniformBoundednessForTheSecondOrderDerivatives}

We can now prove the  statement (1) of Theorem \ref{theorem3} of Section \ref{sec:TheMainTheorems}. First note that $L\sim n_{\Sigma_{\tau}}$ in region $\mathcal{A}$. Therefore, the theorem  follows from the estimates \eqref{hypersurawayH} and \eqref{bulkawayH}, the estimates that we derived in Section \ref{sec:EstimatesOfTheBulkTerms} (by taking $\epsilon$ sufficiently small) and the energy identity \eqref{eiL}. 

Note that the right hand side of the estimate of statement (1) of Theorem \ref{theorem3}  is bounded by
$C\int_{\Sigma_{0}}{J_{\mu}^{N}[\psi]n^{\mu}_{\Sigma_{0}}}+C\int_{\Sigma_{0}}{J_{\mu}^{N}[N\psi]n^{\mu}_{\Sigma_{0}}}.$
This means that equivalently we can apply $N$ as multiplier but also as commutator for frequencies $l\geq 1$. As we shall see in the companion paper \cite{aretakis2}, such commutation does not yield the above estimate for $l=0$. One can also obtain similar spacetime estimates for spatially compact regions which include $\mathcal{A}$ by commuting with $T$ and $T^{2}$ and applying $X$ as a multiplier. Note that we need to commute with $T^{2}$ in view of the photon sphere. Note also that no commutation with the generators of the Lie algebra so(3) is required.

\section{Integrated Weighted Energy Decay}
\label{sec:IntegratedWeightedEnergyDecay}

We are now in position to eliminate the degeneracy of \eqref{nk}.
\begin{proposition}
There exists a uniform constant $C>0$ which depends on $M$ and $\Sigma_{0}$ such that for all solutions $\psi$ of the wave equation which are supported on the frequencies $l\geq 1$  we have
\begin{equation*}
\begin{split}
\int_{\mathcal{A}}{\left(\partial_{r}\psi\right)^{2}}\leq C\int_{\Sigma_{0}}{J_{\mu}^{N}[\psi]n^{\mu}_{\Sigma_{0}}}+C\int_{\Sigma_{0}}{J_{\mu}^{N}[T\psi]n^{\mu}_{\Sigma_{0}}}+C\int_{\Sigma_{0}\cap\mathcal{A}}{J_{\mu}^{N}[\partial_{r}\psi]n^{\mu}_{\Sigma_{0}}}.
\end{split}
\end{equation*}
\label{r1back}
\end{proposition}
\begin{proof}
Immediate from the third Hardy inequality, the statement (1) of Theorem \ref{theorem3} and Theorem \ref{xtheo}. 
\end{proof}
Note that the above proposition shows that in order to obtain this non-degenerate estimate in a neighbourhood of $\mathcal{H}^{+}$ one needs to commute the wave equation with $\partial_{r}$ and thus require higher regularity for $\psi$. This allows us to conclude that \textbf{trapping} takes place along $\mathcal{H}^{+}$. Remark that all frequencies exhibit trapped behaviour in contrast to the photon sphere where only high frequencies are trapped. This phenomenon is absent in the non-extreme case.

The statement (2) of Theorem \ref{theorem3} of Section \ref{sec:TheMainTheorems} follows from the above proposition and the results of Section \ref{sec:TheVectorFieldTextbfX}.

\section{Uniform Pointwise Boundedness}
\label{sec:UniformPointwiseBoundedness}
It remains to show that all solutions $\psi$ of the wave equation are uniformly bounded. 
\begin{proof}[Proof of Theorem \ref{t5} of Section \ref{sec:TheMainTheorems}]
We  decompose
$\psi=\psi_{0}+\psi_{\geq 1}$
and  prove that each projection is uniformly bounded by a norm that depends only on initial data. Indeed, by applying the following Sobolev inequality on the hypersurfaces $\Sigma_{\tau}$ we have
\begin{equation}
\begin{split}
\left\|\psi_{\geq 1}\right\|_{L^{\infty}\left(\Sigma_{\tau}\right)}&\leq 
C\left(\left\|\psi_{\geq 1}\right\|_{\overset{.}{H}^{1}\left(\Sigma_{\tau}\right)}+\left\|\psi_{\geq 1}\right\|_{\overset{.}{H}^{2}\left(\Sigma_{\tau}\right)}+\lim_{x\rightarrow i^{0}}{\left|\psi_{\geq 1}\right|}\right)
\label{Sobolev}
\end{split}
\end{equation}
where $C$ depends only on $\Sigma_{0}$. Note also that we use the Sobolev inequality that does not involve the $L^{2}$-norms of zeroth order terms. We observe that the vector field $\partial_{v}-\partial_{r}$ is timelike since $g\left(\partial_{v}-\partial_{r},\partial_{v}-\partial_{r}\right)=-D-2$ and, therefore, by an elliptic estimate there exists a uniform positive constant $C$ which depends on $M$ and $\Sigma_{0}$ such that 
\begin{equation*} 
\begin{split}
\left\|\psi_{\geq 1}\right\|^{2}_{\overset{.}{H}^{1}\left(\Sigma_{\tau}\right)}+\left\|\psi_{\geq 1}\right\|^{2}_{\overset{.}{H}^{2}\left(\Sigma_{\tau}\right)}\leq & C\int_{\Sigma_{\tau}}{J_{\mu}^{N}[\psi_{\geq 1}]n^{\mu}_{\Sigma_{\tau}}}+C\int_{\Sigma_{\tau}}{J_{\mu}^{N}[T\psi_{\geq 1}]n^{\mu}_{\Sigma_{\tau}}}+C\int_{\Sigma_{\tau}}{J_{\mu}^{N}[\partial_{r}\psi_{\geq 1}]n^{\mu}_{\Sigma_{\tau}}}.
\end{split}
\end{equation*}
It remains to derive a pointwise bound for $\psi_{0}$. However, from the 1-dimensional Sobolev inequality we have 
\begin{equation*}
\begin{split}
4\pi\psi_{0}^{2}(r_{0},\omega)\leq \frac{C}{r_{0}}\int_{\Sigma_{\tau}\cap\left\{r\geq r_{0}\right\}}{J_{\mu}^{N}[\psi_{0}]n^{\mu}_{\Sigma_{\tau}}},
\end{split}
\end{equation*}
where $C$ is a constant that depends only on $M$ and $\Sigma_{0}$. This completes the proof of the uniform boundedness of $\psi$.
\end{proof}

\section{Acknowledgements}
\label{sec:Acknowledgements}

I would like to thank Mihalis Dafermos for introducing to me  the problem and for his teaching and advice. I also thank Igor Rodnianski for sharing useful insights. I am supported by a Bodossaki Grant and a grant from the European Research Council.

\appendix

\section{Useful Reissner-Nordstr\"{o}m Computations}
\label{sec:UsefulReissnerNordstrOMComputations}

\subsection{The Wave Operator}
\label{sec:TheWaveOperatorINRN}

The wave operator in $(v,r,\theta,\phi)$ coordinates is
\begin{equation*}
\Box_{g}\psi=D\partial_{r}\partial_{r}\psi+2\partial_{v}\partial_{r}\psi+\frac{2}{r}\partial_{v}\psi+R\partial_{r}\psi+\lapp\psi,
\end{equation*}
where  $R=D'+\frac{2D}{r}$ and $D'=\frac{dD}{dr}$. In $(u,v)$ coordinates
\begin{equation*}
\Box_{g}\psi=-\frac{4}{Dr}\partial_{u}\partial_{v}(r\psi)-\frac{D'}{r}\psi+\lapp\psi.
\end{equation*}
\subsection{The Non-Negativity of the Energy-Momentum Tensor \textbf{T}}
\label{sec:TheHyperbolicityOfTheWaveEquation1}

We use the coordinate system $\left(v,r\right)$ and suppose that $V=\left(V^{v},V^{r},0,0\right), \ n=\left(n^{v},n^{r},0,0\right)$. The reader can easily verify that if  $\xi_{V}=\frac{1}{2\left(V^{v}\right)^{2}}\left(-g\left(V,V\right)\right)$ and similarly for $n$ then 
\begin{equation}
\begin{split}
J^{V}_{\mu}n^{\mu}=&\left[V^{v}n^{v}\left(1-\frac{D^{2}}{D^{2}+2\xi_{V}\cdot\xi_{n}}\right)\right]\left(\partial_{v}\psi\right)^{2}+\left[V^{v}n^{v}\left(\frac{\xi_{V}\cdot\xi_{n}}{2}\right)\right]\left(\partial_{r}\psi\right)^{2}+\left[-\frac{1}{2}g\left(V,n\right)\right]\left|\nabb\psi\right|^{2}\\
&+\left(\sqrt{\frac{D^{2}}{D^{2}+2\xi_{V}\cdot\xi_{n}}}\cdot\partial_{v}\psi+\sqrt{\frac{D^{2}+2\xi_{V}\cdot\xi_{n}}{4}}\cdot\partial_{r}\psi\right)^{2}.
\label{GENERALT}
\end{split}
\end{equation}

\section{Stokes' Theorem on Lorentzian Manifolds}
\label{sec:StokesTheoremOnLorentzianManifolds}

If $\mathcal{R}$ is a pseudo-Riemannian manifold and $P$ is a vector field on it then we have the divergence identity $\int_{\mathcal{R}}{\nabla_{\mu}P^{\mu}}=\int_{\partial\mathcal{R}}{P\cdot n_{\partial\mathcal{R}}}.$ Both integrals are taken with respect to the induced volume form. Note that $n_{\partial\mathcal{R}}$ is the unit normal to $\partial\mathcal{R}$ and its direction depends on the convention of the signature of the metric. For Lorentzian metrics with signature  $\left(-,+,+,+\right)$ the vector $n_{\partial\mathcal{R}}$ is the inward directed unit normal to $\partial\mathcal{R}$ in case $\partial\mathcal{R}$ is spacelike and the outward directed unit normal in case $\partial\mathcal{R}$ is timelike. If $\partial\mathcal{R}$ (or a piece of it) is null then    we take a past  (future) directed null normal to $\partial\mathcal{R}$ if it is future (past) boundary.  The following diagram is embedded in $\mathbb{R}^{1+1}$
 \begin{figure}[H]
	\centering
		\includegraphics[scale=0.08]{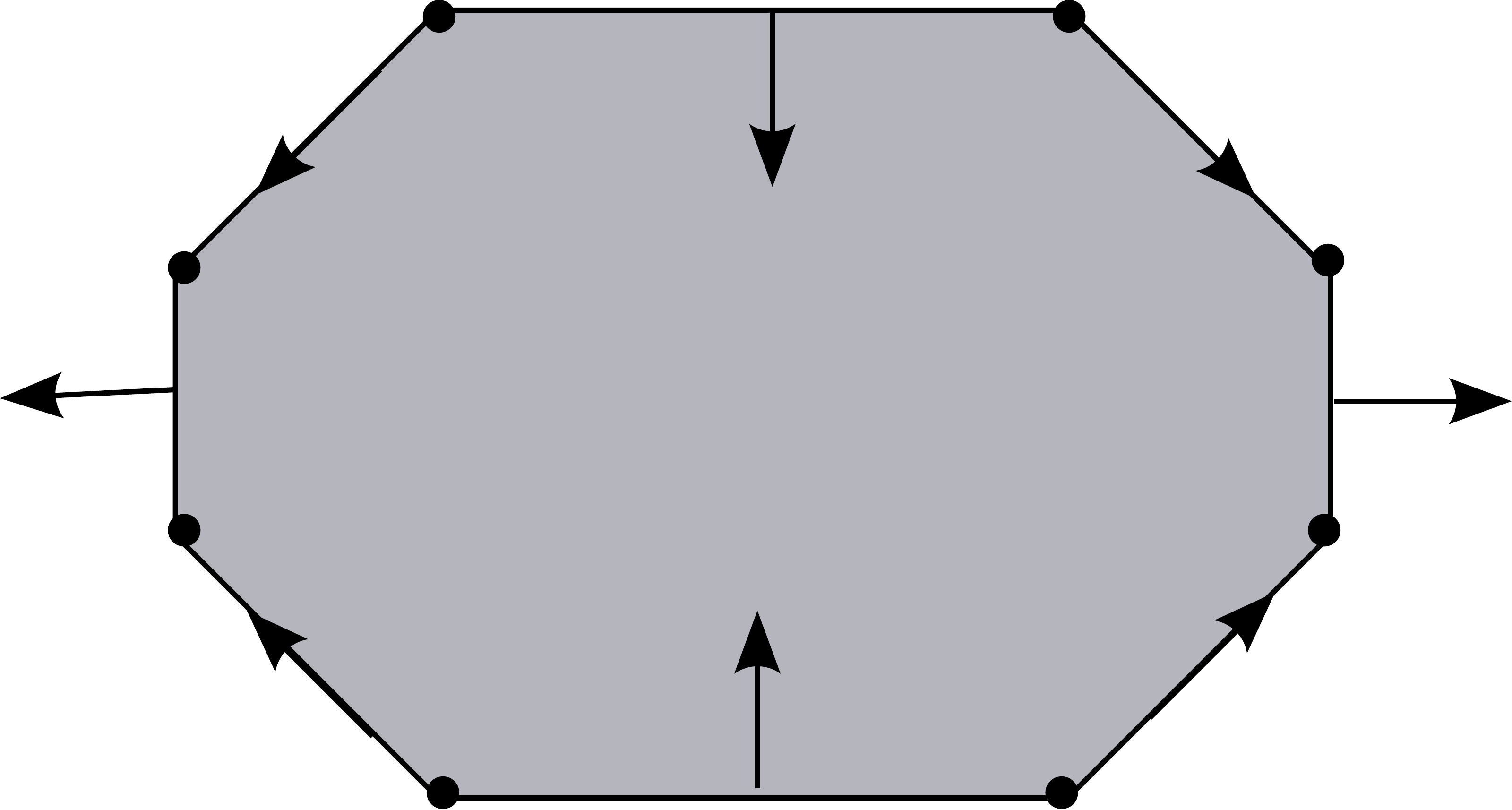}
	\label{fig:appendix}
\end{figure}

\end{document}